\documentclass[reqno,12pt]{article}
\usepackage{varioref} 
\usepackage{xr} 
\usepackage[colorlinks]{hyperref}
\usepackage{booktabs}
\usepackage{subcaption}
\usepackage{multirow}
\RequirePackage[OT1]{fontenc}
\RequirePackage{amsthm,amsmath}

\usepackage{epsfig,amssymb,euscript,amsfonts,hyperref}
\usepackage{natbib}
\usepackage{array}
\usepackage{rotating}
\usepackage{nccbbb}
\usepackage{verbatim}
\usepackage[lined,boxed,commentsnumbered, ruled,linesnumbered]{algorithm2e} 
\usepackage{graphicx}
\usepackage{float}
\usepackage{bm}
\usepackage{bbm}
\usepackage{dsfont}
\usepackage{MnSymbol}
\usepackage{enumitem}
\usepackage{color}
\usepackage{caption,setspace}
\usepackage{lipsum}
\usepackage{xcolor}
\usepackage{mathbbol}
\usepackage{marginnote}
\usepackage{forest}
\usepackage{authblk}

\numberwithin{equation}{section}

\newcolumntype{L}[1]{>{\raggedright\let\newline\\\arraybackslash\hspace{0pt}}m{#1}}
\newcolumntype{C}[1]{>{\centering\let\newline\\\arraybackslash\hspace{0pt}}m{#1}}
\newcolumntype{R}[1]{>{\raggedleft\let\newline\\\arraybackslash\hspace{0pt}}m{#1}}
\newcolumntype{H}{>{\setbox0=\hbox\bgroup}c<{\egroup}@{}}

\def\spacingset#1{\renewcommand{\baselinestretch}
{#1}\small\normalsize} \spacingset{1}

\usepackage{marginnote}

\newtheorem{theorem}{Theorem}[section]
\newtheorem{lemma}[theorem]{Lemma}
\newtheorem{corollary}[theorem]{Corollary}
\newtheorem{proposition}[theorem]{Proposition}

\theoremstyle{definition}
\newtheorem{example}{Example}[section]
\newtheorem{assumption}{Assumption}[section]
\newtheorem{remark}{Remark}[section]
\newtheorem{definition}{Definition}[section]

\hypersetup{citecolor=blue}
\hypersetup{linkcolor=red}

\definecolor{BUred}{rgb}{0.8, 0.0, 0.0}
\definecolor{BLUE}{rgb}{0.0, 0.28, 0.67}
\definecolor{magenta}{rgb}{1,0,.6}
\definecolor{blue0}{rgb}{0.95, 0.97, 1.0}
\definecolor{brown}{rgb}{0.75, 0.25, 0.0}

\usepackage{color}
	\definecolor{brown}{rgb}{0.75, 0.25, 0.0}
	\definecolor{mygreen}{rgb}{0.1, 0.5, 0.1}	
	\definecolor{myred}{rgb}{0.7, 0.1, 0.1}
	\definecolor{myblue}{rgb}{0.1, 0.1, 0.7}
	\definecolor{mygold}{rgb}{0.95, 0.85, 0.37}
	\definecolor{mygray}{rgb}{0.5,0.5,0.5}
	\definecolor{myorange}{rgb}{1,0.5,0}

\usepackage{mathtools}
\makeatletter
\def\widebreve{\mathpalette\wide@breve}
\def\wide@breve#1#2{\sbox\z@{$#1#2$}
     \mathop{\vbox{\m@th\ialign{##\crcr
\kern0.08em\brevefill#1{0.8\wd\z@}\crcr\noalign{\nointerlineskip}
                    $\hss#1#2\hss$\crcr}}}\limits}
\def\brevefill#1#2{$\m@th\sbox\tw@{$#1($}
  \hss\resizebox{#2}{\wd\tw@}{\rotatebox[origin=c]{90}{\upshape(}}\hss$}
\makeatletter

\DeclareMathOperator*{\argmin}{arg\,min}

\newcommand{\D}{\textnormal{d}}
\newcommand{\dd}{\,\textnormal{d}}

\newcommand{\E}{E} 
\newcommand{\Cov}{\mathrm{cov}}
\newcommand{\Var}{\mathrm{var}}

\newcommand{\Bias}{\mathop{\mathrm{Bias}}}

\newcommand{\SMSE}{\mathop{\mathrm{SMSE}}}
\newcommand{\MSE}{\mathop{\mathrm{MSE}}}

\newcommand{\MISE}{\mathop{\mathrm{MISE}}}
 
\newcommand{\Li}{ {\mathrm{Li}}}

\newcommand{\inD}{   \overset{ \textnormal{d}   }{\rightarrow} }

\newcommand{\inP}{  \overset{ \textnormal{pr}    }{\rightarrow} }

\newcommand{\iid}{\textsc{iid}} 
\newcommand{\simIID}{   \overset{ \iid   }{\sim} }

\newcommand{\Normal}{\textnormal{N}}

\def\s{{ \mathrm{\scriptscriptstyle s} }}

\def\bias{{ \mathrm{\scriptscriptstyle bias} }}
\def\se{{ \mathrm{\scriptscriptstyle se} }}

\def\T{{ \mathrm{\scriptscriptstyle T} }}

\def\AB{{ \textsc {Ba} }}
\def\AQ{{ \textsc {QS} }}
\def\CE{{ \mathrm{OCE} }}
\def\BCE{{ \mathrm{BCE} }}

\def\AR{{\textsc{ar}}}
\def\MA{{\textsc{ma}}}
\def\ARMA{{\textsc{arma}}}

\newcommand{\THK}{\textsc{TH}}
\newcommand{\PAK}{\textsc{Pa}}
\newcommand{\SPK}{\textsc{qu}}

\newcommand{\DF}{\textsc{DF}}
\newcommand{\BaK}{\textsc{Ba}}
\newcommand{\QSK}{\textsc{qs}}
\newcommand{\loose}{\textnormal{loose}}
\newcommand{\tight}{\textnormal{tight}}
\newcommand{\subtight}{\textnormal{sub-tight}}

\newcommand{\non}{\textnormal{non}}
\newcommand{\pw}{\textsc{pw}}
\newcommand{\HAC}{\textsc{hac}}

\usepackage{tikz}
\DeclareRobustCommand\solidLine  {\tikz[baseline=-0.6ex]\draw[thick] (0,0)--(0.5,0);}
\DeclareRobustCommand\dottedLine{\tikz[baseline=-0.6ex]\draw[thick,dotted] (0,0)--(0.54,0);}
\DeclareRobustCommand\dashedLine{\tikz[baseline=-0.6ex]\draw[thick,dashed] (0,0)--(0.54,0);}

\newcommand{\showName}{1}
\newcommand{\isMain}{1}
\newcommand{\paperTitle}{
Tight differencing in spectral density estimation with centrosymmetric kernels
}

\graphicspath{{./art/}}
\allowdisplaybreaks
\begin{document}

\if1\isMain{
\title{\bf \paperTitle}
}\else{
\title{\bf Supplementary note to ``\paperTitle''}
}\fi

\if1\showName{
	\author[1]{Yaxuan Wang \thanks{Email: 
				\href{ellisawang@link.cuhk.edu.hk}
				{\nolinkurl{ellisawang@link.cuhk.edu.hk}}}}
	\author[1]{Kin Wai Chan \thanks{Email: 
				\href{mailto:kinwaichan@cuhk.edu.hk}
				{\nolinkurl{kinwaichan@cuhk.edu.hk}}}}
	\affil[1]{Department of Statistics and Data Science, The Chinese University of Hong Kong.}
}\else{
	\author[]{}
}\fi	
\maketitle

\bigskip
\begin{abstract}
Mean-robust estimation of spectral density and long-run variance is crucial for many statistical inference procedures. 
However, existing methods often degrade when serially dependent data exhibit volatile, time-varying trends, 
or sudden jumps, particularly in small samples.
While differencing and kernel averaging are standard tools for achieving mean robustness and consistency, 
they are not inherently compatible. 
Combining them can compromise optimality. 
Specifically, tight differencing, an operation of taking small-lag differences to enhance local de-trending, 
introduces strong correlations that distort the high-order properties of kernel-averaged estimators.
To resolve this incompatibility, we introduce a novel class of centrosymmetric kernels 
explicitly designed to integrate with tight differencing. 
We demonstrate that the optimal tight difference sequence for serially dependent data differs 
from classical sequences designed for independent data. 
Notably, these proposed optimal sequences are data-independent 
and can be applied directly without pre-fitting. 
Finally, 
the proposed estimators are demonstrated to be 
useful across various statistical inference tasks, including tests for stationarity and white noise.
\end{abstract}

\spacingset{1.7}
\newpage

\section{Introduction} \label{sec:intro} 

Differencing and kernel averaging are standard tools
for removing trends and enhancing efficiency, respectively, in 
long-run variance estimation. 
The principal question addressed in this article is: 
Can differencing and kernel smoothing be used together
to ensure both estimation robustness and optimal convergence rate?
In general, the answer is both yes and no.
Prior work 
shows that it may work, but it may distort their individual properties and lose optimality unless the differencing lag is large enough; 
see, e.g., \cite{wu_zhao_2007,DetteWu2019,Chan2022}. 
Yet a shorter differencing lag provides better demeaning ability.  
It reveals an incompatibility issue.
This article addresses this problem. 

We consider the data $X_1, \ldots, X_n$ generated as 
$X_i = \mu_i + Z_i$,
where the means $\mu_1, \ldots, \mu_n$ are deterministic and 
the noise sequence $\{Z_i\}_{i\in\mathbb{Z}}$ is stationary with mean zero and 
autocovariance $\gamma_k=\E(Z_0 Z_k)$, $k \in \mathbb{Z}$. 
We write $\mu_i = \mu(i/n)$ where $\mu:[0,1]\rightarrow\mathbb{R}$ is an unknown function.
The quantities of interest are 
the spectral density and 
the long-run variance of $\bar{X}_n=\sum_{i=1}^{n}X_i/n$, i.e., 
\begin{align} \label{eq:lrv_spectral}
	f(\theta) = \frac{1}{2\pi}\sum_{k\in \mathbb{Z}}\gamma_k e^{2\pi\theta k\surd{-1}} 
     \quad \text{and} \quad
	v= \lim_{n \to \infty} n \Var(\bar{X}_n) 
\end{align}
for $\theta \in [0,0.5]$.
Under suitable conditions, 
$v = \sum_{k\in\mathbb{Z}}\gamma_k$,  
so estimating $v = 2\pi f(0)$ is a special case of 
estimating $f(\theta)$.
For clarity, we begin with studying $v$
and then extend it to $f(\theta)$
in \S\ref{sec:spectral}.

A consistent estimator of $v$   
that is robust to time-varying means and serial dependence 
is essential for inference in, 
e.g., nonparametric regression \citep{vogt2012, chen2012testing}, 
change point test \citep{ChenWangWu2021, dette2020, Tomasz2018}, 
construction of confidence bands \citep{wu_zhao_2007}, etc.
Recent work highlights that low-frequency components in the mean can distort estimation of the spectral density at frequency zero,
a problem known as low-frequency contamination
\citep{Casini2023Misspecified, CasiniDengPerron2025LowFrequency}.
Difference-based estimators mitigate this contamination by attenuating low-frequency components in the data.
We consider the class of difference-based kernel estimators proposed in \cite{Chan2022}: 
\begin{eqnarray}\label{eq:v_hat}
		\hat{v}
			= \sum_{|k|\leq \ell}K \left ( \frac{k}{\ell} \right )\hat{\gamma}_k,\quad \text{with}\ \ \hat{\gamma}_k=\frac{1}{n} \sum_{i=mh+|k|+1}^n D_i D_{i-|k|},
 	\end{eqnarray}
bandwidth $\ell$, kernel $K$,
differencing lag $h=\lambda \ell$, 
differencing order $m$, and difference statistics
\begin{eqnarray}\label{eq:D}
	D_i=\sum_{j=0}^{m}d_j (X_{i-jh}-\bar{X}_n), \quad i=mh+1,\ldots,n,
\end{eqnarray}
where the difference sequence $d_{0:m} = \{d_j\}_{j=0}^m$ satisfies 
$\sum_{j=0}^{m} d_j=0$ and $\sum_{j=0}^{m}d_j^2=1$ for $m>0$, and $d_0=1$ for $m=0$.  
The statistic $D_i$ is obtained by combining observations spaced $h$ apart with coefficients $\{d_j\}$, 
thereby reducing the effect of non-constant mean components. 
The construction of $D_i$ in \eqref{eq:D} is referred to as the differencing, 
whereas \eqref{eq:v_hat} performs kernel smoothing over the autocovariances of the differenced series.
We write 
$\hat{v} = \hat{v}(K) = \hat{v}(K, \ell, d_{0:m}, \lambda)$
to emphasize the inputs,
which play distinct roles.
The kernel $K$ determines how sample autocovariances are weighted.
Typically, $K\in\mathcal{K}$, where
\begin{align}\label{eqt:kernel}  
	\mathcal{K}=\{ G \in \mathcal{U}: G(0)=1, G(t)=0 \text{ for } |t| \geq 1, G(t)=G(-t) \text{ for all } t \in \mathbb{R} \},  
\end{align} 
where $\mathcal{U}$ is a set of functions $G:\mathbb{R}\rightarrow\mathbb{R}$ 
with finitely many discontinuities.
The bandwidth $\ell$ sets the truncation level of autocovariances, and thus governs the bias--variance trade-off. Larger $\ell$ reduces squared bias but increases variance.
The running index $k$ in (\ref{eq:v_hat}) represents autocovariance lag in the kernel sum.
Thus, (\ref{eq:v_hat}) defines a broad class and covers many existing estimators.

The major goal of this article is to study when this general estimator $\hat{v}$ remains rate-optimal under 
the shortest differencing lag $h$, 
later formalized as tight differencing,
and to propose a constructive remedy when the optimality fails.

We now review various differencing techniques and summarize how $h$ is selected in the literature.
Differencing is a commonly used approach in time series; see, e.g., 
\cite{Brown2007variance,TecuapetlaMunk2017,chan2020mean,to2024mean, LeungChan2022, bai2023difference}.
When $m=0$, \eqref{eq:D} reduces to global demeaning or centering, 
$D_i=X_i-\bar X_n$, and $\hat{v}$ 
reduces to the classical non-difference-based kernel estimators, which
are consistent for $v$ under a constant mean. 
The literature is vast; see, e.g.,  
\cite{Bartlett1950periodogram,Carlstein1986,NeweyWest1987,Hans1989,andrews1991heteroskedasticity,james2010,Politis2011,LiuFlegal2018,Zhu2020,vats2021lugsail,casini2024prewhitened,LeungChan2025,liuchan2026}. 
When the mean varies over time, global centering ($m=0$) 
may leave non-negligible mean components in $D_i$, which enter the sample autocovariances and bias the estimation of $v$; see an example in \S\ref{sec:ex_nondiff_bad}. 
When $m>0$, \eqref{eq:D} performs local centering
and thus is preferred under nonstationary means.
Finally, the lag $h$, or equivalently the lag-to-bandwidth ratio $\lambda=h/\ell$, 
sets the spacing between observations in each $D_i$, 
and thus determines how aggressively local centering is performed within each $D_i$,
Specifically, smaller $h$ provides stronger de-trending effects but interacts more closely with kernel averaging. 
For time series errors $\{Z_i\}$, the ratio $\lambda=h/\ell$ 
determines the behavior of $\hat{v}$, leading to the following cases: 

\begin{enumerate}[topsep=3pt]
	\item 
			If $\lambda<1$ (i.e., $h<\ell$), then $\hat{v}$ may be inconsistent for $v$, since under serial dependence,
            a large $h$ is needed to decorrelate the data used within each $D_i$; 
            see \S\ref{sec:lambda_smaller1} for  details. 
            When $\{Z_i\}$ are independent, $\hat{v}$ with $h=1$ is consistent for $\gamma_0$ \citep{Rice1984}, and thus for $v$ as $\gamma_k = 0$ for all $k\neq 0$. 
            Besides, taking larger $m$ improves precision \citep{Hall1990}.

	\item 
			If $\lambda=1$ (i.e., $h=\ell$), $\hat{v}$ is consistent for $v$. 
            In this case, many estimators use $m=1$ with the Bartlett kernel; see, e.g., \citet{wu2001, wu2004CP, wu_zhao_2007, dette2020, ChenWangWu2021}.
            It can be shown that 
            $\hat{v}\rightarrow v$ in $\mathcal{L}^2$ 
            at the optimal rate of $n^{1/3}$.
            For general $m$ and $K$, rate-optimality is not guaranteed: 
            for $m>1$, the optimal $\{d_0,\ldots,d_m\}$ remains open; and 
            for non-Bartlett $K$, optimality may even fail;
            see Figure \ref{fig:optimal}.

   \item 
			If $1<\lambda<2$ (i.e., $\ell<h<2\ell$), $\hat v$ is consistent for $v$. 
		   Rate-optimality is attainable since 
		   the kernel $K$ preserves the same high-order properties after differencing.  
		   But, as in the case $h=\ell$, 
		   the optimal $\{d_0,\ldots,d_m\}$ are unsolved, thus optimality is not guaranteed. 

	\item 
			If $\lambda\geq2$ (i.e., $h\geq 2\ell$), $\hat{v}$ is rate-optimal with any $K\in\mathcal{K}$; see \citet{Chan2022},
			but the demeaning effect is poorer than those with $h=\ell$ as $D_i$ is formed by well-separated data.
\end{enumerate} 

Hence, from the robustness perspective,
the ideal $\hat{v}$ should use 
the minimal ratio, i.e., $\lambda = h/\ell=1$, 
where rate optimality is attainable in principle.
However, as shown in \S\ref{sec:motivation}, 
reducing from the common choice $\lambda=2$ to $\lambda=1$ is highly non-trivial, 
since differencing at $\lambda=1$ may ruin the high-order properties of kernels, 
causing
the estimator $\hat{v}$ to lose its optimal convergence rate even under a constant mean.
This arises because, when $\lambda=1$, observations are too densely used in a single $D_i$ in (\ref{eq:D}), 
leading to a significant and non-trivial overlap with observations in nearby
$D_{i'}$, where $i'\approx i$. 
Consequently, the traditional kernel averaging in (\ref{eq:v_hat}) fails to achieve the desired properties.
This article reveals and resolves this incompatibility issue
between differencing (\ref{eq:D}) with $\lambda=1$ and kernel averaging (\ref{eq:v_hat}).

We make three major contributions: 
(i) a new class of centrosymmetric kernels 
tailored for differencing with the minimal $h/\ell$ ratio; 
(ii) a proof that at $\lambda=1$, the optimal differencing procedure differs from the classical results; and   
(iii) theoretical analysis showing that our proposal achieves better robustness and smaller mean squared error. 
Henceforth, call $D_i$ in \eqref{eq:D} tight when $\lambda=1$ and loose when $\lambda=2$.
Denote  the tight, loose and non-difference based estimators as
\[
	\hat{v}_{\tight}(K) = \hat{v}(K, \ell, d_{0:m}, 1), \qquad
	\hat{v}_{\loose}(K) = \hat{v}(K, \ell, d_{0:m}, 2), \qquad
	\hat{v}_{\non}(K) = \hat{v}(K, \ell, d_0, 0),
\] 
We may drop the input $K$ in the above notation
when no confusion is possible.

\section{Motivation} \label{sec:motivation}
\subsection{Origin and boundary characteristic exponents} \label{sec:CE_BCE}
Statistical properties, e.g., the mean squared error, of $\hat{v}(K)$ are closely 
related to that of 
\[ 
	\hat{v}^Z(K)
		= \sum_{|k|\leq \ell}K \left ( \frac{k}{\ell} \right )\hat{\gamma}_k^Z,
		\quad \text{where}\ \ 
	\hat{\gamma}_k^Z
		= \frac{1}{n} \sum_{i=|k|+1}^n Z_i Z_{i-|k|}. 
\]
When $m=0$, $\hat{v}(K) \approx \hat{v}^Z(K)$ in the $\mathcal{L}^2$-sense; 
see, e.g., \cite{andrews1991heteroskedasticity},  
\citet{Alexopoulos2004}, and \citet{ChanYau2017}.  
When $m>0$, 
differencing distorts 
the kernel $K$ to an effective kernel $K_d$ 
in the sense that $\hat{v}(K) \approx \hat{v}^Z(K_d)$ in $\mathcal{L}^2$, where 
\begin{align}\label{eqt:Kd}
	K_{d}(t)
		= \sum_{s=\left \lceil -(1+t)/\lambda \right \rceil}^{\left \lfloor (1-t)/\lambda \right \rfloor} \delta_{s}K(t+\lambda s) 
        \quad \text{with}  \quad
        \delta_s=\left\{ \begin{array}{ll}
            \displaystyle{\sum_{j=|s|}^md_jd_{j-|s|}} & \text{for $|s|\leq m$};\\
            0 & \text{otherwise}; \\
            \end{array} \right.
\end{align}
see 
\citet{wu_zhao_2007, DetteWu2019, Chan2022}
for similar approaches.
In this section, we show that $K_d$ may not be smooth enough and leads to  
suboptimal convergence rate of 
$\hat{v}^Z(K_d)\rightarrow v$ in $\mathcal{L}^2$.
To present the reason behind, we need to formalize the near-origin and near-boundary behavior of $K(t)$
in terms of the flatness at $t=0$ and $t=1$:  

\begin{definition}[Kernel order] \label{def:CE_BCE}
For $r\in\mathbb{N}$, 
let $B(r) = \lim_{t\downarrow 0} \{K(t)-K(0)\}/t^r$ and $B'(r) = \lim_{t\downarrow 0} \{K(1)-K(1-t)\}/t^r$
when the limits exist. 
The origin characteristic exponent and the boundary characteristic exponent of a kernel function $K\in\mathcal{K}$ are 
\begin{align*}
	\CE(K)=\sup\left \{r \in \mathbb{N}   : |B(r)| < \infty \right \} 
	\quad \text{and} \quad
	\BCE(K)=\sup\left \{r \in \mathbb{N}   :  |B'(r)| < \infty \right \}.
\end{align*}
The corresponding origin flatness coefficient and boundary flatness coefficient of $K$ are
\begin{align}\label{eqt:BBprime}
	B = B(q) = \lim_{t \downarrow 0} \frac{ K(t)-K(0)}{t^q} 
	\quad \text{and} \quad 
	B^{\prime} = B^{\prime}(q') = \lim_{t \downarrow 0}\frac{ K(1)-K(1-t)}{t^{q^{\prime}}}, 
\end{align}
where $q=\CE(K)$ and $q' = \BCE(K)$.
\end{definition}

\begin{figure}[t] 
\captionsetup{font=small}
	\begin{center}
	\includegraphics[width=0.8\linewidth]{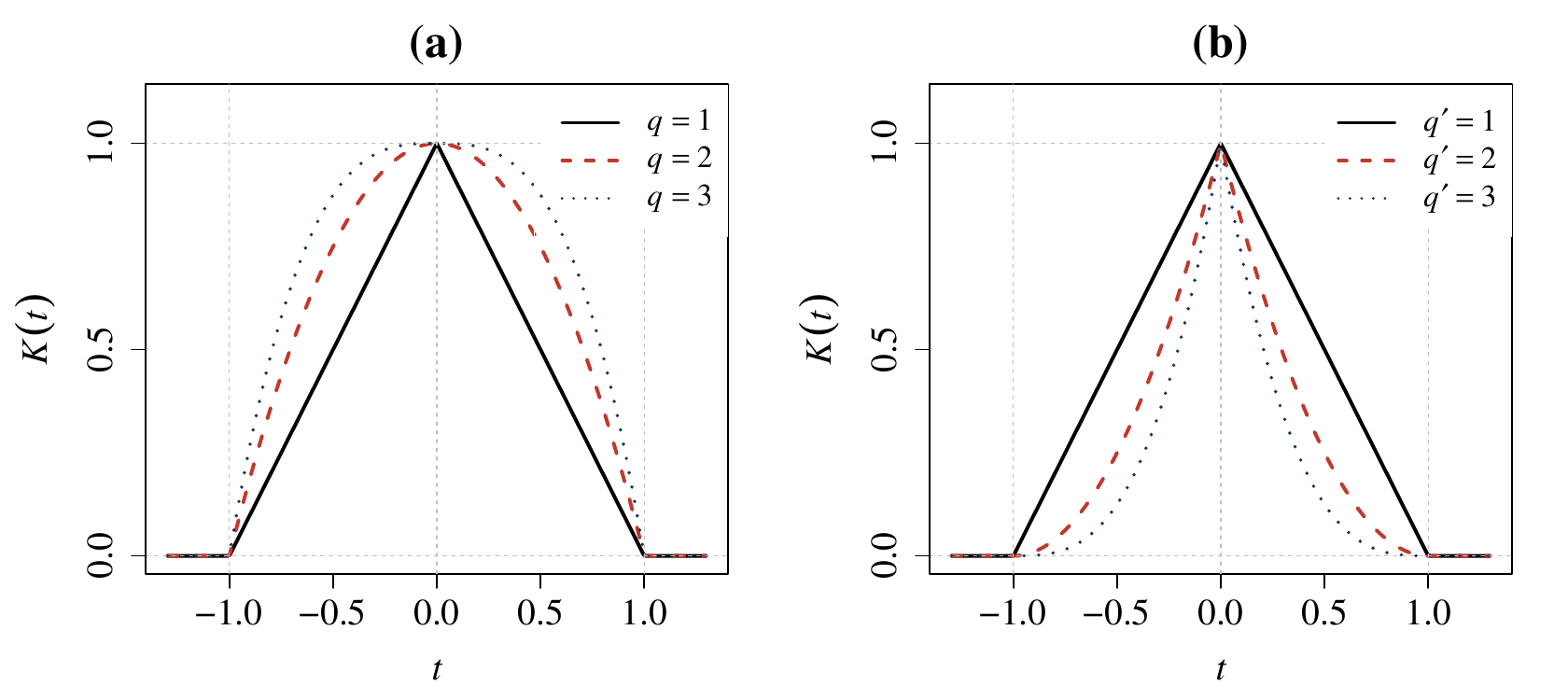} \vspace{-0.3cm}
	\caption{(a) The kernels $K(t) = (1-|t|^q)^+$ for 
	$q=1$ (solid line \solidLine), $q=2$ (dashed line {\color{blue}\dashedLine}), and $q=3$ (dotted line {\color{red}\dottedLine}),
    where $(x)^{+} = \max(x,0)$.
	Here $\CE(K)=q$ and $\BCE(K) = 1$.
	(b) The kernels $K(t) = \{(1-|t|)^{q^{\prime}}\}^+$ for 
	$q'=1$ (solid line \solidLine), $q'=2$ (dashed line {\color{blue}\dashedLine}), and $q'=3$ (dotted line {\color{red}\dottedLine}).
    Here $\BCE(K)=q^{\prime}$ and $\CE(K) = 1$.
	} 
	\label{fig:OCE_BCE}
	\end{center}  \vspace{-0.3cm}
\end{figure} 

In Definition \ref{def:CE_BCE}, 
larger $q=\CE(K)$ and $q^{\prime}=\BCE(K)$ mean $K(t)$ is flatter 
near the origin $t=0$ and near the boundaries $t=\pm1$; 
see Figure \ref{fig:OCE_BCE}. 
So, we call $\CE(K)$ and $\BCE(K)$
the near-origin and the near-boundary kernel orders, respectively. 
Examples of some commonly used kernels and their corresponding values of
$\CE(K)$ and $\BCE(K)$ are summarized in Table \ref{table:kernel_oce_bce} of the supplement. 
When $m=0$,
under conditions, including
$\sum_{k\in\mathbb{Z}}|k|^q|\gamma_k| <\infty$
and $\mu_1=\cdots=\mu_n$, 
the optimal mean-squared error of $\hat{v}(K)$ is 
$O\{n^{-2q/(1+2q)}\}$ when 
$\CE(K)=q$ \citep{Parzen1957}.
When $m>0$, 
the theoretical properties of $\hat{v}(K)$
is controlled by the origin flatness of $K_d$ instead of $K$. 
However, $\CE(K_d)$ need not equal $\CE(K)$; see Proposition \ref{prop:CE_Kdiff} below.

\begin{proposition}\label{prop:CE_Kdiff}
Let $K\in\mathcal{K}$.
Then the effective kernel $K_d$ defined in (\ref{eqt:Kd}) satisfies 
\begin{align}\label{CEKd}
	\CE(K_d)= \left\{
	\begin{array}{ll}
	\min \{ \CE(K), \BCE(K)\}, & \text{if $\lambda=1$}; \\
	\CE(K) , & \text{if $\lambda\in(1, \infty)$}.
	\end{array}
	\right.
\end{align}
\end{proposition}

The proof of Proposition \ref{prop:CE_Kdiff} is given in \S\ref{pf:prop_CE_Kdiff}.  
Property (\ref{CEKd})
reveals that the incompatibility issue arises from a mismatch between the flatness at the origin and the flatness at the boundaries.
More precisely, 
tight differencing ($\lambda=1$) reduces the near-origin order of the effective kernel $K_d$
when $\BCE(K)<\CE(K)$ due to insufficient boundary flatness.  
It thereby reduces the fastest attainable convergence rate of $\hat v$.
However, this issue does not appear when $\lambda>1$. 

To preserve rate optimality, we require $\CE(K_d)=\CE(K)$, that is, 
$K_d$ cannot be flatter than $K$ at the origin,
which, by Proposition \ref{prop:CE_Kdiff}, is equivalent to requiring $\CE(K)\leq \BCE(K)$. However, many commonly used kernels do not satisfy this requirement, which motivates the development of the centrosymmetrization procedure in \S\ref{sec:symmetrization} to restore compatibility when $\lambda=1$.
We remark that although $\lambda>1$ does not distort the $\CE$, it sacrifices mean robustness, 
while $\lambda<1$ may lead to inconsistency, as $K(0)=1$ required in (\ref{eqt:kernel}) may not hold; 
see \S\ref{sec:lambda_smaller1} of the supplement for more details. 
Also, the kernel class $\mathcal{K}$ considered in (\ref{eqt:kernel}) is truncated, i.e., 
$K(t)=0$ for $|t|\geq 1$. 
This ensures that $K(0)=1$ for any $\lambda\geq 1$. 
Section \ref{sec:QS_kernel} proves this claim and demonstrates that the quadratic spectral kernel, 
which is not in the class $\mathcal{K}$, fails to satisfy $K(0)=1$.
We now visualize Proposition \ref{prop:CE_Kdiff} in Example \ref{eg:Kd_and_K} below.

\begin{example}[Flatness of $K_d$]
\label{eg:Kd_and_K}
We consider the (i) quadratic kernel $K_{\SPK}(t)=(1-t^2)^+$, where $(x)^{+} = \max(x,0)$, 
(ii) Tukey--Hanning kernel $K_{\THK}(t)=\{1+\cos(\pi t)\}\mathbb{1}(|t|\le1)/2$, 
and (iii) Parzen kernel 
$K_{\PAK}(t)=(1-6t^2+6|t|^3)\mathbb{1}(|t|\le1/2)+2(1-|t|)^3\mathbb{1}(1/2<|t|\le1)$.
These yield $(q,q^{\prime})=(2,1),(2,2),(2,3)$, respectively \citep{Parzen1957,andrews1991heteroskedasticity,Gallant1987}. 
Under tight differencing, $K_d$ is less flat at $t=0$ when $q>q^{\prime}$; see Figure~\ref{fig:Kdiff_ex}.

\begin{figure}[t]
\captionsetup{font=small}
	\begin{center} 
	\includegraphics[width=\linewidth]{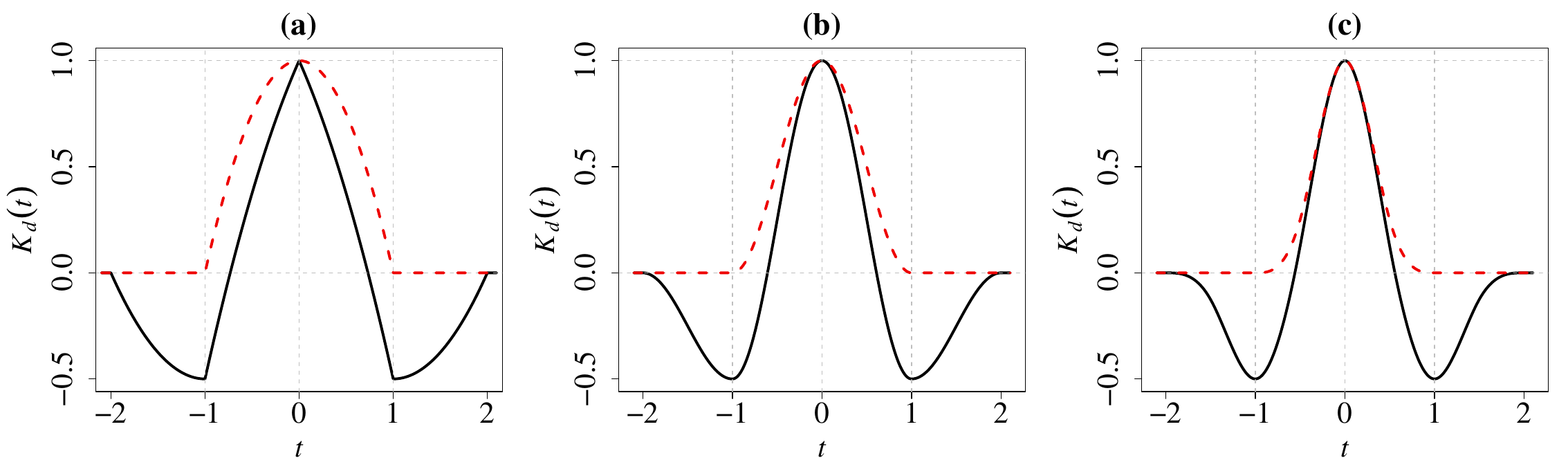} \vspace{-0.3cm}
	\caption{Plots (a)--(c) show $K_d$ (black solid line \solidLine) with $\lambda=1$ 
        when $K=K_{\SPK}, K_{\THK}, K_{\PAK}$ 
        (red dashed line {\color{red}\dashedLine}), respectively. 
	}
	\label{fig:Kdiff_ex}
	\end{center} \vspace{-0.2cm}
\end{figure}
\end{example}

Based on Proposition \ref{prop:CE_Kdiff},
we refine $\mathcal{K}$ to
$$
	\mathcal{K}_{q,q^{\prime}} = \left\{ K \in \mathcal{K}: q = \CE(K),\;\text{and}\; q^{\prime}=\BCE(K) \right\}.
$$
We propose to use 
$\hat{v}_{\tight}(K)= \hat{v}(K,\ell,d_{0:m}, 1) $ for 
$K\in\mathcal{K}_{q,q^{\prime}}$ with $q\leq q^{\prime}$,
which indicates that the kernel is at least as flat near the boundary $|t|=1$ as near the origin $t=0$,
to ensure $\CE(K_d)=\CE(K)$.
But even if effective kernels have the same near-origin order, their estimators may still differ in bias if flatness magnitudes of $K_d$ are different.

We will show in \S\ref{sec:theory} that the leading bias term of $\hat{v}$ is related to the constants $B$ and $B'$ defined in (\ref{eqt:BBprime}).
The following proposition states the properties of $K_d$ when $K\in\mathcal{K}_{q,q^{\prime}}$ with $q\leq q^{\prime}$.

\begin{proposition}\label{prop:flatness_ock}
Let $K\in\mathcal{K}_{q,q^{\prime}}$ 
for some $q,q^{\prime}\in\mathbb{N}$ and $q\leq q^{\prime}$.  
Define $B$ and $B'$ as in \eqref{eqt:BBprime}.
Then 
(i) $\CE(K_d) = \CE(K) = q$ for $\lambda \geq 1$; 
(ii) the origin flatness coefficient of $K_d$ is 
\[
	B_d 
	= \lim_{t \downarrow 0} \frac{ K_d(t)-K_d(0)}{t^q}
	= B - \delta_1 B^{\prime}\mathbb{1}(\lambda=1) \mathbb{1}(q=q^{\prime}). 
\]
\end{proposition}   

Proposition \ref{prop:flatness_ock} implies that using $K\in\mathcal{K}_{q,q^{\prime}}$ with $q\leq q^{\prime}$ 
resolves the compatibility issue 
because $K_d$ inherits the origin characteristic exponent from $K$, i.e., 
$\CE(K_d) = \CE(K)$, even if $\lambda=1$. 
See \S\ref{pf:prop_flatness_ock} for the proof of Proposition \ref{prop:flatness_ock}. 
Also, when $\lambda=1$, 
the origin flatness coefficient $B_d$ 
depends on $\delta_1 = \sum_{j=1}^m d_jd_{j-1}$,
which guides the optimal choice of $d_0,\ldots, d_m$ in \S\ref{sec:Op_tightD}. 

\subsection{Centrosymmetrization} \label{sec:symmetrization} 
Many commonly used kernels $K$ 
violate the compatibility constraint $\BCE(K) \geq \CE(K)$
because they are less flat near the boundaries than at the origin. As a result, they yield
$K_d$ with a reduced effective kernel order according to (\ref{CEKd}); see Table \ref{table:kernel_oce_bce} for examples. 
Next,
we propose a procedure called centrosymmetrization to transform any $K$ as to enforce $\BCE(K) = \CE(K)$. 

\begin{definition}\label{def:Ksym}
Let $\mathcal{K}^{\circ}$ be the class of centrosymmetric kernels that admit the form
\begin{align}\label{eqt:K_transformed}
	K(t) = \left\{ \begin{array}{ll} 
						K_0(t) & \text{if $|t|< t_0$};\\
						H(t) & \text{if $t_0 \leq |t| \leq 1-t_0$};\\
						1-K_0(1-|t|) & \text{if $1-t_0 < |t|\leq 1$}, 
					\end{array}\right. 
\end{align}
and $K(t)=0$ elsewhere, for some $t_0\in(0, 0.5]$ and $K_0 , H \in \mathcal{K}$. 
Define the class of $q$th order centrosymmetric kernels as $\mathcal{K}_q^{\circ} = \mathcal{K}^{\circ} \cap \mathcal{K}_{q,q}$.
\end{definition} 

Definition \ref{def:Ksym} indicates that,
the function $K$ in (\ref{eqt:K_transformed}) over the supports $[0,t_0)$ and $(1-t_0, 1]$ are centrosymmetric about the point $(0.5,0.5)$.
Via (\ref{eqt:K_transformed}), 
one may centrosymmetrize any $K_0 \in \mathcal{K}_{q,q^{\prime}}$ with $q> q^{\prime}$
to $K\in \mathcal{K}_{q,q}$, 
which ensures that tight differencing does not reduce the kernel order.
Consequently, the tight difference-based estimator can achieve the best possible rate of mean squared error in \S\ref{sec:theory}.

Since the interior part $H(t)$ does not affect the properties of $K_d$, 
we can set it flexibly.
In particular, we can set $H$ so that the resulting kernel $K$ in (\ref{eqt:K_transformed}) 
satisfies some desired properties, such as continuity or differentiability at $t_0$. 
Three types of $K_d$ are discussed below.

\begin{figure}[t] 
	\begin{center}
		\includegraphics[width=\linewidth]{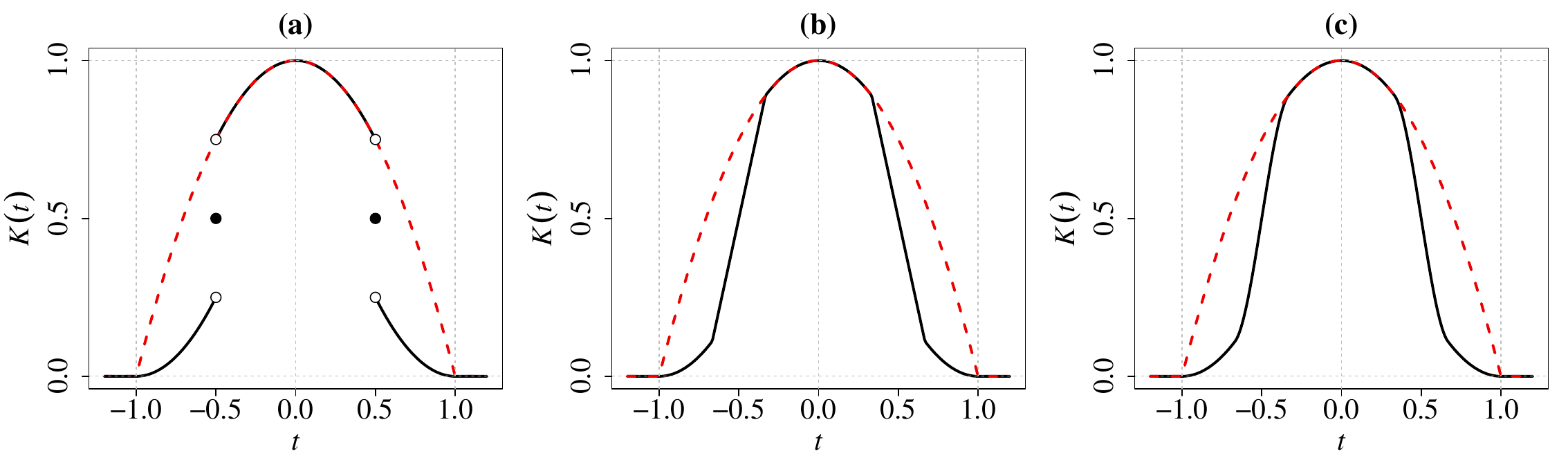} \vspace{-0.3cm}
		\captionsetup{font=small} 
        \vspace{-0.1cm}
		\caption{
		The centrosymmetrized kernels $K$ 
        (black solid line \solidLine) defined in Examples \ref{ex:Kd}--\ref{ex:Ks}
		are shown in (a)--(c), respectively. 
        The red dashed lines (red dashed line {\color{red}\dashedLine}) represent the input kernel $K_0(t) = (1-t^2)^+$. }
		\label{fig:Ksym_ex}
	\end{center}  \vspace{-0.2cm}
\end{figure}

\begin{example}[Discontinuous $K$] \label{ex:Kd}
Set $H(t) = \mathbb{1}(|t|=1/2)/2$ and $t_0=1/2$. 
Although the formula is simple, $K$ is discontinuous. 
Denote such $K$ as $K_{\text{d}}$. 
If $K_0(t) = (1-|t|^q)^+$ with $q>1$, constructing $K(t)$ via \eqref{eqt:K_transformed} using the above $H(t)$ yields $K(t) \in\mathcal K_q^\circ$ with $B=B'=-1$.
\end{example} 

\begin{example}[Continuous $K$] \label{ex:Kc}
Set 
$H(t) = (a_1|t|+a_0) \mathbb{1}\left( 1/3 \leq |t| \leq 2/3\right)$ and $t_0=1/3$, 
where  $a_1=3\left\{1-2K_0 ( 1/3  )\right\} $ 
and $a_0=3K_0 ( 1/3 )-1.$
Then $K$ is continuous.
If $K_0(t) = (1-|t|^q)^+$ with $q>1$, 
constructing $K(t)$ via \eqref{eqt:K_transformed} using the above $H(t)$ yields $K(t) \in \mathcal{K}_q^\circ$ with 
$B=B^\prime =-1$ and 
$a_1=2/3^{q-1}-3$ and $a_0=2-1/3^{q-1}$. 
Denote such $K$ as $K_{\text{c}}$.
\end{example} 

\begin{example}[Smooth $K$] \label{ex:Ks}
Set 
$H(t)= \sum_{j=0}^3 \alpha_j |t|^j  \mathbb{1} (1/3 \leq |t| \leq 2/3 )$ and $t_0=1/3$, 
where  
$\alpha_3 = 108K_0\left(1/3 \right) + 18 K^{\prime}_0\left(1/3 \right) -54$,
$\alpha_2 = -162K_0\left(1/3 \right) -27 K^{\prime}_0\left(1/3 \right) +81$,
$\alpha_1 =72K_0\left(1/3 \right) + 13 K^{\prime}_0\left(1/3 \right) -36$,
and $\alpha_0 = -9K_0\left(1/3 \right) -2 K^{\prime}_0\left(1/3 \right) +5$.
Here $K^{\prime}_0(t) = \dd K_0(t)/\dd t$. 
Then $K$ is differentiable. 
If $K_0(t) = (1-|t|^q)^+$ with $q>1$, 
constructing $K(t)$ via \eqref{eqt:K_transformed} using the above $H(t)$ yields $K(t) \in\mathcal K_q^\circ$ with $B=B'=-1$ and 
$\alpha_3 = -2(q+2)/3^{q-3}+54$, $\alpha_2 = (q+2)/3^{q-4}-81$, $\alpha_1 = -(13q+24)/3^{q-1}+36$, and $\alpha_0 =(2q+3)/3^{q-1}-4$.
Denote such $K$ as $K_{\text{s}}$.
See \S\ref{sec:general_t0} for the results with a general $t_0$. 
\end{example} 

Figure \ref{fig:Ksym_ex} visualizes $K_{\text{d}}$, $K_{\text{c}}$, and $K_{\text{s}}$.
By default, we use $K_{\text{s}}$ for $\hat{v}_{\tight}$.
We conclude with an example illustrating the incompatibility.

\begin{example}[Rate-optimality] \label{ex:rateOptimality} 
Let $K_0(t)=(1-t^2)^+$.
We compare  
$\hat{v}_{\tight}(K_0)$, $\hat{v}_{\tight}(K_\text{d})$, $\hat{v}_{\tight}(K_\text{c})$, $\hat{v}_{\tight}(K_\text{s})$, 
$\hat{v}_{\loose}(K_0)$, and
$\hat{v}_{\non}(K_0)$, 
with $m=1$ for difference-based estimators.  
All estimators use $\ell = \lfloor 2 n^{1/5}\rfloor$.
We test if these estimators attain the optimal convergence rate.

Let the noise follow an order-1 autoregressive ($\AR(1)$) model: 
$Z_i=0.5Z_{i-1}+\varepsilon_i$, with $\varepsilon_i\sim\Normal(0,1)$ independently.  
We consider $\mu(t)\equiv0$ and $\mu(t)=5\{t+\mathbb{1}(t>1/2)\}$. 
Figure \ref{fig:optimal} shows the results. 
With constant mean, $\hat v_{\tight}(K_0)$ fails to attain the optimal rate as tight differencing 
conflicts with the non-centrosymmetric $K_0$. 
Using centrosymmetric kernels can restore optimality. 
The similar performance among $K_{\text{d}}$, $K_{\text{c}}$, and $K_{\text{s}}$ 
suggests that the centrosymmetrization choice matters little.
Loose differencing
degrades quickly when $\mu(t)$ is non-constant.
In that case, $\hat v_{\non}(K_0)$ is inconsistent, 
whereas $\hat v_{\tight}(K_{\text{d}})$, $\hat v_{\tight}(K_{\text{c}})$, and $\hat v_{\tight}(K_{\text{s}})$ remain reliable.
\end{example}

\begin{figure} [t] 
\captionsetup{font=small}
	\begin{center}  
		\includegraphics[width=0.8\linewidth]{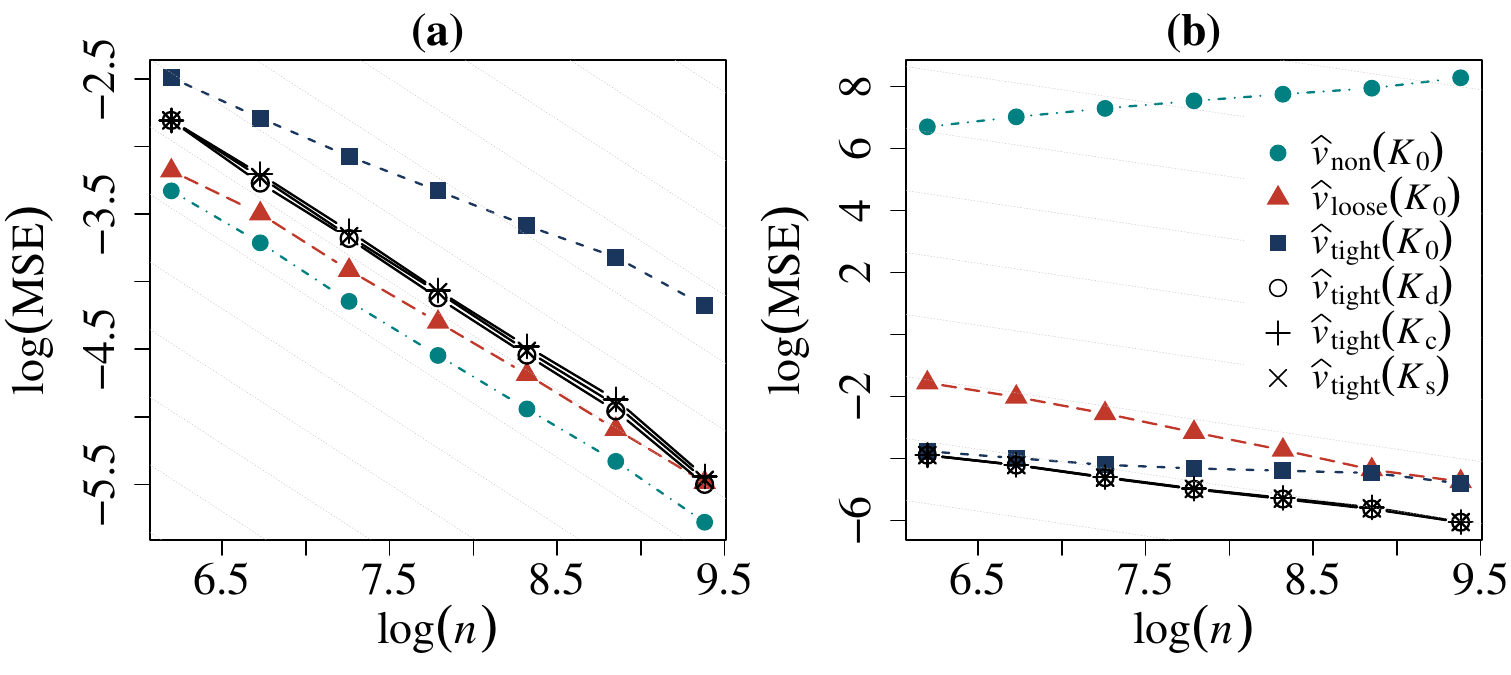}
        \vspace{-0.3cm}
		\caption{
				The values of $\log\MSE(\cdot)$ for 
				$\hat{v}_{\non}(K_0)$, 
				$\hat{v}_{\loose}(K_0)$, 
				$\hat{v}_{\tight}(K_0)$, 
				$\hat{v}_{\tight}(K_\text{d})$, 
				$\hat{v}_{\tight}(K_\text{c})$, and 
				$\hat{v}_{\tight}(K_\text{s})$ 
				are plotted against $\log n$
                under (a) the constant mean case and (b) the non-constant mean case; see Example \ref{ex:rateOptimality}.
                The gray lines are of slope $-4/5$ indicating the optimal rate of convergence; 
                see \S\ref{sec:Op_tightD}. The experiment is performed with $2^{10}$ replications.}
		\label{fig:optimal}
	\end{center} \vspace{-0.3cm}
\end{figure} 

\section{Asymptotic theory}\label{sec:theory}
\subsection{Consistency} \label{sec:consistency}

Let $\mu(t) = c(t) + s(t)$, where $c(t)$ is a continuous function and $s(t)=\sum_{s=0}^{\mathcal{J}}\iota_s\mathbb{1}(T_s/n \leq t < T_{s+1}/n)$ is a step-discontinuous function. 
Here $\mathcal{J}$ is the number of discontinuities, 
$\iota_0, \ldots, \iota_{\mathcal{J}}$ are step sizes such that $\iota_s \neq \iota_{s+1}$ for each $s$, 
and $1 \equiv T_0 < T_1 < \cdots < T_{\mathcal{J}} <T_{\mathcal{J}+1} \equiv n+1$ are the time points of discontinuities. 
Denote 
$\mathcal{G}=\min_{0 \leq s \leq \mathcal{J}}(T_{s+1}-T_s)$, 
$\mathcal{S}=\sup_{0 \leq s \leq \mathcal{J}} \left|\iota_s - \iota_{s+1}\right|$,  and 
$\mathcal{C}=\sup_{0 \leq t' < t \leq 1} \left|\{c(t)-c(t')\}/(t-t')\right|$. 
Let $\{\varepsilon_i,\varepsilon_i'\}_{i\in\mathbb Z}$ be independent and identically distributed,
and set $\mathcal F_i=(\ldots,\varepsilon_{i-1},\varepsilon_i)$. 
Assume $Z_i=g(\mathcal F_i)$ is stationary, where $g$ is a measurable function. 
Define 
$\mathcal F_{i,\{j\}}=(\mathcal F_{j-1},\varepsilon_j',\varepsilon_{j+1},\ldots,\varepsilon_i)$ 
and $Z_{i,\{j\}}=g(\mathcal F_{i,\{j\}})$. 
For $\nu\ge1$, the physical dependence measure is 
$\theta_{\nu,i}=\|Z_i-Z_{i,\{0\}}\|_\nu$ with 
$\|X\|_\nu=\{\E(|X|^\nu)\}^{1/\nu}$, and 
$\Theta_\nu=\sum_{i=0}^\infty\theta_{\nu,i}$.
Write $\E_0$ and $\Var_0$ for expectation and variance under constant $\mu(t)$.
Define ${\Bias}_0(\hat{v})=\E_0(\hat{v})-v$ and ${\Bias}(\hat{v})=\E(\hat{v})-v$. 

\begin{assumption}\label{assump:weakdep}
Let $\{Z_i\}_{i \in \mathbb{Z}}$ be strictly stationary with $\E(Z_1^r) < \infty$, $r>4$ and $\Theta_4 < \infty$. 
\end{assumption}

The finiteness of $\Theta_p$ is a mild regularity condition that ensures weak dependence of $\{Z_i\}$ \citep{Wu2005,Wu2007,Wu2011}.
Assumption \ref{assump:weakdep} ensures $v$ exists.
See Remark \ref{rem:nonstationary} for discussion of the nonstationary case under heteroskedasticity. 
For $p\in\mathbb{N}_0\equiv\{0,1,2,\ldots\}$, let 
$u_p=\sum_{k\in\mathbb{Z}}|k|^p|\gamma_k|$ and 
$v_p=\sum_{k\in\mathbb{Z}}|k|^p\gamma_k$. 
The next theorem gives the bias and variance of $\hat{v}_{\tight}$.

\begin{theorem}[Bias and variance of $\hat{v}_{\tight}$]\label{thm:MSE_propsal_lrv}
Suppose Assumption \ref{assump:weakdep} holds.
Let $K \in \mathcal{K}_{q,q^{\prime}}$ with $q,q^{\prime}\in\mathbb{N}$ such that 
$q \leq q^{\prime}$. 
Define $B$ and $B^{\prime}$ as in (\ref{eqt:BBprime}). 
Assume that $u_q <\infty$,
$m\in\mathbb{N}$, 
$\lambda=h/\ell=1$, and $1/\ell+\ell/n=o(1)$ as $n \to \infty$. 
(i) If $\mu_1 = \cdots = \mu_n$, then 
        \[
	       {\Bias}_0(\hat{v}_{\tight})= \frac{B_d v_q}{\ell^q}+O\left(\frac{\ell}{n}\right)+o\left(\frac{1}{\ell^q}\right) 
	       \quad \text{and} \quad
	       {\Var}_0(\hat{v}_{\tight}) = \frac{4 A_d \ell v^2}{n} + o\left(\frac{\ell}{n}\right),  
        \]
        where 
        $\kappa_1 = \int_{0}^{1}K(1-t)K(t)\, \dd t$,  
        $\kappa_2=\int_{0}^{1}K^2(t)\, \dd t$,   
        \begin{align}
	       A_d = 2 \kappa_1 \sum_{s=1}^{m}\delta_{s-1}\delta_s + 
                \kappa_2 \sum_{|s|\leq m}\delta_s^2
            \qquad \text{and} \qquad
	       B_d = B - \delta_1 B^{\prime}\mathbb{1}(q=q^{\prime}).
            \label{eqt:ABd} 
        \end{align}  
(ii) If $\mu_1, \ldots, \mu_n$ are arbitrary and satisfy $\mathcal{G} \gtrsim \ell+mh$, then 
\begin{align} 
	{R}_{\bias} &= \Bias(\hat{v}_{\tight})-{\Bias}_0(\hat{v}_{\tight}) 
            = O \left\{ m\ell^2  (\mathcal{S}^2 \mathcal{J} + m\ell\mathcal{C}^2/n  ) /n\right\} ,  \label{eqt:R_bias}\\
	{R}_{\se} &= {\Var}^{1/2}(\hat{v}_{\tight}) - {\Var}_0^{1/2}(\hat{v}_{\tight}) 
            = O \left\{ (m\ell^3)^{1/2}  (\mathcal{S}^2 \mathcal{J} + m\ell\mathcal{C}^2/n )^{1/2}/n\right\},
            \label{eqt:R_se}
\end{align}
where $\mathcal G$, $\mathcal S$, $\mathcal J$, and $\mathcal C$ are defined at the beginning of \S\ref{sec:consistency}.
\end{theorem}
The proof of Theorem \ref{thm:MSE_propsal_lrv} is given in \S\ref{pf:thm_MSE(K)} of the supplement. 
The generalized result with $\lambda\in(1,2)$, which we refer to as the sub-tight regime, is deferred to \S\ref{sec:sub-tightTheory}. 

The asymptotic expansion in Theorem \ref{thm:MSE_propsal_lrv} clarifies how different components affect the bias and variance. 
We discuss these aspects below and provide numerical comparisons in \S\ref{sec:paraVals_thm1}. 
(i) The bandwidth $\ell$ trades bias for variance: a larger $\ell$ includes more sample autocovariances in the kernel sum, reducing the bias by capturing more of the dependence structure of the process. 
However, estimating additional autocovariances introduces sampling variability, so the variance increases proportionally to $\ell/n$. 
(ii) The leading bias constant $B_d$ reflects the near-origin flatness of the effective kernel $K_d$. 
For example, if the proposed centrosymmetric kernel in (\ref{eqt:K_transformed}) is used with $K_0(t) = (1-|t|^q)^+$, 
then $B=B'=-1$ and $q=q'$. In this case, $B_d$ simplifies neatly to $B_d = \delta_1 - 1$.
The leading variance constant $A_d$ can be written as its $\mathcal{L}^2$ mass: $A_d=\int_0^{1+m} K_d^2(t)\, \D t$. 
A more spread-out $K_d$ therefore yields a larger $A_d$ and hence a larger variance.
The explicit formula for $A_d$ using the kernel $K_{\text{s}}$ defined in Example \ref{ex:Ks} is provided in \S\ref{sec:general_t0}. 
(iii) The kernel order $q$ determines the bias decay rate. 
Higher-order kernels yield a faster bias rate $O(\ell^{-q})$, but typically involve larger constants $A_d$, which may increase finite-sample variance. 
(iv) The order $m$ of the difference sequence $d_{0:m}$ determines the degrees of freedom available for removing a time-varying mean and reducing estimation variance. The choice of difference sequence affects the bias and variance of $\hat{v}_{\text{tight}}$ through the coefficients $\delta_0, \ldots, \delta_m$, which in turn influence the constants $A_d$ and $B_d$.
(v) The quantities $v$ and $v_q$ reflect properties of the underlying dependence structure. 
Specifically, $v$ is the long-run variance itself, while $v_q = \sum_{k\in\mathbb{Z}} |k|^q \gamma_k$ 
measures the strength of serial dependence at larger lags. 
A more strongly positively correlated time series results in a larger $v_q$, and consequently, a larger bias. 
However, for negatively correlated time series, the effect on $v_q$ is more complex. 
In general, the estimation of $v_q$ is needed for the optimal selection of $\ell$, which will be discussed in (\ref{eq:vq_hat}).
(vi) According to (\ref{eqt:R_bias})--(\ref{eqt:R_se}), 
the arbitrary means $\mu_1, \ldots, \mu_n$ do not affect the leading term of $\MSE(\hat{v}_{\tight})$, 
provided that $R_{\bias}^2+R_{\se}^2=o(1/\ell^{2q}+\ell/n)$, 
which holds when $\mathcal{S}^2\mathcal{J} + \ell\mathcal{C}^2/n = o(n/\ell^{2+q})$.
In particular, if $\ell \asymp n^{1/(1+2q)}$, i.e., the optimal bandwidth to be derived in \S\ref{sec:Op_tightD}, 
then this condition is satisfied if 
$\mathcal{S}^2\mathcal{J} = o\{n^{(q-1)/(2q+1)}\}$ and $\mathcal{C}^2 = o\{n^{(3q-1)/(2q+1)}\}$.
Hence, $\hat{v}_{\tight}$ can handle a divergent number of non-diminishing jumps when $q>1$.

This theorem is the theoretical ground for 
the detailed behavior of $\hat{v}_{\tight}$, including  
(i) deriving optimal parameters,  
(ii) theoretically comparing the mean squared error among 
$\hat{v}_{\tight}$, $\hat{v}_{\loose}$ and $\hat{v}_{\non}$  under constant means, and 
(iii) assessing the robustness of $\hat v_{\tight}$ and $\hat v_{\loose}$ to non-constant means; 
see \S\S\ref{sec:Op_tightD}--\ref{sec:robustness}, respectively.  
Moreover, the assumption $u_q<\infty$ here holds for polynomially decaying autocovariance. 
Choices of $\ell$ for other decay regimes, such as geometric decay and finite-lag autocorrelation, are given in \S\ref{sec:diff_various_strength}. 
If $\ell=o\{ n^{1/(1+q)}\}$, then 
\begin{eqnarray} \label{eq:MSE(K)}
	{\MSE}_0(\hat{v}_{\tight}) 
	= {\Bias}_0^2(\hat{v}_{\tight})+{\Var}_0(\hat{v}_{\tight}) 
	\sim \left( B_d v_q / \ell^{q} \right)^2 +  4 A_d \ell v^2 / n, 
\end{eqnarray}
where $a_n \sim b_n$ 
means that $a_n/b_n \to 1$ as $n \to \infty$.
Minimizing the leading terms in \eqref{eq:MSE(K)} with respect to $\ell$
gives the optimal bandwidth order $\ell \asymp n^{1/(1+2q)}$,
yielding $\MSE_0(\hat v_{\tight}) = O(n^{-2q/(1+2q)})$.
This is the same order as achieved by 
the classical non-difference-based estimators $\hat{v}_{\non}$ 
and 
loose-difference-based estimators $\hat{v}_{\loose}$ 
obtained with a $q$th-order kernel under the dependence condition
$u_q=\sum_{k\in\mathbb Z}|k|^q|\gamma_k|<\infty$, which we assume in Theorem \ref{thm:MSE_propsal_lrv};
see, e.g., \citet{andrews1991heteroskedasticity,NeweyWest1987,politis2003,Chan2022}.
Our proposed $\hat{v}_{\tight}$ attains the same bias and variance rates as $\hat v_{\non}$ \citep{Parzen1957}, 
but with two key differences. 
(i) The bias of $\hat{v}_{\tight}$ depends not only on the origin-flatness $B$ but also on the boundary-flatness $B'$ 
and the differencing term $\delta_1$, 
whereas the biases of $\hat v_{\non}$ and $\hat v_{\loose}$ depend only on $B$ \citep{Parzen1957,Chan2022}. 
(ii) The variance of $\hat{v}_{\tight}$ depends on $K$ through both $\kappa_2$ and $\kappa_1$, with weights determined by $\delta_1,\ldots,\delta_m$, 
while the variances of $\hat v_{\non}$ and $\hat v_{\loose}$ involve only $\kappa_2$. 
Both findings are new. 
The next theorem states the asymptotic normality
for $\hat v_{\tight}$; see \S\ref{pf:thm_clt_vhat} for the proof. 

\begin{theorem}[Asymptotic normality of $\hat{v}_{\tight}$] \label{thm:clt_vhat}
    Assume the conditions of Theorem \ref{thm:MSE_propsal_lrv}. 
    Further suppose $\sum_{i=0}^\infty i \theta_{\nu,i} < \infty$ for some $\nu>4$.
    If $\mu_1=\cdots=\mu_n$, then  
    \[
    	(n/\ell)^{1/2} \left\{ \hat{v}_{\tight} - \E(\hat{v}_{\tight}) \right\} \inD \Normal(0, 4v^2A_d).
    \] 
\end{theorem}
\noindent Let $p \in \mathbb{N}_0$.  
Similar to (\ref{eq:v_hat}), we define the tight estimator of $v_p$, i.e., with $\lambda=1$ in \eqref{eq:D}: 
\begin{align} \label{eq:vq_hat} 
    \hat{v}_{p,\tight} 
    = \sum_{|k| \leq \ell} |k|^p K({k}/{\ell}) \hat{\gamma}_k,  
\end{align}
which is used to select $\ell$ empirically; see \S\ref{sec:Op_tightD}.
Analogues of Theorems \ref{thm:MSE_propsal_lrv}--\ref{thm:clt_vhat} for $\hat{v}_{p,\tight}$ appear in \S\ref{sec:theory_vp_tight} of the supplement. 
In particular, under suitable regularity, if $u_{p+q}<\infty$ for some $q\in\mathbb N$ and $\ell\asymp n^{1/(1+2p+2q)}$, then 
$\MSE_0(\hat{v}_{p,\tight}) = O\{n^{-2q/(1+2p+2q)}\}$.

\begin{remark}[Nonstationary extensions]\label{rem:nonstationary}
Our theoretical results are established under second order stationarity of $\{Z_t\}$.
Literature on heteroskedasticity and autocorrelation consistent (\HAC) long-run variance estimation has also considered nonstationary settings;
see, for example,
\citet{andrews1991heteroskedasticity}, \citet{AndrewsMonahan1992}, and \citet{Casini2022CommentAndrews1991}.
More recent work, such as \citet{Casini2023Misspecified}, \citet{CasiniDengPerron2025LowFrequency}, and \citet{liuchan2025} 
developed \HAC-type estimators that explicitly account for nonstationarity  
and issues arising in hypothesis testing for structural breaks.
As in the literature above, 
the proposed $\hat{v}_{\tight}$ is constructed under the small-bandwidth regime ($\ell/n \to 0$) and may also be adapted to certain nonstationary settings, 
such as processes with locally stationary dependence.
By contrast, fixed-bandwidth ($\ell/n \rightarrow b \in\mathbb{R}^+$) asymptotics, developed in 
\citet{kiefer2002heteroskedasticity} and \citet{kiefer2005new}, relies critically on covariance stationarity. 
Under nonstationarity, the limiting distribution may become non-pivotal; see \cite{Casini2023Misspecified} and \cite{casini2024prewhitened}.
\end{remark}

\subsection{Optimal tight differencing} \label{sec:Op_tightD}

We derive the optimal $\ell$ and $d_{0:m}$ for 
$\hat{v}_{\tight} = \hat{v}_{\tight}(K, \ell, d_{0:m})$ by minimizing ${\MSE}_0(\hat{v}_{\tight})$.
Since \eqref{eq:MSE(K)} is minimized at $\ell\asymp n^{1/(1+2q)}$ with $\MSE(\hat{v}_{\tight})=O\{n^{-2q/(1+2q)}\}$, 
we set $\ell=\xi n^{1/(1+2q)}$ for $\xi>0$. 
Let $\mathcal D_m$ be the set of $m$th-order difference sequences. 
We define
\begin{align}
    \bar{\xi}^*(K, d_{0:m})
    &= \argmin_{\xi \in\mathbb{R}^+ } \lim_{n\rightarrow \infty} n^{\frac{2q}{1+2q}}{\MSE}_0\{\hat{v}_{\tight}(K, \ell, d_{0:m})\}, \label{eqt:optBW0}  \\
    d_{0:m}^*(K)  &= \argmin_{d_{0:m}\in\mathcal{D}_m} \lim_{n\rightarrow \infty} n^{\frac{2q}{1+2q}}{\MSE}_0\{\hat{v}_{\tight}(K, \bar{\ell}^*, d_{0:m})\}, \label{eqt:optDS} 
\end{align}
where $\ell = \xi n^{1/(1+2q)}$ in (\ref{eqt:optBW0}) and  
$\bar{\ell}^* = \bar{\ell}^*(K, d_{0:m}) = \bar{\xi}^*(K, d_{0:m}) n^{1/(1+2q)}$ in (\ref{eqt:optDS}).
Putting the optimal difference sequence (\ref{eqt:optDS}) into (\ref{eqt:optBW0}),
the final optimal $\ell$ is
\begin{align}
	\ell^*(K) = \bar{\xi}^*(K, d_{0:m}^*(K)) n^{1/(1+2q)}.
	\label{eqt:optBW}
\end{align}
An analytical form of $\bar{\ell}^*(K, d_{0:m})$ and 
an equivalent $d_{0:m}^*(K)$ are shown below. 

\begin{proposition}[Optimal parameters] \label{prop:l_lrv}
Assume the conditions of Theorem \ref{thm:MSE_propsal_lrv}. 
Define $B_d$ and $A_d$ as in (\ref{eqt:ABd}).
If $v_q \neq 0$,
then
\begin{align}
	\bar{\xi}^*(K, d_{0:m}) &=  \left\{ \frac{ (B_d v_q/v)^2 q }{2 A_d} \right \}^{1/(1+2q)},   \label{eq:op_l} \\
	 d^*_{0:m}(K)  &= \argmin_{d_{0:m}\in\mathcal{D}_m} \left \{B - \delta_1 B^{\prime}\mathbb{1}(q=q^{\prime})\right\}^2   
    \left (2 \kappa_1 \sum_{s=1}^{m}\delta_{s-1}\delta_s + \kappa_2 \sum_{|s|\leq m}\delta_s^2  \right)^{2q} . \label{eqt:optDS_problem}
\end{align}
\end{proposition}
Proposition~\ref{prop:l_lrv} allows us to compute $d_{0:m}^*(K)$ 
numerically via~(\ref{eqt:optDS_problem}); see \S\ref{sec:more_dj_opt} for details. 
To obtain $\bar{\xi}^*(K, d_{0:m})$ empirically, we need to estimate $|v_q/v|$; 
see \S\ref{sec:eta}.
When $m>1$, the optimal tight difference sequences $d_{0:m}^*(K)$ 
may be different from the optimal $d_{0:m}$ for
independent data \citep{Hall1990} and for 
dependent data under loose differencing \citep{Chan2022},
which satisfy  
\begin{align} \label{eqt:loose_optDS_problem} 
    d_{0:m}^{\dagger} 
    = \argmin_{d_{0:m}\in\mathcal{D}_m}\,  B^2  \left ( \kappa_2 \sum_{|s|\leq m}\delta_s^2 \right)^{2q}
    = \argmin_{d_{0:m}\in\mathcal{D}_m}\, \sum_{|s|\leq m}\delta_s^2,  
\end{align} 
which is free of $K$.
Although the optimal $d_{0:m}^*(K)$ depends on $K$
through $B$, $B^\prime$, $\kappa_1$, and $\kappa_2$, 
it is independent of the data-generating process. 
This surprising property is due to the separation of 
the process-dependent quantities $(v,v_q)$ and the difference sequence-dependent quantities $(A_d,B_d)$
in the leading mean-squared error; see \S\ref{sec:proof_prop_optimal_d_l} for details. 
So, $d_{0:m}^*(K)$ can be tabulated for implementation; see the example below.  

\begin{table}[t]
\small
\captionsetup{font=small}
\centering
\caption{
The optimal values of $(d_0, \ldots, d_m)$ under 
tight differencing and loose differencing for estimation of $v$.
The centrosymmetric kernel $K_{\text{s}}$ in Example \ref{ex:Ks} is used 
with $K_0(t) = (1-|t|^q)^+$.
}\vspace{-0.2cm}
\begin{tabular}{llll} 
$m$& $q$& Tight difference sequence ($\lambda=1$) & Loose difference sequence ($\lambda=2$) \\
$2$ & $2$ & $0.3203,\,0.4902,\,-0.8106$ &  $0.8090, \,-0.5,\, -0.3090$\\
 & $3$ & $0.3201,\,0.4904,\,-0.8106$  & Same as above\\
 & $4$ & $0.3144,\,0.4954,\,-0.8098$  & Same as above\\
$3$  & $2$ & $0.2084,\,0.2781,\,0.3736,\,-0.8601$ & $0.1942, 0.2809, 0.3832, -0.8582$\\
  & $3$ & $0.2053,\,0.2819,\,0.3727,\,-0.8599$ & Same as above\\
  & $4$ & $0.2011,\,0.2793,\,0.3787,\,-0.8591$ & Same as above\\ 
\end{tabular}  
\label{table:dj_dc_opt}  
\end{table}

\begin{example}[Optimal difference sequence] 
We use the centrosymmetric kernel in Example \ref{ex:Ks} with 
$K_0(t)=(1-|t|^q)^+$. 
Table \ref{table:dj_dc_opt} gives the optimal tight and loose sequences
$d_{0:m}^*(K)$ and $d_{0:m}^{\dagger}$ for $q\in\{1,2,3\}$ and $m\in\{2,3\}$; 
see Tables \ref{table:dj_dc_opt_more}--\ref{table:dj_s_opt} of the supplement for additional cases. 
Unlike $d_{0:m}^{\dagger}$, which is $q$-invariant and coincides with the results for independent data, 
$d_{0:m}^*(K)$ is adapted to different values of $q$, which measures the strength of serial dependence. 
\end{example} 

\subsection{Efficiency comparison under constant mean} \label{sec:efficiency}
\begin{table}[t]
\small
\captionsetup{font=small}
\centering
\caption{
The limits of $n^{2q/(1+2q)}(v_q/v)^{-2/(1+2q)}\,\MSE_0(\cdot)/v^2$ across estimators. 
Here, $K_{\text d}$, $K_{\text c}$, and $K_{\text s}$ are as in Examples \ref{ex:Kd}--\ref{ex:Ks} with $K_0(t)=(1-|t|^q)^+$. 
Percentages in parentheses are relative changes compared to $\hat v_{\non}(K_0)$.
}\vspace{-0.2cm}
\begin{tabular}{cc ccccc}
$q$ & $m$ & $\hat{v}_{\non}(K_0)$ & 
$\hat{v}_{\loose}(K_0)$ & 
$\hat{v}_{\tight}(K_{\text{d}})$ & 
$\hat{v}_{\tight}(K_{\text{c}})$ & 
$\hat{v}_{\tight}(K_{\text{s}})$ \\
$2$ & $2$ & $3.024$ & $3.615$\ $(+19.5\%)$ & $3.186$\ $(+5.4\%)$ & $3.036$\ $(+0.4\%)$ & $3.079$\ $(+1.8\%)$ \\
    & $3$ & $3.024$ & $3.421$\ $(+13.1\%)$ & $2.977$\ $(-1.6\%)$ & $2.851$\ $(-5.7\%)$ & $2.888$\ $(-4.5\%)$ \\
	& $4$ & $3.024$ & $3.323$\ $(+9.9\%)$ & $2.869$\ $(-5.1\%)$ & $2.755$\ $(-8.9\%)$ & $2.788$\ $(-7.8\%)$ \\
$3$ & $2$ & $3.386$ & $4.100$\ $(+21.1\%)$ & $3.287$\ $(-2.9\%)$ & $3.001$\ $(-11.4\%)$ & $3.077$\ $(-9.1\%)$ \\
    & $3$ & $3.386$ & $3.864$\ $(+14.1\%)$ & $3.054$\ $(-9.8\%)$ & $2.815$\ $(-16.9\%)$ & $2.878$\ $(-15.0\%)$ \\
	& $4$ & $3.386$ & $3.746$\ $(+10.6\%)$ & $2.937$\ $(-13.3\%)$ & $2.719$\ $(-19.7\%)$ & $2.777$\ $(-18.0\%)$ \\
\end{tabular} 
\label{table:MSEvs} 
\end{table}
We compare $\hat v_{\tight}(K)$, $\hat v_{\loose}(K_0)$, and $\hat v_{\non}(K_0)$, 
where $K_0(t)=(1-|t|^q)^+$ and $K$ is the centrosymmetric version in Examples \ref{ex:Kd}--\ref{ex:Ks} built from the same $K_0$. 
We view $\hat v_{\non}$ as an oracle as it is consistent only under constant means. 
Each estimator uses its optimal tuning parameters.  
Let $L_q = \{(4 q)/(1 + 3 q + 2 q^2)\}^{2q/(1+2q)}$.
Under the conditions in Theorem \ref{thm:MSE_propsal_lrv} (i), we have
\begin{align*} 
	n^{2q/(1+2q)}{\MSE}_0(\hat{v}_{\tight}) &\rightarrow (1+2q) v^2 \left\vert v_q/v  \right\vert^{2/(1+2q)}\left \{ |B_d|  ( 2 A_d /q )^{q}  \right \} ^{2/(1+2q)}, \\
    n^{2q/(1+2q)}{\MSE}_0(\hat{v}_{\loose}) &\rightarrow (1+2q)v^2\left\vert v_q/v  \right\vert^{2/(1+2q)}  L_q \left\{ 1+1/(2m) \right\}^{2q/(1+2q)} , \\
	n^{2q/(1+2q)}{\MSE}_0(\hat{v}_{\non}) &\rightarrow (1+2q)v^2 \left\vert  v_q/v   \right\vert^{2/(1+2q)} L_q.
\end{align*}
Table \ref{table:MSEvs} compares the leading constants. 
All three estimators satisfy
$\Bias_0^2(\cdot):\Var_0(\cdot):\MSE_0(\cdot)=1:2q:(1+2q)$.
Therefore, $\hat v_{\tight}$ strictly improves $\hat v_{\loose}$ in terms of squared bias, variance, and 
mean squared error for $q\in\{2,3\}$ and $m\in\{2,3,4\}$.
It is interesting to observe that, loose differencing reduces statistical efficiency, in the sense that 
$\lim_{n\rightarrow\infty}{\MSE}_0(\hat{v}_{\loose}) / {\MSE}_0(\hat{v}_{\non}) >1$, 
for all $q$ and $m$ concerned, 
however, tight differencing can even improve statistical efficiency, in the sense that 
$\lim_{n\rightarrow\infty}{\MSE}_0(\hat{v}_{\tight}) / {\MSE}_0(\hat{v}_{\non}) <1$, 
provided that high enough orders of kernel and differencing are used, i.e., $\min(q,m)>2$.
This can happen because the effective kernel $K_d$ under tight differencing may have a reduced $\mathcal{L}^2$ kernel mass 
$A_d=\int_0^{1+m} K_d^2(t)\, \D t$, which consequently reduces the variance of $\hat v_{\tight}$.

\subsection{Strength of robustness against time-varying mean} \label{sec:robustness}

To assess the robustness to $\mu(\cdot)$, 
we represent $\hat{v}$ as
$\hat{v}= \hat{v}^{\mu\mu} + \hat{v}^{ZZ} + \hat{v}^{\mu Z} + \hat{v}^{Z\mu}$,
where 
\[
\hat{v}^{\alpha \beta} = \sum_{|k| \leq \ell}K({k}/{\ell}) \left( \frac{1}{n} \sum_{i=mh+|k|+1}^n D_i^\alpha D_{i-|k|}^\beta \right) 
\]
for $\alpha, \beta \in \{\mu, Z\}$
with 
$D_i^{Z} = \sum_{j=0}^{m}d_j Z_{i-jh}$ and $D_i^{\mu} = \sum_{j=0}^{m}d_j \mu_{i-jh}$. 
Since
$\E(\hat{v}^{\mu Z}) = \E(\hat{v}^{Z \mu}) =0$, 
we have 
$
    \Bias(\hat{v})
        = {\Bias}_0(\hat{v})+\hat{v}^{\mu\mu}. 
$
So, $\hat{v}^{\mu\mu}$ measures the additional bias introduced by the mean,
whose smaller values indicate greater robustness to $\mu(\cdot)$.

\begin{proposition}[Robustness Strength] \label{prop:robust}
Let $\ell \sim \xi n^{\vartheta}$, where $\xi \in \mathbb{R}^+$ and ${\vartheta} \in (0,1)$. 
Let $\mathcal{C}_0, \mathcal{C},\iota_0, \ldots, \iota_\mathcal{J}, \mathcal{J}$ be some quantities whose 
magnitude may diverge to $\infty$.
(i) 
If $\mu(t) = \mathcal{C}_0 + \mathcal{C} t$,  
			then 
            $\hat{v}^{\mu\mu}/n^{3{\vartheta}-2} \sim 2 \mathcal{C}^2 \xi^3 \kappa_0 \lambda^2 (\sum_{j=0}^m j d_j)^2$.
(ii) If $\mu(t) = \sum_{s=0}^{\mathcal{J}}\iota_s\mathbb{1}(T_s/n \leq t < T_{s+1}/n)$ such that $\mathcal{G} \gtrsim \ell+mh$, 
            then $\hat{v}^{\mu\mu}/n^{2{\vartheta}-1} \sim 2\xi^2\kappa_0 \lambda ( \sum_{s=1}^{\mathcal{J}} \iota_s ) \sum_{p=0}^{m-1}(\sum_{j=0}^p d_j)^2$.  
\end{proposition}

Proposition \ref{prop:robust} (i) and (ii) discuss continuous $\mu(t)$ and piecewise constant $\mu(t)$, respectively.  
Let $\hat v^{\mu\mu}_{\tight}$ and $\hat v^{\mu\mu}_{\loose}$ denote $\hat v^{\mu\mu}$ for $\lambda=1$ and $\lambda=2$. 
By Proposition \ref{prop:robust}, 
$\hat v^{\mu\mu}_{\loose}/\hat v^{\mu\mu}_{\tight}\to4$ in Case 1 with $r=1$ and 
${\hat{v}^{\mu\mu}_{\loose}}/{\hat{v}^{\mu\mu}_{\tight}} \to 2$
in Case 2.
We emphasize that Case 2 allows $\mathcal{J}\rightarrow\infty$. 
Thus, $\hat v^{\mu\mu}_{\tight}$ is $4$ times and $2$ times more robust than $\hat v^{\mu\mu}_{\loose}$ in these two cases, respectively. 

\section{Spectral density estimation} \label{sec:spectral}

\subsection{Tight-difference-based spectral density estimator} \label{sec:consistency_spec}
We now extend the tight-differencing construction from long-run variance estimation, where $v=2\pi f(0)$, to the spectral density estimation of $f(\theta)$ over $\theta \in [0,0.5]$.
The spectral density can be written as 
$f(\theta) = (2\pi)^{-1} \sum_{k\in \mathbb{Z}}\gamma_k \cos(2\pi\theta k)$
for $\theta \in [0,0.5]$.
With $\lambda=h/\ell=1$, 
the tight-difference-based spectral density estimator is  
\begin{align}  \label{eq:spec.sym}
	\hat{f}_{\tight}(\theta) = \frac{1}{2\pi}  \sum_{|k| \leq \ell} K
		(k/{\ell}) \hat{\gamma}_k \cos(2\pi\theta k), \quad \theta \in [0,0.5]. 
\end{align} 
We remark that the same
compatibility issue 
as discussed in long-run-variance estimation persists across frequencies.
To make $\hat{f}_{\tight}(\theta)$ rate-optimal, we require the kernel $K \in \mathcal{K}_{q,q^{\prime}}$ with $q \leq q^{\prime}$.

\begin{theorem}[Bias and variance of $\hat{f}_{\tight}$] \label{thm:MSE_proposal_density}
Assume the conditions in Theorem \ref{thm:MSE_propsal_lrv} hold.

(i) If $\mu_1 = \cdots = \mu_n$, then with $B_d$ and $A_d$ defined in (\ref{eqt:ABd}), we have  
\begin{align*} 
    {\Bias}_0 \{\hat{f}_{\tight}(\theta) \} &\equiv \E_0\{\hat{f}_{\tight}(\theta) \} - f(\theta)
        =  {B_d f_{q}(\theta)}/{\ell^q}+O({\ell}/{n} )+o(1/{\ell^q})  , \\
    {\Var}_0\{\hat{f}_{\tight}(\theta)\} 
        &= 4 A_d \ell f^2(\theta)\varpi(\theta)/n + o ({\ell}/{n}),    
\end{align*} 
where $f_{q}(\theta)=\sum_{k \in \mathbb{Z}} |k|^q \gamma_k \cos(2 \pi \theta k)/ (2 \pi)$ and $\varpi(\theta) = 1-\mathbb{1}(0<\theta <0.5)/2$. 

(ii) If $\mu_1, \ldots, \mu_n$ satisfy $\mathcal{G} \gtrsim \ell+mh$, then  
$\Bias\{\hat{f}_{\tight}(\theta)\} ={\Bias}_0\{\hat{f}_{\tight}(\theta)\}  + {R}_{\bias}$  
and
${\Var}^{1/2}\{\hat{f}_{\tight}(\theta)\}  = {\Var}_0^{1/2}\{\hat{f}_{\tight}(\theta)\}  +  {R}_{\se}$,
where ${R}_{\bias} $ and ${R}_{\se} $ are defined in (\ref{eqt:R_bias})--(\ref{eqt:R_se}).
\end{theorem}

Theorem \ref{thm:MSE_proposal_density} shows that the bias and variance structure derived for $\hat v_{\tight}$ is retained 
in $\hat{f}_{\tight}(\theta)$ at each frequency $\theta$. 
The constants $B_d$ and $A_d$, determined by the effective kernel, remain the same, 
while $\theta$ only changes $f_q(\theta)$ in the bias and $f^2(\theta)\varpi(\theta)$ in the variance.
Thus, under the stated conditions, tight differencing is consistent and rate-optimal not only at frequency zero, but also at other frequencies. Specifically,
by Theorem \ref{thm:MSE_proposal_density}, 
if $\ell=o\{ n^{1/(1+q)}\}$, 
then for each $\theta$,
\begin{eqnarray} \label{eq:MSE_spec}
{\MSE}_0\{ \hat{f}_{\tight}(\theta) \} 
= {\Bias}_0^2  \{\hat{f}_{\tight}(\theta)\}
     + {\Var}_0  \{\hat{f}_{\tight}(\theta)\} 
\sim \frac{B_d^2 f^2_{q}(\theta)}{\ell^{2q}} + \frac{4 A_d \ell f^2(\theta)\varpi(\theta)}{n}.
\end{eqnarray}
And the mean integrated squared error of $\hat{f}_{\tight}(\theta)$ satisfies
\begin{eqnarray} \label{eq:MISE_spec}
	{\MISE}_0 \{ \hat{f}_{\tight}(\cdot) \}
	= \int_{0}^{0.5} \left [{\Bias}_0^2  \{\hat{f}_{\tight}(\theta)   \}
     + {\Var}_0  \{\hat{f}_{\tight}(\theta)    \} \right ]\dd \theta 
	\sim \frac{B_d^2F_q}{\ell^{2q}} 
     + \frac{4 A_d \ell F}{n}, 
\end{eqnarray}
where
$F_q = \int_{0}^{0.5} f_q^2(\theta) \dd \theta$
and $F = F_0$. 
The results in (\ref{eq:MSE_spec})--(\ref{eq:MISE_spec}) mirror \eqref{eq:MSE(K)} for $\hat v_\tight$,
with $f_q(\theta)$ and $f^2(\theta)\varpi(\theta)$, or $F_q$ and $F$, replacing the frequency-zero counterparts $v_q$ and $v^2$. 
Consequently, the $\textsc{mse}$- and $\textsc{mise}$-optimal bandwidths 
remain of order $\ell\asymp n^{1/(1+2q)}$,
yielding $\MSE_0\{ \hat{f}_{\tight}(\theta) \} = O(n^{-2q/(1+2q)})$ and 
$\MISE_0\{ \hat{f}_{\tight}(\cdot) \} = O(n^{-2q/(1+2q)})$.

\begin{theorem}[Asymptotic normality of $\hat{f}_{\tight}$] \label{thm:clt_spec}
    Suppose the conditions in Theorem \ref{thm:MSE_proposal_density} hold.
    Assume $\sum_{i=0}^\infty i \theta_{\nu,i} < \infty$ for some $\nu>4$.
    If $\mu_1=\cdots=\mu_n$, then for $\theta \in [0, 0.5]$,
    \[
		(n/\ell)^{1/2}
    	\left[ \hat{f}_{\tight}(\theta) - \E\{\hat{f}_{\tight}(\theta)\} \right] \inD \Normal(0,  4 f^2(\theta) A_d \varpi(\theta) ).
    \]
\end{theorem}

Theorem \ref{thm:clt_spec} admits a similar form as in \cite{Wu2010}, 
except that differencing affects the value of asymptotic variance through $A_d$; 
see also \cite{WuShao2007}.
Similarly, $f_p(\theta) = (2\pi)^{-1}\sum_{k\in \mathbb{Z}} |k|^p \gamma_k \cos(2\pi\theta k)$ with $p\in \mathbb{N}_0$ 
can be estimated by 
\begin{eqnarray}  \label{eq:fq_hat}
	\hat{f}_{p,\tight}(\theta) = \frac{1}{2\pi}  \sum_{|k| \leq \ell} |k|^p
        K ({k}/{\ell}) \hat{\gamma}_k \cos(2\pi\theta k), \quad \theta \in [0,0.5].
\end{eqnarray}
The mean-squared error and asymptotic normality
of $\hat{f}_{p,\tight}(\theta)$ are stated in Corollaries \ref{corol:MSE_fq}--\ref{corol:clt_spec_w}. 

\subsection{Parameter optimization} \label{sec:Op_tightD_spec}
The asymptotic results in \eqref{eq:MSE_spec}--\eqref{eq:MISE_spec}
lead to the same bias--variance trade-off as in long-run-variance estimation, 
giving optimal bandwidth of order $\ell\asymp n^{1/(1+2q)}$. 
Let $\ell = \xi n^{1/(1+2q)}$ for some $\xi>0$.
For $\hat{f}_{\tight}(\theta) = \hat{f}_{\tight}(\theta, K, \ell, d_{0:m})$ with given $K$ and $d_{0:m}$, 
we define 
\begin{align*}
    \bar{\xi}_{f(\theta)}^*(K, d_{0:m}) 
    &= \argmin_{\ell\in \mathbb{N}} \lim_{n\rightarrow \infty} n^{2q/(1+2q)}{\MSE}_0\left \{\hat{f}_{\tight}(\theta, K, \ell, d_{0:m})\right \}, \\
    \bar{\xi}_{f}^*(K, d_{0:m}) 
    &= \argmin_{\ell\in \mathbb{N}} \lim_{n\rightarrow \infty} n^{2q/(1+2q)}{\MISE}_0\left \{\hat{f}_{\tight}(\cdot, K, \ell, d_{0:m})\right \},
\end{align*}
which are the asymptotic  $\textsc{mse}$- and $\textsc{mise}$-optimal 
$\xi$ with close forms shown below.

\begin{proposition}[Optimal bandwidth]\label{prop:opt_l_f}
Assume the conditions stated in Theorem \ref{thm:MSE_propsal_lrv}. 
Define $B_d$ and $A_d$ as in (\ref{eqt:ABd}).
If 
$f_q(\theta)\neq 0$ and $F_q \neq 0$,
then, respectively,  
\begin{eqnarray} \label{eq:op_lF} 
\bar{\xi}^*_{f(\theta)}(K, d_{0:m}) 
    = \left \{ \left|\frac{f_q(\theta)}{f(\theta)}\right|^2 \frac{ B_d^2 q}{2 A_d \varpi(\theta)} \right \}^{1/(1+2q)}
    \;\text{and} \quad
	\bar{\xi}^*_{f}(K, d_{0:m}) 
    = \left( \frac{F_q}{F} \frac{ B_d^2 q}{2 A_d } \right )^{1/(1+2q)}.
\end{eqnarray}
\end{proposition}
Proposition \ref{prop:opt_l_f} shows that the parameter optimization in spectral density estimation has the same structure as that for long-run variance estimation.
Compared with \eqref{eq:op_l}, the factors $f_q(\theta)/f(\theta)$, $\varpi(\theta)$ and $F_q/F$ affect only the optimal bandwidth constant. 
The part depending on $d_{0:m}$ is still determined by the constants $B_d$ and $A_d$, which have the same form as in long-run-variance estimation.
Consequently, the asymptotic $\textsc{mse}$- and $\textsc{mise}$-optimal difference sequences $d_{0:m}^*(K)$,
obtained by substituting \eqref{eq:op_lF} into \eqref{eq:MISE_spec},
coincide and match the one used for estimating $v$; 
see (\ref{eqt:optDS_problem}) and Table \ref{table:dj_dc_opt}. 
Plugging $d_{0:m}^*$ back into \eqref{eq:op_lF} gives the asymptotic $\textsc{mse}$- and $\textsc{mise}$-optimal  
\begin{align}\label{eqt:opt_bw_spec} 
\ell^*_{f(\theta)}(K) = \bar{\xi}^*_{f(\theta)}(K, d_{0:m}^*(K)) n^{1/(1+2q)} 
\quad \text{and} \quad 
\ell^*_{f}(K) = \bar{\xi}^*_{f}(K, d_{0:m}^*(K)) n^{1/(1+2q)}, 
\end{align}
for estimating the pointwise $f(\theta)$ and the full spectrum, respectively; see \S\ref{sec:eta}.

\section{Implementation} \label{sec:etaCP}

\subsection{Optimal bandwidth estimation} \label{sec:eta}

The optimal bandwidths 
in (\ref{eqt:opt_bw_spec}) 
depend on 
the unknown dependence ratios 
$\psi_q(\theta) \equiv \{f_q(\theta)/f(\theta)\}^2$ and 
$\Psi_q \equiv F_q/F$.
We remark that $\ell^*(K) \equiv \ell^*_{f(0)}(K)$, so we discuss only the general case $\ell^*_{f(\theta)}(K)$. 
Two methods for constructing pilot estimators of $\psi_q(\theta)$ and $\Psi_q$ are provided. 

\begin{example}[Nonparametric pilot]\label{eg:pilot_np} 
Let $\widetilde f_q(\theta)$ be the pilot in \eqref{eq:spec.sym} with 
$m=1$, $\lambda=1$, $K=K_{\mathrm{s}}$ equipped with $K_0(t)=(1-t^2)^+$; see Example \ref{ex:Ks}, 
and bandwidth $\ell=\max\{5,\lfloor 2n^{1/(5+2q)}\rfloor\}$. 
The plug-in estimator of $F_q$ based on $\widetilde{f}_q(\theta)$ is
$\widetilde{F}_q = \sum_{|k|\leq \ell} |k|^q K(k/\ell) \hat{\gamma}_k^2/(2\pi)^2$.
Then, $\psi_q(\theta)$ and $\Psi_q$ 
can be consistently estimated by $\widetilde{\psi}_q(\theta) = \{\widetilde{f}_q(\theta)/\widetilde{f}_0(\theta)\}^2$  
and $\widetilde{\Psi}_q = \widetilde{F}_q/\widetilde{F}_0$, respectively.  
For loose or no differencing, use $(\lambda,K)=(2,K_0)$ or $(0,K_0)$, respectively. 
\end{example}

\begin{example}[Parametric pilot]\label{eg:pilot_p}
Assume a pilot $\AR(1)$ noise model $Z_i=\phi Z_{i-1}+\varepsilon_i$, $|\phi|<1$, with independent $\varepsilon_i$ of finite variance 
$\sigma^2$; see, e.g., \citet{andrews1991heteroskedasticity} for a similar approach. 
Estimate $\phi$ by $\widetilde\phi=\widetilde\gamma_1/\widetilde\gamma_0$, where $\widetilde\gamma_k$ is the lag-$k$ sample autocovariance based on the differenced data  $\widetilde D_i=X_i-X_{i-\lfloor n^{1/2}\rfloor}$. 
Then $f_q(\theta)$ and $F_q$ can be estimated by 
$\widetilde{f}_q(\theta) = \sigma^2 \Re\{ \mathbb{1}(q=0)+2\Li_{-q}(\widetilde{\phi} e^{2\pi \theta i})\} / \{ 2\pi(1-\widetilde{\phi}^2) \}$ and 
$\widetilde{F}_q = \sigma^4 \{ \mathbb{1}(q=0)+2\Li_{-2q}(\widetilde{\phi}^2)\} / \{ 8\pi^2 (1-\widetilde{\phi}^2)^2 \}$, 
where $\Re$ denotes the real part and $\Li_s(x)=\sum_{k\ge1}x^k/k^s$ is the polylogarithm, which is numerically tractable. 
The value of $\sigma^2$ will be cancelled in the quotients
$\widetilde\psi_q(\theta)=\{\widetilde f_q(\theta)/\widetilde f_0(\theta)\}^2$ and 
$\widetilde\Psi_q=\widetilde F_q/\widetilde F_0$. 
And we remark that these pilots are consistent when the $\AR(1)$ specification is correct.
\end{example}

By default, we use the nonparametric pilot in Example \ref{eg:pilot_np}.
Other plug-in choices are possible.
For example, cross‐validation may be used; see \citet{Robinson1991} for a comprehensive review. 
Other methods include frequency-domain cross validation \citep{Wahba1975}, cross-validated log-likelihood \citep{Hurvich1990}, 
and universal bandwidth via converging kernels \citep{chan2024asymptotically,liuchan2025} 

\subsection{Robustification via trimming} \label{sec:CPremoval}
Differencing is effective for asymptotically canceling small-to-moderate mean variation,
But dramatic changes in mean can hurt the finite‐sample performance.
We enhance the robustness of $\hat{v}_{\tight}$ by trimming the most extreme difference statistics.
Let $\Lambda$ be the set of indices after trimming the smallest $\alpha_n/ 2$ and the largest $\alpha_n/ 2$ values in the original $\{ D_i \}_{i=mh+1}^n$, where $\alpha_n \rightarrow \infty$, i.e., 
\[
\Lambda = \left \{i=1,\ldots,n-m\ell: {\frac{\alpha_n}{2}} \leq \hat{S}(D_i) \leq 1-\frac{\alpha_n}{2} \right \}, 
\quad \text{where} \quad 
\hat{S}(t)=\sum_{i=1}^{n-m\ell}\mathbb{1}(D_i \leq t).
\]
Then the trimmed tight difference-based long-run variance estimator is 
\[
	\hat{v}^{\Diamond}_{\tight}=\sum_{|k|\leq \ell}K ({k}/{\ell} )\hat{\gamma}_k^{\Diamond}, 
	\quad \text{where} \quad
	\hat{\gamma}_k^{\Diamond}=\frac{1}{n}\sum_{i, i-|k| \in \Lambda} D_i D_{i-|k|}.
\] 

\begin{theorem} [Consistency]\label{thm:CP}
Assume the conditions stated in Theorem \ref{thm:MSE_propsal_lrv}. 
Let $\alpha_n=O(n^a)$ with $a \in \left [0,q /(1+2q)\right]$.
Then, we have
$\MSE(\hat{v}^{\Diamond}_{\tight})/ \MSE(\hat{v}_{\tight}) \to 1.$
\end{theorem}

So,
$\hat{v}^{\Diamond}_{\tight}$ and $\hat{v}_{\tight}$
have the same asymptotic properties if we remove a small number of $D_i$'s.
This enhances the robustness when distinct changes are observed. 
For example, if $n=200$ and $q=2$, then one may set $\alpha_n=\lfloor 400^{2/5} \rfloor \approx 10$ in practice.
It trimes around 4\% of the most extreme difference statistics, which guards against structural breaks 
without sacrificing efficiency.

\section{Simulation Experiments} \label{sec:Exp}

\subsection{Efficiency under constant mean}
\label{sec:lrvEst_constMean}
We compare different estimators of $f(\theta)$: 
(i) $\hat{f}_{\non}(\theta)$, non-difference-based estimator with $K=K_0$, where $K_0(t) = (1-|t|^3)^+$;
(ii) $\hat{f}_{\loose}(\theta)$, order-$m$ loose-difference-based estimator with $K=K_0$; and 
(iii) $\hat{f}_{\tight}(\theta)$, order-$m$ tight-difference-based estimators with $K=K_{\text{s}}$ in Example \ref{ex:Ks}. 
Each estimator is implemented with its own estimated, asymptotically $\textsc{mse}$-optimal parameters according to \S\ref{sec:eta}. 
For $\hat f_{\loose}(\theta)$ and $\hat f_{\tight}(\theta)$,
we consider $m\in\{1,2,3\}$ with their corresponding optimal difference sequences (Table~\ref{table:dj_dc_opt}). 
Efficiency is measured by mean squared error under $\mu(t)=0$. 
Unless stated otherwise, \S\ref{sec:Exp} uses $2^{10}$ Monte Carlo replications.

\begin{figure}[!t]
    \centering
    \captionsetup{font=small}
    \includegraphics[width=\linewidth] {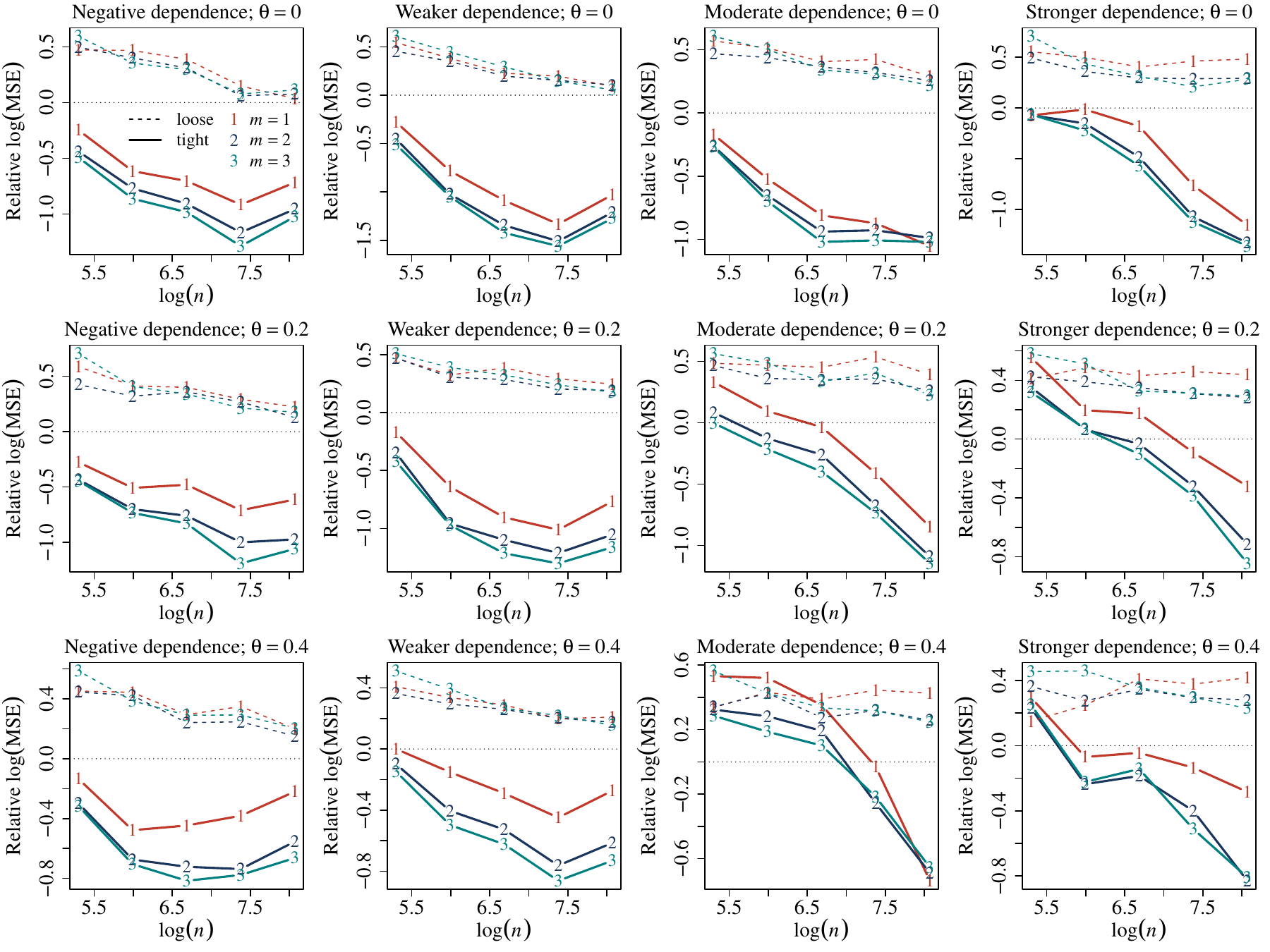}
        \vspace{-0.6cm}
    \caption{The mean squared errors of estimators of $f(\theta)$ relative to 
    the non-difference-based estimator $\hat{f}_{\non}(\theta)$ are compared under the $\ARMA(1,1)$ model. 
    Dashed lines correspond to $\hat{f}_{\loose}(\theta)$, solid lines to $\hat{f}_{\tight}(\theta)$. 
    The numbers on the lines denote the differencing order. Lower curves indicate more efficient estimators.}
    \label{fig:efficiency_ARMA11} 
\end{figure}

First, we consider an $\ARMA(1,1)$ model
$Z_i=\phi Z_{i-1}+0.6\,\varepsilon_{i-1}+\varepsilon_i$, where $\varepsilon_i\sim\Normal(0,1)$ independently and 
$\phi\in\{-0.4,0.2,0.4,0.6\}$.
Figure \ref{fig:efficiency_ARMA11} plots the mean squared errors relative to the benchmark $\hat f_{\non}(\theta)$
for $n\in\{200,400,800,1600,3200\}$ and $\theta\in\{0,0.2,0.4\}$. Results are similar across $\theta$.
Overall, $\hat f_{\tight}(\theta)$ outperforms $\hat f_{\loose}(\theta)$ and, as $n$ grows, also outperforms $\hat f_{\non}(\theta)$.
Efficiency of $\hat{f}_{\tight}(\theta)$ improves with $m$, especially from $m=1$ to $m=2$, with diminishing gains from $m=2$ to $m=3$. 
See \S\ref{sec:lrvEst_constMean_additional} of the supplement for more experiments.  
In addition, we study four more complicated time series models, including  
$\ARMA(4,4)$, bilinear, $\textsc{arch}(1)$, and geometric $\AR(1)$ models in \S\ref{sec:lrvEst_constMean_additional}. 
The proposed estimator $\hat{f}_{\tight}(\theta)$ continues to perform the best.

\subsection{Robustness against non-constant means}\label{sec:lrvEst_nonconstMean}
\begin{figure}[!t]
\captionsetup{font=small}  
	\begin{center}
		\includegraphics[width=\linewidth]{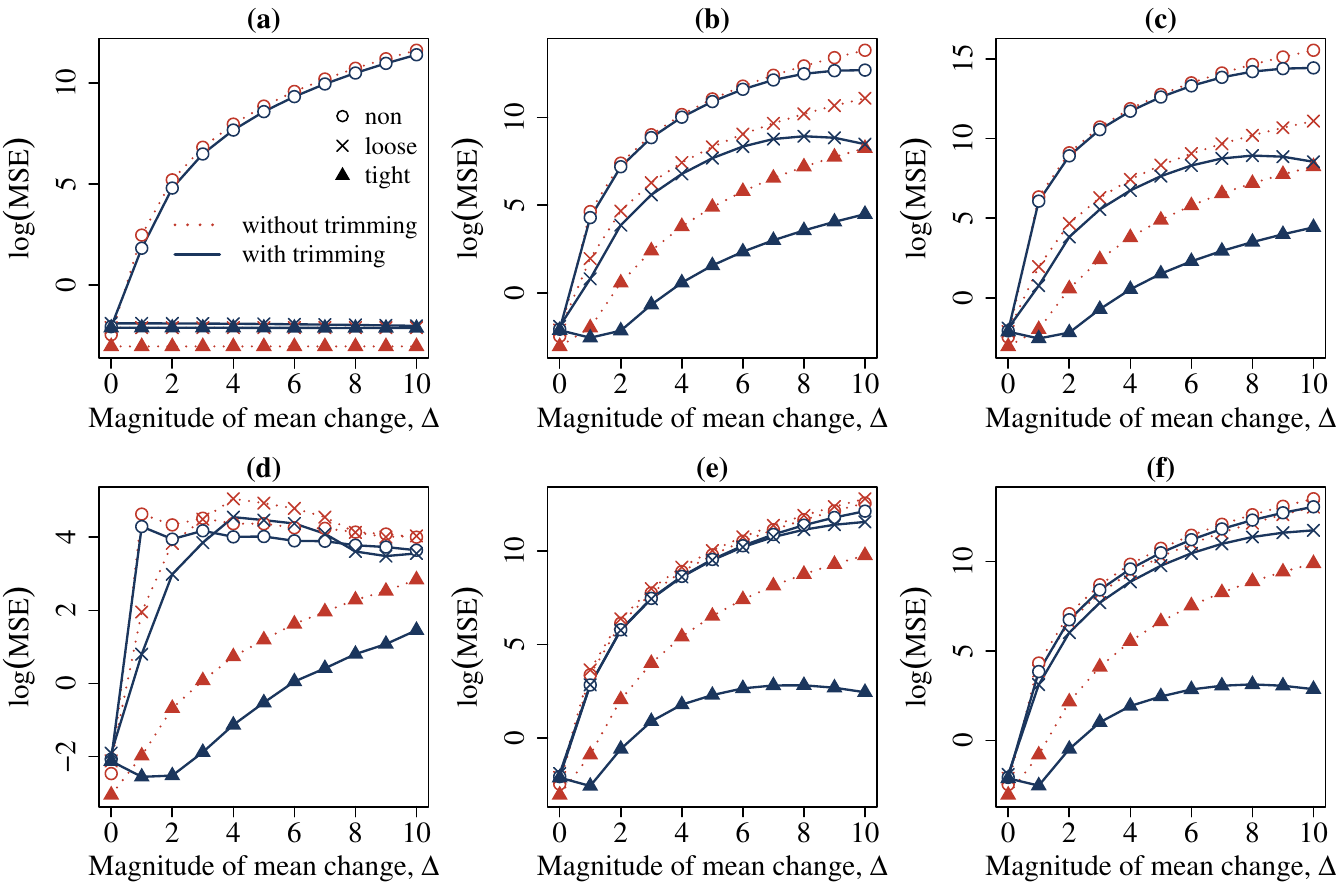} \vspace{-0.6cm}
		\caption{
		Plots (a)--(f) show the values of $\log\mathrm{MSE}(\cdot)$ for estimators of $f(0)$ 
		against $\Delta$, the magnitude of mean change, under the mean structures 
		$\mu_{\textsc{a}}, \ldots, \mu_{\textsc{f}}$, respectively; see \S\ref{sec:lrvEst_nonconstMean}.
		The estimators are $\hat{f}_{\mathrm{non}}(0)$ (hollow circle $-\bigcirc-$), $\hat{f}_{\mathrm{loose}}(0)$ (cross $-\bigtimes-$), 
		and $\hat{f}_{\mathrm{tight}}(0)$ (filled triangle $-\text{\large$\blacktriangle$}-$).
		The solid (\solidLine) and dotted (\dottedLine) lines correspond to trimmed and untrimmed versions.
        }
		\label{fig:robust} 
	\end{center} \vspace{-0.3cm}
\end{figure}

We assess the mean-robustness of estimators of $f(\theta)$. 
In particular, we consider 
\begin{gather*}
	\mu_{\textsc{a}}(t) = v^{1/2}\Delta t, \quad
	\mu_{\textsc{b}}(t) = v^{1/2}\Delta \mathbb{1}(t>0.5), \quad
	\mu_{\textsc{c}}(t)= \mu_{\textsc{a}}(t)+\mu_{\textsc{b}}(t), \\
	\mu_{\textsc{d}}(t) = v^{1/2}\sum_{j=1}^{\Delta} (-1)^{j+1} \mathbb{1}\{t>j/(\Delta+1)\},  \\
	\mu_{\textsc{e}}(t) = v^{1/2}\Delta\{0.4\mathbb{1}(t>0.25) - \mathbb{1}(t>0.5) + 0.8 \mathbb{1}(t>0.75)\}, \quad
	\mu_{\textsc{f}}(t) = v^{1/2}\Delta t^3 + \mu_{\textsc{e}}(t),
\end{gather*}
where $\Delta$ controls the slope, jump magnitude, or number of jumps. 
In all cases, $\Delta\in\{0,\ldots,10\}$ measures mean non-constancy and $v$ denotes the true long-run variance of the noises
generated from the $\ARMA(1,1)$ model of \S\ref{sec:lrvEst_constMean} with $\phi=0.6$ and $n=800$. 
We compare $\hat f_{\non}(\theta)$, $\hat f_{\loose}(\theta)$, and $\hat f_{\tight}(\theta)$
as well as their trimmed variants 
$\hat f_{\non}^{\Diamond}(\theta)$, $\hat f_{\loose}^{\Diamond}(\theta)$, and $\hat f_{\tight}^{\Diamond}(\theta)$
with $\alpha_n=\lfloor n^{q/(1+2q)}\rfloor$; see \S\ref{sec:CPremoval}. 
We set $q=3$.
All difference-based estimators use $m=3$.
Figure \ref{fig:robust} reports results for $\theta=0$. 
For $\Delta>0$, $\hat f_{\non}(\theta)$ deteriorates rapidly, even with trimming; $\hat f_{\loose}(\theta)$ shows some robustness but can be as poor as $\hat f_{\non}(\theta)$; see the cases of $\mu_{\textsc{e}}$ and $\mu_{\textsc{f}}$. 
The proposed $\hat f_{\tight}(\theta)$ is the most robust, with trimming offering additional small gains.

Additional simulation results for estimation of the full spectrum are presented 
in \S\ref{sec:full_spectrum_MCexp} of the supplement.  
We remark that tight differencing is computationally fast 
because it does not require data-dependent fitting of the trend or estimation of change points. 
Moreover, the optimal tight difference sequence is universal in the sense that it is free of any process-dependent quantities. 
The computation time is compared with other approaches, e.g., residual-based methods, in \S\ref{sec:computation}. 

\section{Applications} \label{sec:App}

\subsection{Testing for stationarity}\label{sec:KPSS} 
\begin{figure}[!t]  
\captionsetup{font=small}
	\begin{center}
		\includegraphics[width=\linewidth]{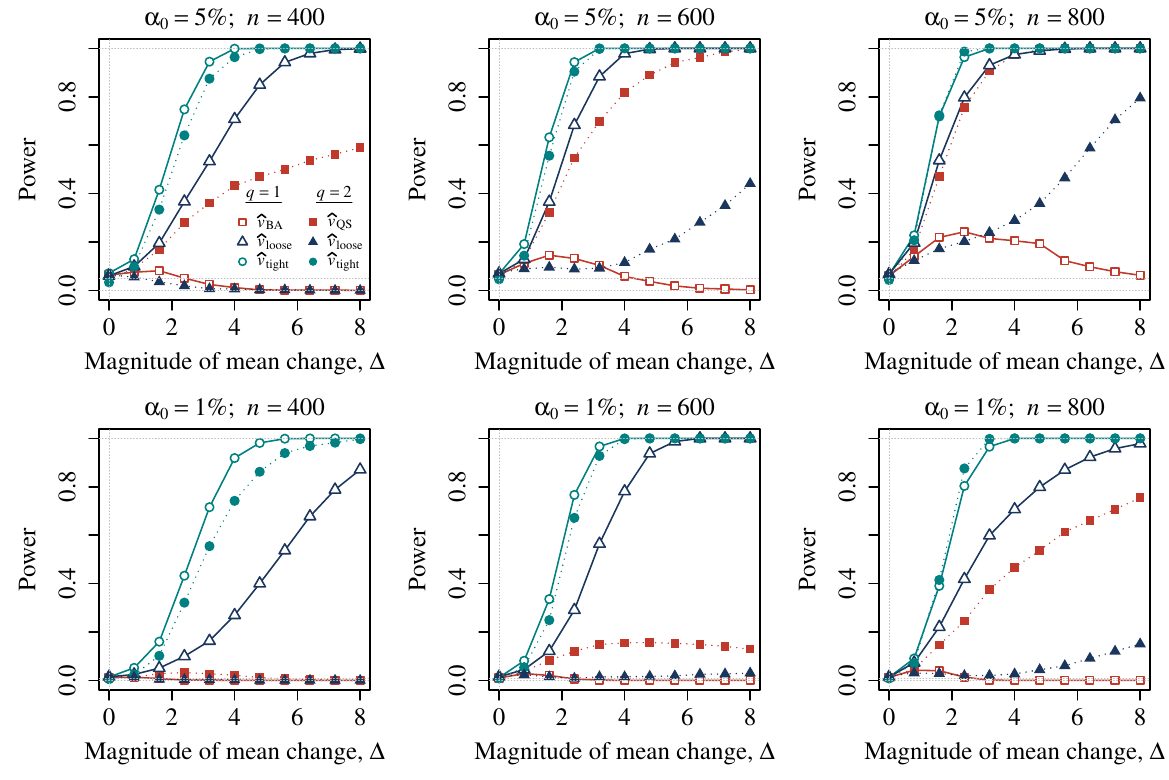}  \vspace{-0.6cm}
		\caption{The power curves of the stationarity tests based on 
		non-difference-based $\hat{v}_{\AB}$ and $\hat{v}_{\AQ}$, 
		loose-difference-based $\hat{v}_{\loose}$ with $q=1,2$, and 
		tight-difference-based $\hat{v}_{\tight}$ with $q=1,2$, 
		are plotted against $\Delta$,
		the magnitude of mean change; see \S\ref{sec:KPSS}.
		The upper and lower rows correspond to nominal size $\alpha_0 = 5\%$ and $1\%$, respectively. 
        }
		\label{fig:KPSS} 
	\end{center}  \vspace{-0.3cm}
\end{figure}

We use the stationarity test of \citet{Hobijn2004}, a generalization of \cite{KPSS1992}, testing
$H_0:\mu(t)\equiv0$ against $H_1:\mu(t)\not\equiv0$ with statistic
$T = n^{-2}\sum_{t=1}^{n}(\sum_{i=1}^{t} X_i)^2/\hat{v}$
where $\hat v$ estimates the long-run variance $v$.  
Critical values are simulated from independent $\Normal(0,1)$ with $10^6$ replications at nominal sizes $\alpha_0\in\{1\%,5\%\}$.
The data are generated with $\ARMA(4,4)$ noise, i.e., 
$
	Z_i = \sum_{j=1}^4 (a_j Z_{i-j} + b_j \varepsilon_{i-j} ) + \varepsilon_i,
$
and mean function
$
	\mu(t) = \Delta \{ \mathbb{1} (t>1/3) - 2\mathbb{1}(t>2/3) + t\sin(10\pi t)/2 \},
$
where 
$(a_1,\ldots,a_4)=(0.8,-0.6,0.4,-0.2)$, 
$(b_1,\ldots,b_4)=(0.7,0.5,0.3,0.1)$,
$\varepsilon_i\sim\Normal(0,1)$ independently, and 
$\Delta\ge 0$ controls the non-constancy. 
We consider $n\in\{400,600,800\}$.

Six estimators of $v$ are compared:
(i) non-difference based estimators based on Bartlett kernel ($\hat v_{\AB}$) and quadratic spectral kernel ($\hat v_{\AQ}$) 
with the automatic bandwidth of \citet{Hobijn2004};
(ii) $\hat v_{\loose}$ with $m=3$ and $K(t)=(1-|t|^{q})^{+}$ for $q\in\{1,2\}$; and
(iii) $\hat v_{\tight}$ with $m=3$ and $K=K_{\mathrm{s}}$ based on $K_0=(1-|t|^{q})^{+}$ and $q\in\{1,2\}$.

Figure \ref{fig:KPSS} shows the power curves.
Tests using $\hat{v}_{\tight}$ are the most powerful, 
with the largest gains for small $n$ or weak signals,
highlighting the practical value of tight differencing in small-samples or for detecting subtle non-constant signals.
In \S\ref{sec:realdata_supp}, 
we provide a real-data analysis of global mean sea level variations and test for several types of mean stationarity.
In \S\ref{sec:t_test_pw},  
we also study the mean test under strong persistence, 
and show that it has promising size control and power. 
It can also be used directly with the standard prewhitening technique \citep{AndrewsMonahan1992}; see also \citet{casini2024prewhitened,liuchan2026}.

\subsection{Testing for white noise}\label{sec:whiteness}
\begin{figure}[!t]  
\captionsetup{font=small}
	\begin{center}
		\includegraphics[width=\linewidth]{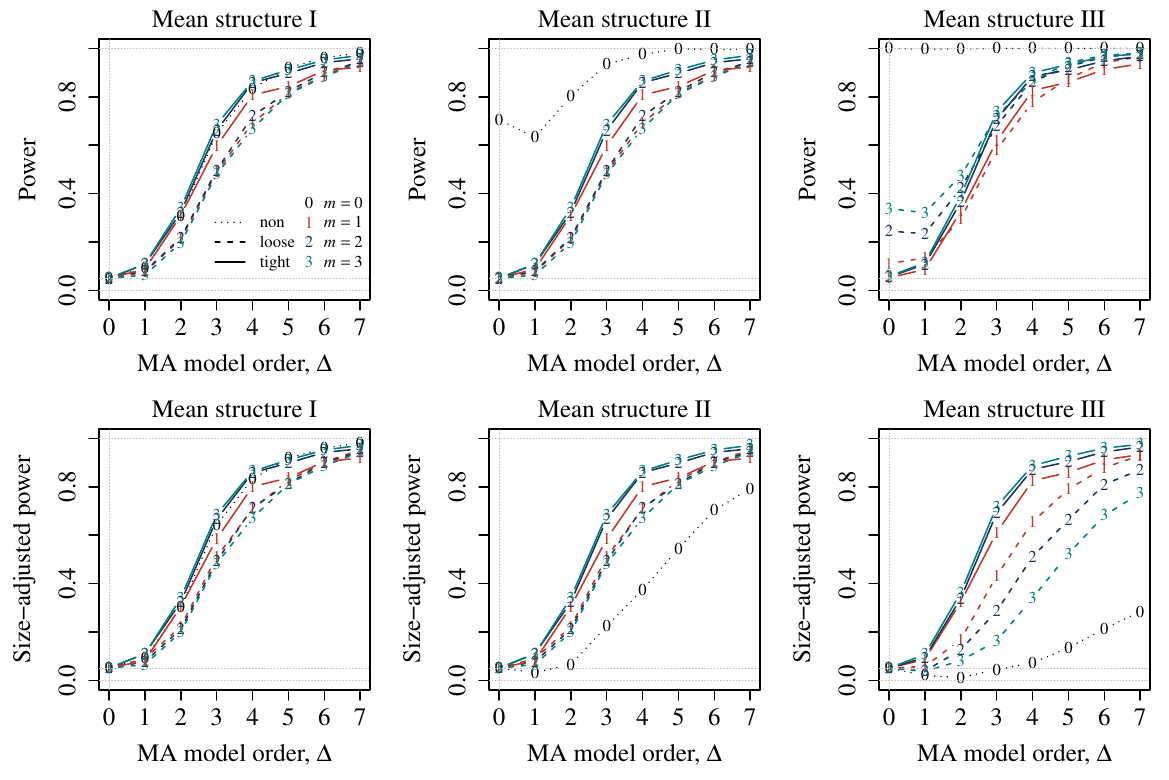}  \vspace{-0.6cm}
		\caption{Power and size-adjusted power curves of Hong's test at nominal size $\alpha_0=5\%$ using 
		different estimators; see \S\ref{sec:whiteness}.
		Dotted (\dottedLine), dashed (\dashedLine), and solid (\solidLine) lines represent non-differencing, loose differencing, and 
		tight differencing. 
		Numbers on the lines represent the order of differencing.
        }
		\label{fig:whiteness} 
	\end{center}  \vspace{-0.3cm}
\end{figure}

We test the white noise null hypothesis, i.e., 
$H_0:\gamma_k = 0$ for all $k\neq 0$  against $H_1$: $\gamma_k\neq 0$ for some $k\neq 0$.
\cite{hong1996} proposed a test based on the $\mathcal{L}^2$ distance between 
a kernel-based spectral density estimator 
$\hat{f}(\theta)$ and the white noise spectrum $f(\theta) = 1/(2\pi)$,
leading to the statistic
$T = 2\pi^2 \int_{-1/2}^{1/2} \{\hat{f}(\theta) - (2\pi)^{-1} \}^2\, \dd \theta$,  
which is a generalization of the portmanteau test by \cite{boxPierce1970}. 
\cite{shao2011} extended it to allow unknown weak dependence under $1/\ell + \ell/n \rightarrow 0$.
Our goal is to test for white noise under possibly non-constant means. 

We compute Hong's statistic $T$ based on 
$\hat{f}_{\non}(\theta)$ with $K=K_0$, 
$\hat{f}_{\loose}(\theta)$ with $m=1,2,3$ and $K = K_0$, and 
$\hat{f}_{\tight}(\theta)$ with $m=1,2,3$ and $K = K_{\textsc{s}}$, 
where $K_0(t) = (1-|t|^q)^+$ for $q=1,2$.  
Critical values are simulated with independent $\Normal(0,1)$ samples over $10^6$ replications. 
We generate noises from an order $\Delta$ non-linear moving average model:
$Z_i = \sum_{j=1}^{\Delta} b_j \varepsilon_{i-j}\varepsilon_{i-j-1} + \varepsilon_i$, 
where $b_j = 3/\{4(1+j^{1/2})\}$ and $\varepsilon_i$ follows standard Laplace distribution independently.
We set $n=400$ and test at $\alpha_0=5\%$ nominal size. 
Three mean structures are considered: (I) $\mu(t)=0$, (II) $\mu(t) = t$, and (III) $\mu(t) = \mathbb{1}(t>{1}/{2})$.
The power and size-adjusted power curves are shown in Figure \ref{fig:whiteness}.
With a constant mean, the power of the tests using $\hat{f}_{\non}(\theta)$ and $\hat{f}_{\tight}(\theta)$
are similar, whereas the power of $\hat{f}_{\loose}(\theta)$ is obviously lower. 
Under the linear increasing mean structure II, 
the traditional test based on $\hat{f}_{\non}(\theta)$ becomes invalid as the type-I error rate reaches 70\%.
The tests based on $\hat{f}_{\loose}(\theta)$ and $\hat{f}_{\tight}(\theta)$ remain valid, 
with higher power achieved by $\hat{f}_{\tight}(\theta)$. 
Under the change point structure III, only the test using $\hat{f}_{\tight}(\theta)$ is valid and powerful.

\vspace{-0.4cm}
\section*{Acknowledgment}
We thank the editor, an associate editor and three referees for their valuable comments and
suggestions that helped to improve the paper.
This research was partially supported by grants GRF-14306421 and 14307922
provided by the Research Grants Council of HKSAR.

\vspace{-0.4cm}
\section*{Supplementary material}
The online supplement includes technical proofs of theorems and additional results. 
The R package \texttt{"tightDiff"} 
contains functions for implementing the
proposed methods.

\appendix 

\renewcommand{\theequation}{\thesection.\arabic{equation}}
\renewcommand{\thefigure}{\thesection.\arabic{figure}}
\renewcommand{\thetable}{\thesection.\arabic{table}}
\renewcommand{\thesection}{\Alph{section}}

\section{Additional theoretical results and examples} 

\subsection{Commonly used kernels} \label{sec:kernel_oce_bce}
Table \ref{table:kernel_oce_bce} reports the origin and boundary characteristic exponents,
$\CE(K)$ and $\BCE(K)$, as defined in
Definition \ref{def:CE_BCE}, for several commonly used kernels. 
Throughout, the notation $(x)^{+}=\max(x,0)$ denotes the positive-part operator. 
\begin{table}[!b]
\small
\def~{\hphantom{0}}
\setlength{\tabcolsep}{4pt}
\centering
\caption{Origin and boundary characteristic exponents for some commonly used kernels.}
\label{table:kernel_oce_bce}
\begin{tabular}{llccc}
\hline
Kernel &$K(t)$ & $\CE(K)$ & $\BCE(K)$ & $q$ v.s. $q^{\prime}$ \\
\hline
Triweight &$\{(1-t^{2})^{3}\}^{+}$ 
  & $2$ & $3$ & $q<q'$ \\
Parzen &$(1-6t^2+6|t|^3)\mathbb{1}(|t|\le\frac12)+2(1-|t|)^3\mathbb{1}(\frac12<|t|\le1)$
  & $2$ & $3$ & $q<q'$ \\
[0.3em]
Bartlett &$(1-|t|)^{+}$ 
  & $1$ & $1$ & $q=q'$ \\
Tukey--Hanning &$\{1+\cos(\pi t)\}\mathbb 1(|t|\le 1)/2$
  & $2$ & $2$ & $q=q'$ \\
[0.3em]
Quadratic &$(1-t^{2})^{+}$ 
  & $2$ & $1$ & $q>q'$ \\
Tricube &$(1-|t|^{3})^{+}$ 
  & $3$ & $1$ & $q>q'$ \\
\hline
\end{tabular}
\end{table}

This table illustrates that the condition $q \le q^{\prime}$ holds when the kernel is at least as
flat near the boundary $|t|=1$ as it is near the origin $t=0$, in the sense of Definition \ref{def:CE_BCE}. 
Many commonly used kernels, such as the quadratic and tricube kernels, violate this condition.
By contrast,
the proposed centrosymmetric kernels always have $q=q^{\prime}$ by construction, ensuring satisfactory performance of the tight-difference-based kernel estimator.

\subsection{Sub-tight differencing: the intermediate regime $1<\lambda<2$} \label{sec:sub-tightTheory}
This section studies the intermediate regime $\ell < h < 2\ell$, or equivalently $\lambda \in (1,2)$, 
which we refer to as the {sub-tight differencing} regime.
We show that, although consistency holds and rate optimality is attainable, the $\MSE$-optimal difference sequence is no longer given by the classical difference sequence for independent data \citep{Hall1990}, nor by the corresponding sequence for dependent data under loose differencing \citep{Chan2022}. Instead, it depends explicitly on the choice of the kernel $K$, similarly to the tight-differencing case $h=\ell$.
Thus, it is not simpler than the tight regime. 

We begin by characterizing the effective kernel induced by differencing in this regime.
When $\lambda=h/\ell\in(1,2)$, the effective kernel $K_d$ admits a piecewise representation 
\begin{align} \label{eqt:Kd_subtight}
K_d(t)
=
\sum_{s=0}^m \delta_s K(|t|-s\lambda)
\end{align}

We now state the corresponding bias and variance results for $\hat v$ when $\lambda=h/\ell \in (1,2)$ under a constant mean.
Denote $\hat{v}_{\subtight}= \hat{v}(K,\ell,d_{0:m},\lambda)$ with $\lambda\in (1,2)$.
\begin{theorem}[Bias and variance under constant mean for $\hat{v}_{\subtight}$]\label{thm:bv_lambda_12}
Suppose Assumption~1 holds. Let $K\in\mathcal K_{q,q'}$ with $q,q'\in\mathbb N$ and assume $u_q<\infty$, $m\in\mathbb N$, and $\lambda=h/\ell\in(1,2)$.
Assume $1/\ell+\ell/n=o(1)$ as $n\to\infty$.
If $\mu_1=\cdots=\mu_n$, then
\[
{\Bias}_0(\hat{v}_{\subtight})
=\frac{Bv_q}{\ell^q}+O\left(\frac{\ell}{n}\right)+o\left(\frac{1}{\ell^q}\right) 
\quad \text{and} \quad
{\Var}_0(\hat{v}_{\subtight})
=\frac{4\tilde A_d\ell v^2}{n}+o\left(\frac{\ell}{n}\right),
\]
where $B$ is defined in \eqref{eqt:BBprime} and 
\begin{align}\label{eqt:tildeA_d}
\tilde{A}_d = \sum_{|s|\leq m}\kappa_{2,s} \delta_s^2 + 
 \sum_{s=1}^{m} \left \{\kappa_{1,s-1}^{(+)} + \kappa_{1,-s}^{(+)} + \kappa_{1,s}^{(-)} + \kappa_{1,-(s-1)}^{(-)}  \right \}\delta_{s-1}\delta_s,
\end{align}
with $L_s = \max(-\lambda s,-1)$, $U_s=\min\{1+(m-s)\lambda,1\}$, and
\[
\kappa_{2,s}=\int_{L_s}^{U_s} K(u)^2 \dd u,
\quad
\kappa_{1,s}^{(+)}
=\int_{L_s}^{U_s} K(u)K(u+\lambda)\dd u,
\quad
\kappa_{1,s}^{(-)}
=\int_{L_s}^{U_s} K(u)K(u-\lambda)\dd u.
\]
\end{theorem}
The proof of Theorem \ref{thm:bv_lambda_12} can be found in \S\ref{pf:thm_bv_lambda_12}.
Theorem \ref{thm:bv_lambda_12} shows that when $\lambda\in(1,2)$, 
the bias and variance orders coincide with those in the loose-differencing regime $\lambda\ge2$, since the effective kernel preserves the near-origin order of $K$,
thus rate-optimality is attainable for general $K$.
However, as in the case of $\lambda=1$, the \textsc{mse}-optimal difference sequence depends on the choice of kernel, 
so the optimality is not guaranteed and the optimal $d_{0:m}$ is again different from the classical difference sequence given by \citet{Hall1990}. 
Consequently, the resulting procedure is not simpler than tight differencing. 

We then derive the optimal $\ell$ and $d_{0:m}$ for $\hat v_{\subtight} = \hat v(K, \ell, d_{0:m}, \lambda)$ with $\lambda \in (1,2)$. 
The optimization strategy follows exactly the same steps as in \S\ref{sec:Op_tightD} for the tight-differencing case $\lambda=1$. 
In particular, we again set
$\ell=\xi n^{1/(1+2q)}$ with $\xi>0$ and minimize 
\[
{\MSE}_0(\hat{v}_{\subtight}) ={\Bias}_0^2(\hat{v}_{\subtight})+{\Var}_0(\hat{v}_{\subtight}) 
	\sim \left( B v_q / \ell^{q} \right)^2 +  4 \tilde A_d \ell v^2 / n, 
\]
with respect to $\xi$ and
$d_{0:m}$.

\begin{proposition}[Optimal parameters for $\hat{v}_{\subtight}$]\label{prop:opt_lambda_12}
Assume the conditions of Theorem~\ref{thm:bv_lambda_12} hold. 
Let $B$ and $\tilde A_d$ be the bias and variance constants in Theorem~\ref{thm:bv_lambda_12}.
If $v_q\neq0$, then
\begin{align}\label{eqt:opt_par_subtight}
\bar{\xi}^*_{\subtight}(K,d_{0:m})
= \left\{
\frac{(B v_q / v)^2 q}{2\,\tilde A_d}
\right\}^{1/(1+2q)}, \quad \text{and} \quad
d_{\subtight,\,0:m}^*(K)
= \argmin_{d_{0:m}\in\mathcal D_m} \tilde A_d.
\end{align}
Consequently, the asymptotic~\textsc{mse}-optimal bandwidth is given by
\begin{align}\label{eqt:opt_bw_subtight}
\ell^*_{\subtight}(K)=\bar{\xi}^*_{\subtight}(K,d_{\subtight,\,0:m}^*(K))\,n^{1/(1+2q)}.
\end{align}
\end{proposition}

Proposition~\ref{prop:opt_lambda_12} shows that the asymptotic~\textsc{mse}-optimal difference sequence
$d_{\subtight,\,0:m}^*(K)$ can be computed numerically once the kernel $K$ is specified, following
the same optimization strategy as in the tight-differencing case $\lambda=1$ stated in \S\ref{sec:more_dj_opt}.
Consequently, the resulting procedure is not simpler than tight differencing, and
using $\lambda$ in this intermediate sub-tight regime does not lead to a kernel-independent or closed-form choice of optimal $d_{0:m}$, 
while the demeaning robustness is weaker than at $\lambda=1$; see Example \ref{eg:1lambda2_robust} below.
For this reason, we advocate tight differencing with $\lambda=1$ in practice.

\begin{example}[Robustness comparison for $1\leq \lambda \leq 2$] \label{eg:1lambda2_robust}

\begin{figure}[!t]
\centering
\includegraphics[width=\textwidth]{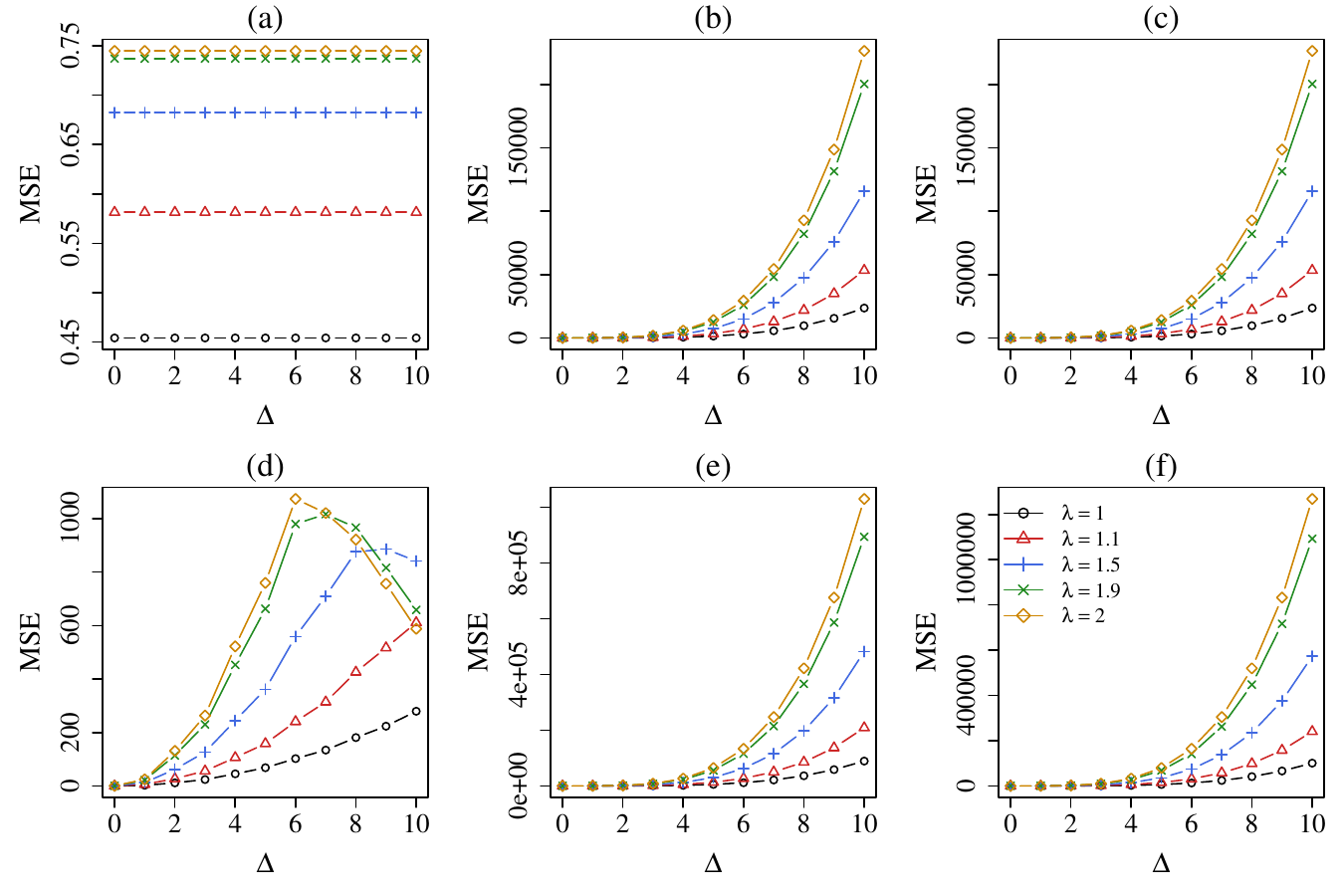}
\captionsetup{font=small}
\caption{The $\MSE(\cdot)$ of estimators are plotted against the magnitude $\Delta$ of mean change 
		under model structures $\mu_{\textsc{a}}, \ldots, \mu_{\textsc{f}}$; see \S\ref{sec:lrvEst_nonconstMean}.
Line colors indicate $\lambda$ value used in the estimators.}
\label{fig:mc_intermediateLam_m1_bartlett}
\end{figure}
We assess the mean-robustness of estimators of $v$.
In particular, we consider the same six mean families $\mu_{\textsc{a}},\ldots,\mu_{\textsc{f}}$ in \S\ref{sec:lrvEst_nonconstMean}, 
in which $v$ denotes the true long-run variance of the noises and $\Delta\in\{0,1,\ldots,10\}$ determines the mean non-constancy.
In all cases, noises are generated from the $\ARMA(1,1)$ model of \S\ref{sec:lrvEst_constMean} with $\phi=0.6$ and $n=800$. 
We compare estimators under different lag-to-bandwidth ratios, i.e., 
$\hat v(K, \ell, d_{0:1}, \lambda)$ with $q=3$, $m=2$
and $\lambda = 1,1.1,1.5,1.9,2$.
Figure \ref{fig:mc_intermediateLam_m1_bartlett} shows that 
for each mean family, smaller $\lambda$ yields generally smaller MSE over $\Delta$, with the gap becoming more pronounced as the mean non-constancy increases.
This provides additional numerical evidence that shortening the differencing scale, or equivalently reducing lag-to-bandwidth ratio $\lambda$, can markedly enhance the robustness to time-varying means.
\end{example}

In a nutshell, 
while sub-tight differencing avoids the need for centrosymmetric kernels, 
the optimal difference sequences lose their elegant, tractable forms due to complex, irregular index overlaps. 
Thus, tight-differencing ($\lambda=1$) with a centrosymmetric kernel is actually simpler to implement optimally. 

\subsection{Inconsistent differencing: the regime $\lambda <1$} \label{sec:lambda_smaller1}
Consistency of the difference-based kernel estimator requires that the effective kernel $K_d$ satisfies 
the kernel normalization condition
\begin{equation} \label{eq:Kd_consistent_req}
	K_d(0)=1,
\end{equation}
mirroring the condition $K(0)=1$ in \eqref{eqt:kernel}.
This condition effectively assigns a weight of $1$ to the estimator of $\gamma_0$
when estimating $v = \gamma_0 + 2\sum_{k=1}^{\infty}\gamma_k$.
Hence, (\ref{eq:Kd_consistent_req}) is a necessary condition for consistency for $\hat{v}$. 

We clarify here why this property holds automatically when $\lambda\ge1$
but generally fails when $\lambda<1$.
By the definition of $K_d(t)$ in \eqref{eqt:Kd} and 
the property that $K(0)=1$,
we have
\begin{equation}
	K_d(0)
	= \sum_{s=\lceil-1/\lambda\rceil}^{\lfloor 1/\lambda\rfloor}
	\delta_s K(\lambda s)
	= 1 + 2\sum_{s=1}^{r} \delta_s K(\lambda s),
	\label{eq:Kd0_formula}
\end{equation}
where $r=\min\left(m,\lfloor 1/\lambda\rfloor\right)$.
Then we consider two cases:
\begin{itemize}
	\item When $\lambda\ge1$, we have $|\lambda s|\ge1$ for all $s\ne0$.
For $K \in \mathcal K$ defined in \eqref{eqt:kernel}, $K(\lambda s)=0$ for all $1\leq s \leq r$.
So, according to \eqref{eq:Kd0_formula}, we have $K_d(0)=1$.
Thus, the kernel normalization condition (\ref{eq:Kd_consistent_req}) is preserved automatically.
	\item When $\lambda<1$, however, $|\lambda s|<1$ for some $1\leq s \leq r$.
According to \eqref{eq:Kd0_formula}, possibly non-zero terms, i.e., $\delta_s K(\lambda s)$, contribute to $K_d(0)$. 
Since $\sum_{s=1}^{r} \delta_s K(\lambda s)$ is not guaranteed to be zero, 
the value $K_d(0)$ may not be one in general. 
In particular, for $\lambda \in [1/2,1)$, we get
\[
K_d(0)
= 1 + 2\delta_1 K(\lambda),
\]
which generally differs from one.
For commonly used difference sequences like 
\citet{Hall1990}'s optimal difference sequence and the binomial difference sequence $d_j= {\binom{m}{j}(-1)^j}/\binom{2m}{m}^{1/2}$ for $j=1,\ldots,m$, 
we have $\delta_1<0$.
Thus, in these typical cases, we have $K_d(0)<1$.
\end{itemize}

\begin{figure}[t] 
\centering
\captionsetup{font=small}
\includegraphics[width=\textwidth]{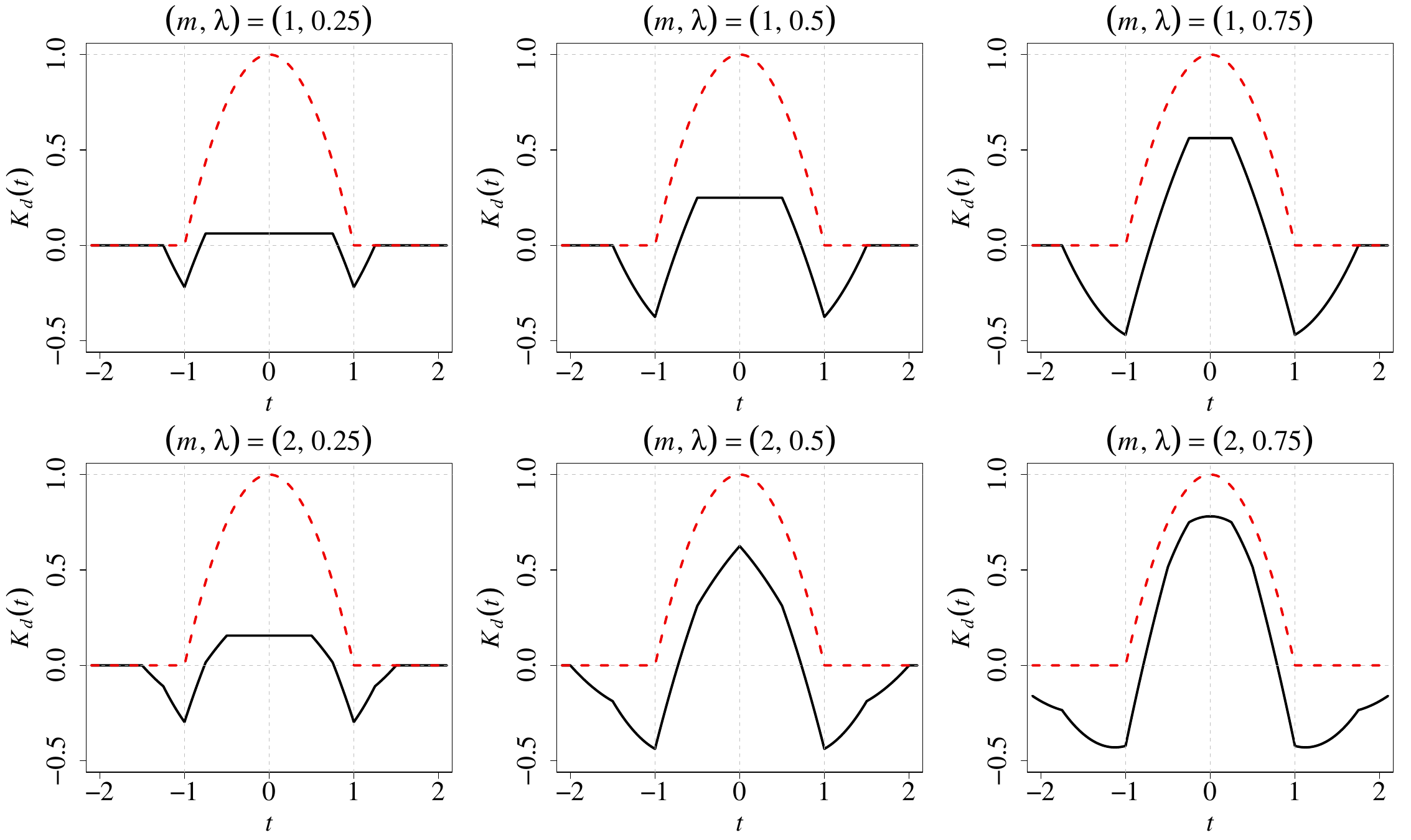}
\caption{Plots of the effective kernels $K_d$ ({\color{black}\solidLine}) 
with \citet{Hall1990}'s difference sequence for 
$m\in\{1,2\}$
and $\lambda\in\{0.25,0.5,0.75\}$.
The Bartlett kernel $K(t) = (1-|t|)^+$ ({\color{red}\dashedLine}) is used as the base kernel. }
\label{fig:Kd_small_lambda}
\end{figure}

Figure \ref{fig:Kd_small_lambda} illustrates this effect for the Bartlett kernel with \citet{Hall1990}'s optimal difference sequence.
Across all panels with $\lambda<1$,
the peak of $t\mapsto K_d(t)$ at $t=0$ lies strictly below one.
This confirms that the normalization condition
$K_d(0)=1$ fails in this regime.

We further verify this inconsistency issue through a simple Monte Carlo experiment. 
We generate $X_1, \ldots, X_n \simIID \Normal(0,1)$.
The true long-run variance is $v=1$.
We compute the estimator $\hat v (K, \ell, d_{0:m}, \lambda)$
with Bartlett kernel $K(t) = (1-|t|)^+$, $m=1$, and thus, $d_{0:1} = \{1/\sqrt{2}, -1/\sqrt{2}\}$.
We consider $\lambda\in\{1/4,1/2,3/4,1,2\}$ and sample size $n=64,128,\ldots,1024$.  

\begin{figure}[t]
\centering  
\captionsetup{font=small}
\includegraphics[width=0.55\textwidth]{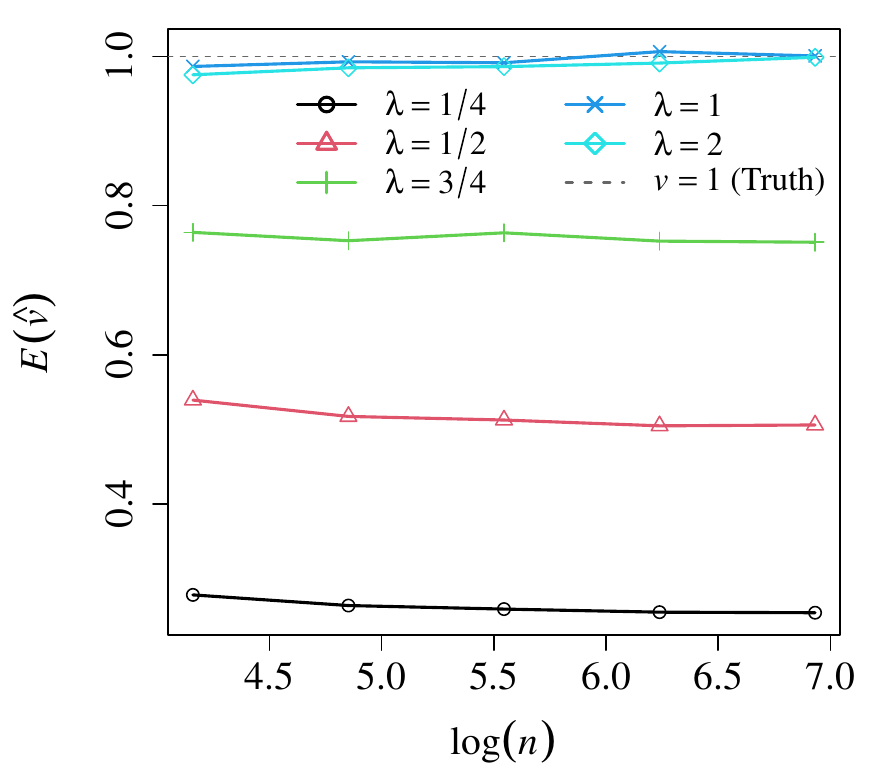}
\caption{Monte Carlo mean of $\hat v(K,\ell,d_{0:1},\lambda)$
under $X_i\simIID \Normal(0,1)$,
with $\ell=\lfloor 2 n^{1/3}\rfloor$
and $\lambda\in\{1/4,1/2,3/4,1,2\}$.
The dashed line marks the true value $v=1$.}
\label{fig:MC_lambda_small}
\end{figure}

Figure \ref{fig:MC_lambda_small} plots the values of $\E(\hat v)$ against $\log(n)$.
When $\lambda \ge 1$, the estimator is centered around the true value $v=1$.
In contrast, when $\lambda<1$, the values of $\E(\hat{v})$ converge to values strictly below one.
It indicates that the bias does not diminish as $n\rightarrow\infty$.
These results confirm that when $\lambda<1$, 
the kernel normalization condition does not hold, and thus, the estimator $\hat{v}$ is inconsistent;
see also \citet{Chan2022} for similar arguments.

\subsection{Quadratic spectral kernel under tight differencing} \label{sec:QS_kernel}

The quadratic spectral (\QSK) kernel,
introduced by \citet{andrews1991heteroskedasticity},
is defined by
\[
K_{\QSK}(t)
=\begin{cases}
1, & t=0, \\
\displaystyle
\frac{25}{12\pi^2 t^2}
\left\{\frac{\sin(6\pi t/5)}{6\pi t/5}-\cos(6\pi t/5) \right\},
& t\neq 0.
\end{cases}
\]
It is widely regarded as a well-performing kernel in the classical
heteroskedasticity and autocorrelation consistent
covariance matrix estimation literature.
However, it does not belong to the class $\mathcal K$
defined in \eqref{eqt:kernel},
since $\mathcal K$ consists of symmetric kernels with compact support,
whereas $K_{\QSK}$ has unbounded support.

For a general kernel $K$ and difference sequence $d_{0:m}$,
the effective kernel is
\[
	K_d(t)
	= \sum_{|s|\le m}
	\delta_s K(t+\lambda s).
\]
Evaluating at $t=0$ under tight differencing ($\lambda=1$) yields
\begin{equation}
\label{eq:Kd0_tight_supp}
K_d(0)
= \delta_0 + 2\sum_{s=1}^m \delta_s K(s).
\end{equation}
For kernels in $\mathcal K$,
compact support implies $K(s)=0$ for $s\ge 1$,
so that $K_d(0)=\delta_0=1$,
ensuring the required normalization.
For the \QSK~kernel, however, $K_{\QSK}(s)\neq 0$ for $s\ge 1$,
and the second term in \eqref{eq:Kd0_tight_supp} does not necessarily vanish.
Consequently, in general $K_d(0)\neq 1$.
The resulting difference-based estimator $\hat{v}$ therefore need not assign a unit weight to $\gamma_0$,
and the normalization condition required for consistency under tight differencing may fail,
analogously to the phenomenon discussed in \S\ref{sec:lambda_smaller1}.
Figure \ref{fig:QS_Kd} displays the effective kernel obtained by applying
tight differencing to $K_{\QSK}$ for $m=1,2,3$.

\begin{figure}[t]
\centering
\includegraphics[width=\textwidth]{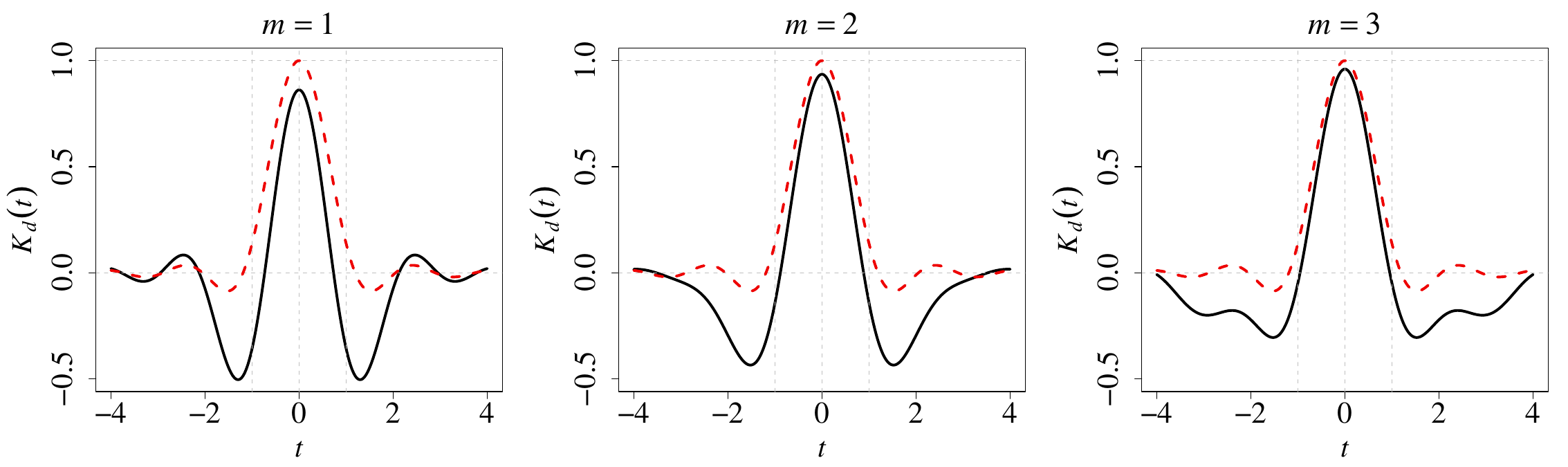}
\captionsetup{font=small}
\caption{ 
The effective kernel $K_d$ ({\color{black}\solidLine}) induced by the \QSK~kernel $K_{\QSK}$ ({\color{red}\dashedLine})
under tight differencing ($\lambda=1$) for $m=1,2,3$.
\citet{Hall1990}'s difference sequence is used. 
}
\label{fig:QS_Kd}
\end{figure}

\subsection{More examples of centrosymmetric kernels}\label{sec:general_t0}
Revisit the smooth centrosymmetric kernels $K_{\text{s}}(\cdot)$ defined in Example \ref{ex:Ks}, 
which fixes $t_0=1/3$. 
In this section, we present the result with a general $t_0 \in (0,1/2]$. 
After straightforward calculation, we obtain 
$H(t)= \sum_{j=0}^3 \alpha_j(t_0) |t|^j  \mathbb{1} (t_0 \leq |t| \leq 1-t_0)$, where 
\begin{align*}
\alpha_3(t_0) &= \{4K_0(t_0)+2(1-2t_0)K_0'(t_0)-2\}/(1-2t_0)^3, \\
\alpha_2(t_0) & = -3\alpha_3(t_0)/2, \\
\alpha_1(t_0) & = \{3-6K_0(t_0)-(1-2t_0)K_0'(t_0)\}/(2-4t_0) + 3\alpha_3(t_0)/4, \\
\alpha_0(t_0) & = 1/2-\alpha_1(t_0)/2+\alpha_3(t_0)/4.
\end{align*}
Putting $t_0=1/3$ into the above formulas, we get back the special case in Example \ref{ex:Ks}. 
It demonstrates that, for any $t_0 \in (0,1/2]$, there always exists $H$ 
for defining the centrosymmetric kernels in (\ref{eqt:K_transformed}).

In this case, for any $q\in\mathbb{N}$, we have $B=B'=-1$ and $q'=q$.  
Consequently, the variance constant $A_d$ and bias constant $B_d$ in (\ref{eqt:ABd}) can be found as 
\[
       A_d = \kappa_1 \left(2\sum_{s=1}^{m}\delta_{s-1}\delta_s \right)+ 
                \kappa_2 \left(1+2\sum_{s=1}^m \delta_s^2 \right)
            \qquad \text{and} \qquad
       B_d = \delta_1  -1, 
\]
where $\kappa_1 = 1/2-\kappa_2$ and
\begin{align*}
	\kappa_2 &= \frac{1}{210} \Bigg[ 
		78 + 54 t_0 
		- 9qt_0^{q-1} + 6(6q - 17)t_0^q - \frac{12(3q^2 - 14q + 18)}{q+1}t_0^{q+1} \\
		&\qquad+q^2t_0^{2q-2} + 6q(3-q)t_0^{2q-1} + 6(2q^2 - 12q + 17)t_0^{2q} 
		- \frac{8(q-1)(q-3)(2q-9)}{2q+1}t_0^{2q+1} \Bigg].
\end{align*}
Since $B_d$ is independent of $t_0$, 
the asymptotic bias of $\hat{v}_{\tight}$ does not vary with $t_0$. 
Although $A_d$ depends on $t_0$, it does not vary significantly. 
For example, consider $q=2$; 
if we select $t_0 \in [0.1, 0.4]$, ensuring that $H$ is not inserted too close to the boundaries, 
then $0.388 < \kappa_2 < 0.422$. Consequently, $A_d$ remains approximately the same.

\subsection{Numerical comparison of asymptotic bias and variance}
\label{sec:paraVals_thm1}

To illustrate how the quantities appearing in the asymptotic bias and variance expressions in Theorem \ref{thm:MSE_propsal_lrv}
influence the estimator $\hat{v}$,
we consider three representative kernels:
\begin{itemize}
	\item the Bartlett kernel $K_{\BaK}(t)=(1-|t|)^+$, where $K_{\BaK} \in \mathcal K_{1,1}$,
	\item the smooth centrosymmetric quadratic kernel $K_{\text{s}}\in\mathcal K_{2,2}$ defined in Example \ref{ex:Ks} with
$K_0(t)=(1-t^2)^+$, and 
	\item the Tukey--Hanning kernel
$K_{\THK}(t)=\{1+\cos(\pi t)\}\mathbb{1}(|t|\le1)/2$, where $K_{\THK} \in \mathcal K_{2,2}$.
\end{itemize}
For each kernel we compute the effective kernel constants $A_d$ and $B_d$ together with the 
asymptotic leading bias and variance of $\hat v_{\tight}$ given in Theorem \ref{thm:MSE_propsal_lrv}.
Throughout this section we consider the $\AR(1)$ model
\[
X_i=\phi X_{i-1}+\varepsilon_i,
\qquad
\varepsilon_i \simIID \Normal(0,1),
\]
for which the values of $v$, $v_1$, and $v_2$ are 
\[
v=\frac{1}{(1-\phi)^2},
\qquad 
v_1=\frac{2\phi}{(1-\phi)^3(1+\phi)}, 
\qquad\text{and}\qquad
v_2= \frac{2\phi}{(1-\phi)^4}, 
\]
respectively. 
These quantities enter the leading bias and variance terms through $B_d v_q$ and $A_d v^2$, respectively,
where $q = \CE(K)$.

Table \ref{tab:phi05_kernel_compare} illustrates the effect of the differencing order by fixing $\phi=0.5$ and 
varying the order of differencing $m=1,2,3,4$.
For each kernel 
$K\in\{K_{\BaK}, K_{\text{s}}, K_{\THK}\}$, 
we report the constants $A_d$ and $B_d$, together with the quantities determining the leading bias and variance terms.
The table shows that increasing $m$ typically reduces the magnitude of $B_d$ while only slightly affects $A_d$.

Table \ref{tab:m2_kernel_compare} illustrates the role of the dependence strength by fixing $m=2$ and varying the $\AR(1)$ coefficient $\phi$.
While the kernel-specific constants $A_d$ and $B_d$ remain unchanged, the quantities $v$ and $v_q$ explode as $\phi\to1$, which substantially enlarges both the bias and variance terms.
This highlights how stronger serial dependence magnifies the constants appearing in the asymptotic mean squared error.

\begin{table}[t]
\centering
\captionsetup{font=small}
\caption{Comparison of kernel parameters and leading bias and variance constants for $\phi=0.5$ and $m=1,2,3,4$. 
The $\textsc{mse}$-optimal tight difference sequences are used. 
For each kernel $K$, the value of $v_q$ is evaluated at its kernel-specific order $q=\mathrm{OCE}(K)$.}
\begin{tabular}{lrrrrrrrr} \label{tab:phi05_kernel_compare}
Kernel $K$ & $q$ & $m$ & $A_d$ & $B_d$ & $|B_d|v_q$ & $A_d v^2$ & $v$ & $v_q$ \\
$K_{\BaK}$ & $1$ & $1$ & $0.33$ & $-1.50$ & $8.00$ & $21.33$ & $4$ & $5.33$ \\
 &  & $2$ & $0.34$ & $-1.28$ & $6.84$ & $22.04$ & $4$ & $5.33$ \\
 &  & $3$ & $0.33$ & $-1.22$ & $6.51$ & $21.34$ & $4$ & $5.33$ \\
 &  & $4$ & $0.33$ & $-1.17$ & $6.26$ & $21.29$ & $4$ & $5.33$ \\
$K_{\text{s}}$ & $2$ & $1$ & $0.54$ & $-1.50$ & $24.00$ & $34.51$ & $4$ & $16.00$ \\
 &  & $2$ & $0.49$ & $-1.24$ & $19.85$ & $31.35$ & $4$ & $16.00$ \\
 &  & $3$ & $0.47$ & $-1.16$ & $18.55$ & $29.92$ & $4$ & $16.00$ \\
 &  & $4$ & $0.46$ & $-1.12$ & $17.84$ & $29.21$ & $4$ & $16.00$ \\
$K_{\THK}$ & $2$ & $1$ & $0.44$ & $-3.70$ & $59.22$ & $28.00$ & $4$ & $16.00$ \\
 &  & $2$ & $0.41$ & $-3.17$ & $50.75$ & $26.53$ & $4$ & $16.00$ \\
 &  & $3$ & $0.40$ & $-3.01$ & $48.16$ & $25.45$ & $4$ & $16.00$ \\
 &  & $4$ & $0.39$ & $-2.91$ & $46.63$ & $24.95$ & $4$ & $16.00$ 
\end{tabular}
\end{table}

\begin{table}[t]
\centering
\captionsetup{font=small}
\caption{Comparison of kernel parameters and leading bias and variance constants for $\lambda=1$ and $m=2$ 
under different values of $\phi$. 
The $\textsc{mse}$-optimal tight difference sequences are used. 
Here $q=\mathrm{OCE}(K)$.}
\label{tab:m2_kernel_compare}
\begin{tabular}{lrrrrrrrr}
Kernel $K$ & $q$ & $\phi$ & $A_d$ & $B_d$ & $|B_d|v_q$ & $4A_d v^2$ & $v$ & $v_q$ \\
$K_{\BaK}$ 
& $1$ & $0$   & $0.34$ & $-1.28$ & $0.00$ & $1.38$ & $1.00$ & $0.00$ \\
&  & $0.3$ & $0.34$ & $-1.28$ & $1.73$ & $5.74$ & $2.04$ & $1.35$ \\
&  & $0.6$ & $0.34$ & $-1.28$ & $15.03$ & $53.82$ & $6.25$ & $11.72$ \\
&  & $0.9$ & $0.34$ & $-1.28$ & $1.21\times 10^3$ & $1.38\times 10^4$ & $100.00$ & $9.47\times 10^2$ \\
$K_{\text{s}}$
& $2$ & $0$   & $0.49$ & $-1.24$ & $0.00$ & $1.96$ & $1.00$ & $0.00$ \\
&  & $0.3$ & $0.49$ & $-1.24$ & $3.10$ & $8.16$ & $2.04$ & $2.50$ \\
&  & $0.6$ & $0.49$ & $-1.24$ & $58.14$ & $76.53$ & $6.25$ & $46.88$ \\
&  & $0.9$ & $0.49$ & $-1.24$ & $2.23\times 10^4$ & $1.96\times 10^4$ & $100.00$ & $1.80\times 10^4$ \\
$K_{\THK}$
& $2$ & $0$   & $0.41$ & $-3.17$ & $0.00$ & $1.66$ & $1.00$ & $0.00$ \\
&  & $0.3$ & $0.41$ & $-3.17$ & $7.93$ & $6.91$ & $2.04$ & $2.50$ \\
&  & $0.6$ & $0.41$ & $-3.17$ & $148.68$ & $64.78$ & $6.25$ & $46.88$ \\
&  & $0.9$ & $0.41$ & $-3.17$ & $5.71\times 10^4$ & $1.66\times 10^4$ & $100.00$ & $1.80\times 10^4$ \\
\end{tabular}
\end{table}

\subsection{Optimal tight difference sequences} \label{sec:more_dj_opt}
We compute the optimal values of the difference sequence $(d_0, \ldots, d_m)$ numerically in two steps. 
(i) Optimize $(\delta_1,\ldots,\delta_m)$ with standard \texttt{R} solvers 
(e.g., \texttt{optim} and \texttt{BB::dfsane}), 
initialized at $\delta_1=\cdots=\delta_m=-1/(2m)$, i.e., the optimal solution for 
independent data in \citet{Hall1990}. 
(ii) Recover $(d_0,\ldots,d_m)$ from $(\delta_0,\ldots,\delta_m)$ with $\delta_0\equiv1$
via the innovation algorithm; see Propositions 5.2.1--5.2.2 of \citealp{brockwellDavis1991}.
The solution $d_{0:m}^*(K)$ is not unique. 
Since the objective in \eqref{eqt:optDS_problem} 
depends only on $\{\delta_s=\sum_j d_j d_{j-s}\}$, both sign reversal, i.e., $\{-d_j\}_{j=0}^m$, 
and order shifts leave it unchanged. 

For estimation of $v$ or $f(\theta)$, 
the asymptotic-MSE optimal 
tight difference sequences ($d_{0:m}^*(K)$ in (\ref{eqt:optDS_problem})) 
and loose difference sequences ($d_{0:m}^{\dagger}$ in (\ref{eqt:loose_optDS_problem})) are reported  
for $m=2,3,4$, and $K = K_{\text{d}}, K_{\text{c}}, K_{\text{s}}$, 
which are three types of centrosymmetrization of $K_0(t) = (1-|t|^q)^+$ defined in Examples \ref{ex:Kd}--\ref{ex:Ks}.  
Tables \ref{table:dj_dc_opt_more}, \ref{table:dj_c_opt}, and \ref{table:dj_s_opt} show the results
for $K = K_{\text{d}}, K_{\text{c}}, K_{\text{s}}$, respectively. 

\begin{table}[!h]
\def~{\hphantom{0}}
\setlength{\tabcolsep}{4pt}
\centering
\captionsetup{font=small}
\caption{
The optimal values of $(d_0, \ldots, d_m)$ under 
tight differencing and loose differencing for estimation of $v$ and $f(\theta)$.
The centrosymmetric kernel $K_{\text{d}}$ in Example \ref{ex:Kd} is used 
with $K_0(t) = (1-|t|^q)^+$.}
\begin{tabular}{llll}
$m$& $q$& Tight difference sequence ($\lambda=1$) & Loose difference sequence ($\lambda=2$) \\
$2$ & $1$ & $0.2712,\,0.5314,\,-0.8026$ & $0.8090, \,-0.5,\, -0.3090$\\
 & $2$ & $0.3341,\,0.4781,\,-0.8123$ &  Same as above\\ 
 & $3$ & $0.3413,\,0.4717,\,-0.8130$  & Same as above\\
 & $4$ & $0.3395,\,0.4734,\,-0.8128$  & Same as above\\
$3$ & $1$ & $0.2065,\,0.2227,\,0.4240,\,-0.8532$ & $0.1942, 0.2809, 0.3832, -0.8582$\\
  & $2$ & $0.2147,\,0.2890,\,0.3579,\,-0.8616$ & Same as above\\
  & $3$ & $0.2161,\,0.2984,\,0.3477,\,-0.8622$ & Same as above\\
  & $4$ & $0.2141,\,0.2988,\,0.3490,\,-0.8620$ & Same as above\\
$4$ & $1$ & $0.1506,\,0.1765,\,0.2066,\,0.3500,\,-0.8837$ & $0.2708,\,-0.0142,\,0.6909,\,-0.4858,\,-0.4617$\\
  & $2$ & $0.1585,\,0.1949,\,0.2567,\,0.2802,\,-0.8902$ & Same as above\\
  & $3$ & $0.1580,\,0.2007,\,0.2624,\,0.2696,\,-0.8906$ & Same as above\\
  & $4$ & $0.1559,\,0.2014,\,0.2619,\,0.2713,\,-0.8905$ & Same as above\\
\end{tabular} 
\label{table:dj_dc_opt_more}
\end{table}

\begin{table}[!h]
\def~{\hphantom{0}}
\setlength{\tabcolsep}{4pt}
\centering
\captionsetup{font=small}
\caption{
The optimal values of $(d_0, \ldots, d_m)$ under 
tight differencing and loose differencing for estimation of $v$ and $f(\theta)$.
The centrosymmetric kernel $K_{\text{c}}$ in Example \ref{ex:Kc} is used 
with $K_0(t) = (1-|t|^q)^+$.}
\begin{tabular}{llll}
$m$& $q$& Tight difference sequence ($\lambda=1$) & Loose difference sequence ($\lambda=2$) \\
$2$ & $1$ & $0.2712,\,0.5314,\,-0.8026$ & $0.8090, \,-0.5,\, -0.3090$\\
 & $2$ & $0.3145,\,0.4953,\,-0.8098$ &  Same as above\\ 
 & $3$ & $0.3117,\,0.4977,\,-0.8094$  & Same as above\\
 & $4$ & $0.3049,\,0.5035,\,-0.8084$  & Same as above\\
$3$ & $1$ & $0.2065,\,0.2227,\,0.4240,\,-0.8532$ & $0.1942, 0.2809, 0.3832, -0.8582$\\
  & $2$ & $0.2060,\,0.2734,\,0.3799,\,-0.8593$ & Same as above\\
  & $3$ & $0.2016,\,0.2752,\,0.3821,\,-0.8589$ & Same as above\\
  & $4$ & $0.1969,\,0.2717,\,0.3893,\,-0.8578$ & Same as above\\
$4$ & $1$ & $0.1506,\,0.1765,\,0.2066,\,0.3500,\,-0.8837$ & $0.2708,\,-0.0142,\,0.6909,\,-0.4858,\,-0.4617$\\
  & $2$ & $0.1522,\,0.1877,\,0.2433,\,0.3052,\,-0.8884$ & Same as above\\
  & $3$ & $0.1480,\,0.1882,\,0.2432,\,0.3085,\,-0.8879$ & Same as above\\
  & $4$ & $0.1442,\,0.1862,\,0.2397,\,0.3168,\,-0.8870$ & Same as above\\
\end{tabular} 
\label{table:dj_c_opt}
\end{table}

\begin{table}[!h]
\def~{\hphantom{0}}
\setlength{\tabcolsep}{4pt}
\captionsetup{font=small}
\centering
\caption{
The optimal values of $(d_0, \ldots, d_m)$ under 
tight differencing and loose differencing for estimation of $v$ and $f(\theta)$.
The centrosymmetric kernel $K_{\text{s}}$ in Example \ref{ex:Ks} is used 
with $K_0(t) = (1-|t|^q)^+$.}
\begin{tabular}{llll}
$m$& $q$& Tight difference sequence ($\lambda=1$) & Loose difference sequence ($\lambda=2$) \\
$2$ & $1$ & $0.2712,\,0.5314,\,-0.8026$ & $0.8090, \,-0.5,\, -0.3090$\\
 & $2$ & $0.3203,\,0.4902,\,-0.8106$ &  Same as above\\ 
 & $3$ & $0.3201,\,0.4904,\,-0.8106$  & Same as above\\
 & $4$ & $0.3144,\,0.4954,\,-0.8098$  & Same as above\\
$3$ & $1$ & $0.2065,\,0.2227,\,0.4240,\,-0.8532$ & $0.1942, 0.2809, 0.3832, -0.8582$\\
  & $2$ & $0.2084,\,0.2781,\,0.3736,\,-0.8601$ & Same as above\\
  & $3$ & $0.2053,\,0.2819,\,0.3727,\,-0.8599$ & Same as above\\
  & $4$ & $0.2011,\,0.2793,\,0.3787,\,-0.8591$ & Same as above\\
$4$ & $1$ & $0.1506,\,0.1765,\,0.2066,\,0.3500,\,-0.8837$ & $0.2708,\,-0.0142,\,0.6909,\,-0.4858,\,-0.4617$\\
  & $2$ & $0.1540,\,0.1898,\,0.2474,\,0.2978,\,-0.8890$ & Same as above\\
  & $3$ & $0.1506,\,0.1915,\,0.2488,\,0.2979,\,-0.8888$ & Same as above\\
  & $4$ & $0.1472,\,0.1900,\,0.2461,\,0.3048,\,-0.8881$ & Same as above\\
\end{tabular} 
\label{table:dj_s_opt}
\end{table}

\subsection{Tight differencing under various strengths of dependence}\label{sec:diff_various_strength}

We study how autocovariance structure affects 
tight differencing.
The results in the main text are stated under the condition $u_q=\sum_{k\in\mathbb Z}|k|^q|\gamma_k|<\infty$.
Here, we relate and generalize this condition to three 
commonly discussed dependence regimes.
We consider $\{\gamma_k\}_{k\in\mathbb{N}}$ that decay
polynomially,
decay geometrically,
exhibit finite-lag truncation,
\begin{gather} 
    \text{there are $a,b,c\in\mathbb{R}$}\quad
    \text{such that}\quad
    |\gamma_k| \lesssim a(k+b)^c\quad
    \text{for all $k\in\mathbb{N}$}; \label{equ:acvf:poly}\\
    \text{there is $\phi\in(0,1)$}\quad
    \text{such that}\quad
    |\gamma_k| \lesssim \phi^{|k|}\quad
    \text{for all $k\in\mathbb{N}$}; \label{equ:acvf:geom} \\
    \text{there is $b\in\mathbb{N}_0$}\quad
    \text{such that}\quad
    \gamma_k =0\quad
    \text{for all $|k| > b$}.  \label{equ:acvf:b-dep}
\end{gather}

These three cases describe different tail behaviors of $\{\gamma_k\}_{k\in\mathbb N}$:
polynomial decay allows relatively persistent autocovariances, 
geometric decay makes large-lag autocovariances exponentially small, 
and finite-lag truncation sets all autocovariances beyond a fixed lag to zero.

We first consider the polynomial decay of $\{\gamma_k\}$ as in \eqref{equ:acvf:poly}, also called algebraic decay.

\begin{proposition}\label{prop:poly_decay}
    Assume that $\{\gamma_k\}_{k\in\mathbb{Z}}$ satisfies 
    (\ref{equ:acvf:poly}).
    (i) If $c<-1-q$, then $u_q<\infty$.
    (ii) If $K \in \mathcal{K}_{q,q^{\prime}}$ with $q,q^{\prime}\in\mathbb{N}$ such that $q \leq q^{\prime}$ and Assumption \ref{assump:weakdep} holds, 
    then the optimal $d_{0:m}$ and $\ell$ remain the same as 
    $d_{0:m}^*(K)$ in (\ref{eqt:optDS})
    and 
    $\ell^*(K)$ in (\ref{eqt:optBW}), respectively. 
\end{proposition}

Proposition \ref{prop:poly_decay} states that polynomial decay (\ref{equ:acvf:poly}) implies $u_q<\infty$, the condition assumed in Theorem \ref{thm:MSE_propsal_lrv}, whenever $c<-1-q$. 
Hence the bias--variance formulas in Theorem \ref{thm:MSE_propsal_lrv} apply directly. 
In particular, the $\textsc{mse}$-optimal tight difference sequence and bandwidth derived in \S\ref{sec:Op_tightD} remain valid, yielding $\MSE_0(\hat v_{\tight})=O(n^{-2q/(1+2q)})$.

Geometric decay \eqref{equ:acvf:geom},
which is also said to decay exponentially fast, 
vanishes faster than the polynomial case \eqref{equ:acvf:poly}. 
In this case, $\hat v_{\tight}$ can converge faster if the usual leading bias term of order $\ell^{-q}$ is removed, 
so that the bias is asymptotically negligible relative to variance.
We propose the following dual-flat kernel:
\begin{align} \label{eqt:flattop_Ksym}
K_{\DF}(t) = \mathbb{1}(|t|\leq c_1)+ H(t)\cdot \mathbb{1}(c_1 < |t|< c_2) + 0\cdot \mathbb{1}( |t| \geq c_2),
\end{align}
where $0< c_1 \leq c_2 <1$.
Let $c_0=\min(c_1,1-c_2)$. 
The kernel $K_{\DF}$ is similar to the flat-top kernel of \citet{politis2003}
as both are flat near the origin.
The additional flatness near the boundary here is imposed to preserve compatibility with tight differencing.
\begin{proposition} \label{prop:geom_decay}
    Suppose Assumption \ref{assump:weakdep} holds.
    Assume that $\{\gamma_k\}_{k\in\mathbb{Z}}$ satisfies 
    (\ref{equ:acvf:geom}).
    If $m\in\mathbb{N}$, $\lambda=h/\ell=1$, and $1/\ell+\ell/n\rightarrow 0$, 
    then 
    \[
        {\Bias}_0 \{\hat{v}_{\tight}(K_{\DF})  \} = O\left( {\ell}/{n} \right) + 
        O ( \phi^{c_0 \ell}   ) 
        \quad \text{and}\quad
        {\Var}_0 \{\hat{v}_{\tight}(K_{\DF})  \} = {4 A_d \ell v^2}/{n} + o ( {\ell}/{n} ).
    \]
    The asymptotic $\textsc{mse}$-optimal tight difference sequence is 
        $d_{0:m}^{*\textsc{geom}}(K_{\DF}) =  {\argmin_{d_{0:m}\in\mathcal{D}_m}}\,A_d$.
\end{proposition}
The proof of Proposition \ref{prop:geom_decay} can be found in \S\ref{pf:geom_decay}.
By Proposition \ref{prop:geom_decay}, 
under geometric decay,
the bias term decreases at an exponential rate,
so the mean-squared error is primarily governed by the variance term.
Consequently, the asymptotic $\textsc{mse}$-optimal tight difference sequence minimizes $A_d$ rather than $B_d^2 A_d^{2q}$.
The values of $d_{0:m}^{*\textsc{geom}}(K_{\DF})$ are listed in Table \ref{table:opT_dj_FT} of the supplement.
If $\ell \sim \xi \log n$ with $\xi<-1/(c_0\log\phi)$,
then ${\MSE}_0\{\hat{v}_{\tight}(K_{\DF})\} = O(\log n/n)$,
which matches the convergence rate of non-difference-based estimators with flat-top kernels under constant means; see \citet{politis2003}.

Finally, we consider the finite-lag autocovariance case
in (\ref{equ:acvf:b-dep}), which is satisfied by $b$-dependent time series, including classical finite-order moving average (\MA) processes.
Unlike the polynomial and geometric cases,
the long-run variance under (\ref{equ:acvf:b-dep}) is a finite sum,
$v=\sum_{|k|\le b}\gamma_k$. 
Hence the problem here is no longer to control the tail of the autocovariance sequence, but to ensure that all nonzero autocovariances are included with equal weights.
To preserve consistency, we again use the dual-flat kernel $K_{\DF}$ defined in (\ref{eqt:flattop_Ksym}).
We first state the property that guarantees unit weights on all nonzero autocovariances.

\begin{proposition}\label{prop:b-dep kernel}
Define $K_{\DF}$ as in (\ref{eqt:flattop_Ksym}).
Denote $K_{\DF, d}$ as the effective kernel of $K_{\DF}$ according to (\ref{eqt:Kd}) with $K=K_{\DF}$. 
(i) If $\ell \geq b/c_0$, then $K_{\DF, d}(k/\ell)=1$ for all $|k| \leq b$.
(ii) The smallest $\ell$ to ensure $K_{\DF, d}(k/\ell)=1$ for all $|k| \leq b$ is $\ell = 2b$, which is achieved by setting $c_1=c_2=1/2$.
\end{proposition}

The proof of Proposition \ref{prop:b-dep kernel} can be found in \S\ref{pf:prop_b-dep kernel}.
In Proposition \ref{prop:b-dep kernel}(i), 
with fixed $\ell$, each $K_{\DF,d}(k/\ell)$ is constant in $n$. 
Consistency of $\hat v_{\tight}(K_{\DF})$ requires $K_{\DF,d}(k/\ell)=1$ for all $|k|\le b$, i.e., unit weights on $\hat\gamma_0,\ldots,\hat\gamma_b$. 
Unlike the cases of polynomially and geometrically decaying $\gamma_k$,
$\hat v_{\tight}(K_{\DF})$ is parametric here since the estimand $v=\sum_{|k|\le b}\gamma_k$ has finitely many unknown parameters. 
Proposition \ref{prop:b-dep kernel}(ii) shows that the smallest bandwidth achieving this consistency property is $\ell=2b$, obtained by setting $c_1=c_2=1/2$.
This gives the discontinuous centrosymmetric kernel $K_{\DF}^{\circ}(t)=\mathbb{1}(|t|\le 1/2)$ 
in Example~\ref{ex:Kd}, 
producing $\hat v_{\tight}(K_{\DF}^{\circ})$ with the fewest sample autocovariances:
\begin{align}\label{eqt:vHat_DF_d-dep}
    \hat{v}_{\tight}(K_{\DF}^\circ) 
= \sum_{|k| \leq 2b} K_{\DF}^\circ(k/(2b))\hat{\gamma}_k
= \sum_{|k| \leq b} \hat{\gamma}_k,
\quad \text{where} \quad 
D_i = \sum_{j=0}^m d_j X_{i-2bj}.
\end{align} 
The next proposition gives the asymptotic properties of $\hat{v}_{\tight}(K_{\DF}^\circ)$.

\begin{proposition}\label{prop:b-dep} 
Suppose Assumption \ref{assump:weakdep} holds.
Assume $\{\gamma_k\}_{k\in\mathbb{Z}}$ satisfies 
(\ref{equ:acvf:b-dep}). 
(i) Let 
        $B = 2mb+b$,
        ${J}_{q}$ be a $q$-vector of ones, 
        ${W} = (1/2, {J}_{b}^\T, \delta_1{J}_{2mb}^\T)^\T$,
        and symmetric matrix
        ${\Xi} \in\mathbb{R}^{(B+1) \times (B+1)}$
        whose $(i+1,j+1)$-entry with $0\leq i < j\leq B$ satisfies
        $\Xi_{ij} \sim \{ \mathbb{1}(j-i \leq 2b)\sum_{k=-b}^{b+i-j}    \gamma_k \gamma_{k+j-i} + \mathbb{1}(i+j \leq 2b)\sum_{k=-b+i}^{b-j}  \gamma_{k-i} \gamma_{k+j} \}/n$.
        Then,
        \[
        {\Bias}_0 \{\hat{v}_{\tight}(K_{\DF}^\circ) \} = O(b/n) 
        \quad \text{and} \quad
        {\Var}_0 \{\hat{v}_{\tight}(K_{\DF}^\circ) \} = 4 {W}^\T {\Xi} {W} + o(b/n).
        \] 
(ii) Let $\mathcal{H}_1 = \sum_{i=b+1}^{B} \Xi_{ii}$,
        $\mathcal{H}_2 = \sum_{j=b+1}^{2b} \Xi_{0j}$,
        $\mathcal{H}_3 = \sum_{i=1}^b \sum_{j=b+1}^{i+2b} \Xi_{ij}$, and
        $\mathcal{H}_4 = \sum_{i=b+1}^{B-1} \sum_{j=i+1}^{\min(i+2b, B)} \Xi_{ij}$.
        The asymptotic $\textsc{mse}$-optimal tight difference sequence is 
        \begin{align}\label{eqt:opt_d_vanish} 
            d_{0:m}^{*\textsc{vani}}
            = \argmin_{d_{0:m}\in\mathcal{D}_m}\left\{ (\mathcal{H}_1 + 2 \mathcal{H}_4) \delta_1^2  +
    (\mathcal{H}_2 + 2 \mathcal{H}_3) \delta_1 \right\}. 
        \end{align} 

\end{proposition} 

The proof of Proposition \ref{prop:b-dep} can be found in \S\ref{pf:prop_b-dep}.
Since $b<\infty$, Proposition \ref{prop:b-dep} (i) shows that
$\hat v_{\tight}(K_{\DF}^\circ)$ is $n^{1/2}$-consistent.
This reflects the fact that, under \eqref{equ:acvf:b-dep}, the target
$v=\sum_{|k|\le b}\gamma_k$ contains only finitely many autocovariances.
However, the optimal tight difference sequence is no longer universal.
By Proposition \ref{prop:b-dep}(ii), it depends on the quantities $\mathcal H_1,\ldots,\mathcal H_4$, and hence on the entire autocovariance vector $(\gamma_0,\ldots,\gamma_b)$. 
Therefore,
unlike the polynomial or geometric cases,
the optimal difference sequence cannot be tabulated universally and must be estimated per series. 
Here $\Xi$ denotes the asymptotic covariance of $(\hat\gamma_0,\hat\gamma_1,\ldots,\hat\gamma_{(2m+1)b})$.
And we have  
\begin{gather*}
    \mathcal{H}_1 
        \sim \frac{2}{n} \left\{ mb\gamma_0^2 + \sum_{k=1}^b (2mb-k+1) \gamma_k^2 \right\},  \quad
    \mathcal{H}_3 
        \sim\frac{1}{n} \sum_{i=1}^b \sum_{j=b+1}^{i+2b} \sum_{k=i-b}^{b-i-j}\gamma_k \gamma_{k+j-i} + \gamma_{k-i} \gamma_{k+j} , \\
    \mathcal{H}_2  
        \sim \frac{2}{n} \sum_{k=1}^b \left ( \sum_{s=0}^{k-1} \gamma_{b-s}\right)\gamma_k,  \qquad
    \mathcal{H}_4 
        \sim \frac{1}{n}\sum_{i=b+1}^{B-1} \sum_{j=i+1}^{\min(i+2b, B)} \sum_{k=-b}^{\min\{i-b,\ b-j+i \}} \gamma_k\gamma_{k+j-i}.  
\end{gather*} 
In practice, $\gamma_0,\ldots,\gamma_b$ are unknown.
We estimate them as follows. 
First, compute the first-order-differenced series $Y_i=X_i-X_{i-1}$, $i=2,\ldots,n$.
Second, fit a moving average (\MA) model of order $b+1$ to $\{Y_i\}$ to obtain autocovariances $\{\widetilde\gamma^Y_j\}_{j=0}^{b+1}$ based on the fitted \MA~coefficients.
Finally, estimate $\gamma_b, \ldots, \gamma_0$ recursively by
\[
\tilde\gamma_b=-\tilde\gamma^Y_{b+1},\qquad
\tilde\gamma_{b-1}=2\tilde\gamma_b-\tilde\gamma^Y_b,
\]
and
\[
\tilde\gamma_j
=
2\tilde\gamma_{j+1}-\tilde\gamma^Y_{j+1}-\tilde\gamma_{j+2},
\qquad j=b-2,b-3,\ldots,0.
\]
These plug-in estimates are then used in \eqref{eqt:opt_d_vanish} to compute the series-specific optimal tight difference sequence.

\subsection{Theoretical results for $\hat{v}_{p,\tight}$}\label{sec:theory_vp_tight}
The proposed estimator of $v_p = \sum_{k\in\mathbb{Z}}|k|^p \gamma_k$ is 
\begin{eqnarray*}
	\hat{v}_{p,\tight} 
    = \sum_{|k| \leq \ell} |k|^p K \left ( \frac{k}{\ell}\right) \hat{\gamma}_k, 
\end{eqnarray*}
where 
$p \in \mathbb{N}_0 = \{0, 1,2,\ldots\}$ and $\lambda = 1$ is used in the definition of $\{D_i\}$. 
This section derives the theoretical properties for $\hat{v}_{p,\tight}$.
First, we present the bias, variance and robustness of 
$\hat{v}_{p,\tight}$ in the following corollary.

\begin{corollary}[Bias and variance of $\hat{v}_{p,\tight}$] \label{corol:MSE_vq}
Suppose Assumption \ref{assump:weakdep} holds. 
Let $p \in \mathbb{N}_0$ and
$K \in \mathcal{K}_{q,q^{\prime}}$ with $q,q^{\prime}\in\mathbb{N}$ such that $q \leq q^{\prime}$. 
Assume that $u_{p+q} <\infty$,
$m\in\mathbb{N}$, 
$\lambda=h/\ell=1$, and $1/\ell+\ell^{1+2p}/n=o(1)$ as $n \to \infty$. 
    (i) If $\mu_1 = \cdots = \mu_n$, then 
    \begin{align*}
	   {\Bias}_0(\hat{v}_{p,\tight})&= \frac{B_d v_{p+q}}{\ell^q}+O\left(\frac{\ell^{1+p}}{n}\right)+o\left(\frac{1}{\ell^q}\right) ,\\
	   {\Var}_0(\hat{v}_{p,\tight}) &= \frac{4 A_d^{(p)} \ell^{1+2p} v^2}{n} + o\left(\frac{\ell^{1+2p}}{n}\right), 
    \end{align*}
    where $B_d$ is defined in (\ref{eqt:ABd}),
    $\kappa_1^{(p)} = \int_{0}^{1}t^{2p} K(1-t)K(t)\, \dd t$,  
    $\kappa_2^{(p)} = \int_{0}^{1}t^{2p} K^2(t)\, \dd t$,
    and
    \begin{eqnarray}
	    A_d^{(p)} &=& 2 \kappa_1^{(p)} \sum_{s=1}^{m}\delta_{s-1}\delta_s + 
                \kappa_2^{(p)} \sum_{|s|\leq m}\delta_s^2.
        \label{eqt:Ad(w)} 
    \end{eqnarray} 
    (ii) If $\mu_1,\ldots,\mu_n$ are arbitrary and 
    satisfy $\mathcal{G} \gtrsim \ell+mh$, then 
    \begin{align*} 
	   {R}_{p, \bias} &\equiv \Bias(\hat{v}_{p,\tight})-{\Bias}_0(\hat{v}_{p,\tight}) 
            = O \left\{\frac{\ell^{1+p}}{n}m\ell \left (\mathcal{S}^2 \mathcal{J} + \frac{m\ell}{n}\mathcal{C}^2 \right )\right\} ,  \\
	   {R}_{p, \se} &\equiv {\Var}^{1/2}(\hat{v}_{p,\tight}) - {\Var}_0^{1/2}(\hat{v}_{p,\tight}) 
            = O \left\{ \frac{\ell^{1+p}}{n} (m\ell)^{1/2} \left (\mathcal{S}^2 \mathcal{J} + \frac{m\ell}{n}\mathcal{C}^2 \right )^{1/2}\right\}.
    \end{align*}
\end{corollary}

When $p=0$, Corollary \ref{corol:MSE_vq} reduces to Theorem \ref{thm:MSE_propsal_lrv}. 
By Corollary \ref{corol:MSE_vq}, we know that 
if $\ell=o\{n^{1/(1+2p+2q)} \}$, then under constant mean,
$\hat{v}_{p,\tight}$ converges optimally in $\mathcal{L}^2$,
i.e., $\MSE_0(\hat{v}_{p,\tight}) = O\{n^{-2q/(1+2p+2q)}\}$.
The robustness property of $\hat{v}_{p,\tight}$ is also similar to 
that of $\hat{v}_{\tight}$.
The limiting distribution of $\hat{v}_{p,\tight}$ is presented below. 

\begin{corollary}[Asymptotic normality of $\hat{v}_{p,\tight}$]
    \label{corol:clt_vq_hat}
    Assume the conditions stated in Corollary \ref{corol:MSE_vq}.
    Suppose further that $\sum_{i=0}^\infty i \theta_{\nu,i} < \infty$ for some $\nu>4$.
    If $\mu_1=\cdots=\mu_n$, then
    \[
    \left\{\frac{n}{\ell^{2p+1} A_d^{(p)}} \right\}^{1/2}
    \left\{ \hat{v}_{p,\tight} - \E(\hat{v}_{p,\tight})  \right\}  \inD \Normal(0, 4v^2).
    \]
\end{corollary}

Next, we derive the asymptotic \textsc{mse}-optimal bandwidth and tight difference sequence for estimating $v_p$.

\begin{proposition}[Optimal bandwidth for estimating $v_p$]  \label{prop:l_vq}
Assume the conditions stated in Corollary \ref{corol:MSE_vq}. 
Further assume $v_{p+q} \neq 0$ and $\ell=o\{n^{1/(1+p+q)}\}$. 
Let $\ell = \xi n^{1/(1+2p+2q)}$ for some $\xi\in\mathbb{R}^+$.
Then, with $B_d$ defined in (\ref{eqt:ABd}) 
and $A_d^{(p)}$ defined in (\ref{eqt:Ad(w)}), the optimal $\xi$ is
\begin{eqnarray} \label{eq:op_lq}
	\bar{\xi}_{v_p}^* \equiv \bar{\xi}^*_{v_p}(K, d_{0:m}) 
	=  \left\{ \frac{ (B_d v_{p+q}/v)^2 q }{2(1+2p) A_d^{(p)}} \right \}^{\frac{1}{1+2p+2q}}.
\end{eqnarray} 
\end{proposition}

By Proposition \ref{prop:l_vq}, 
if $\bar{\ell}_{v_p}^* = \bar{\xi}_{v_p}^* n^{1/(1+2p+2q)}$ is used,
then the best possible mean squared error satisfies  
\begin{align*}
n^{\frac{2q}{1+2p+2q}}{\MSE}_0 &\left \{\hat{v}_{p,\tight}(K,\bar{\ell}^*_{v_p}, d_{0:m})\right \} \\
&\to \left ( 1+\frac{2q}{1+2p}\right)\left \{  (B_d v_{p+q} )^{2(1+2p)}   \left (\frac{2(1+2p) A_d^{(p)} v^2}{q} \right )^{2q} \right \} ^{\frac{1}{1+2p+2q}}.
\end{align*}
Once the kernel $K$, the values of $p$ and $q$ are given, 
we can find the optimal difference sequence for $\hat{v}_{p,\tight}$ 
by the following minimization procedure: 
\begin{align}
    d_{v_p, 0:m}^{*}(K) 
    = \argmin_{d_{0:m}\in\mathcal{D}_m} B_d^{2(1+2p)}\left\{ A_d^{(p)}\right\}^{2q} , \label{eqt:optDS_vq_problem}
\end{align}
which depends on $p$, therefore, 
the optimal tight difference sequence for estimating $v$ and $v_p$ ($p\neq 0$) may be different.
An example is given below. 

\begin{example}[Optimal difference sequence for $\hat{v}_{p,\tight}$]\label{eg:optDiffSeq_v_w}
We consider $K(t) = (1-|t|)^+$, where $K\in \mathcal{K}_{1,1}$. 
The optimal tight and loose difference sequences for 
$\hat{v}_{p,\tight}(K)$
and 
$\hat{v}_{p,\loose}(K)$
are shown in Table \ref{table:dj_vq}
for $p=1,2,3$ and $m=2,3$.

\begin{table}[t]
\def~{\hphantom{0}}
\setlength{\tabcolsep}{4pt}
\centering
\captionsetup{font=small}
\caption{
The optimal values of $(d_0, \ldots, d_m)$ under 
tight differencing and loose differencing 
for estimating $v_p$ by using $K(t) = (1-|t|)^+$. 
}
\begin{tabular}{llll} 
$m$& $p$& Tight difference sequence ($\lambda=1$) & Loose difference sequence ($\lambda=2$) \\
$2$ & $1$ & $0.0000,\, 0.7071,\, -0.7071$ & $0.8090, \,-0.5,\, -0.3090$\\
 & $2$ & $0.2368,\, 0.5583,\, -0.7951$ &  Same as above\\ 
 & $3$ & $0.3195,\, 0.4910,\, -0.8105$  & Same as above\\
$3$ & $1$ & $0.1546,\, 0.0382,\, 0.5950,\, -0.7878$ & $0.1942, 0.2809, 0.3832, -0.8582$\\
 & $2$ & $0.1984,\, 0.1834,\, 0.4626,\, -0.8444$ & Same as above\\
 & $3$ & $0.1917,\, 0.2385,\, 0.4228,\, -0.8530$ & Same as above\\
\end{tabular} 
\label{table:dj_vq}
\end{table}
\end{example}

\subsection{Theoretical results for $\hat{f}_{p,\tight}$}\label{sec:otherACVF_fw}
The proposed estimator of $f_p(\theta) = (2\pi)^{-1}\sum_{k\in \mathbb{Z}} |k|^p \gamma_k \cos(2\pi\theta k)$ is 
\begin{eqnarray}  \label{eq:fq_hat}
	\hat{f}_{p}(\theta) = \frac{1}{2\pi}  \sum_{|k| \leq \ell} |k|^p
        K \left(\frac{k}{\ell}\right) \hat{\gamma}_k \cos(2\pi\theta k), \quad \theta \in [0,0.5], 
\end{eqnarray}
where $p \in \mathbb{N}_0 = \{0, 1,2,\ldots\}$ and $\lambda = 1$ is used in the definition of $\{D_i\}$. 
This section derives the theoretical properties for $\hat{f}_{p,\tight}$.
First, we present the bias, variance and limiting distribution of  
$\hat{f}_{p,\tight}$ in the following Corollaries.

\begin{corollary}[Bias and variance of $\hat{f}_{p}(\theta)$] \label{corol:MSE_fq}
Suppose Assumption \ref{assump:weakdep} holds. 
Let $p \in \mathbb{N}_0$ and
$K \in \mathcal{K}_{q,q^{\prime}}$ with $q,q^{\prime}\in\mathbb{N}$ such that 
$q \leq q^{\prime}$. 
Assume that $u_{p+q} <\infty$,
$m\in\mathbb{N}$, 
$\lambda=h/\ell=1$, and $1/\ell+\ell^{1+2p}/n=o(1)$ as $n \to \infty$. 
If $\mu_1=\cdots=\mu_n$, then 
\begin{align*}
    {\Bias}_0 \{\hat{f}_{p}(\theta) \} 
    	&\equiv \E_0\{\hat{f}_{p}(\theta) \} - f_p(\theta)
        = \frac{B_d f_{p+q}(\theta)}{\ell^q}
            +O\left (\frac{\ell^{1+p}}{n} \right)
            +o \left(\frac{1}{\ell^q} \right)  , \\
    {\Var}_0\{\hat{f}_{p}(\theta)\} 
        &= \frac{4 A_d^{(p)} \ell^{1+2p} f^2(\theta)\varpi(\theta)}{n} + o\left (\frac{\ell^{1+2p}}{n}\right).
\end{align*}
where $B_d$ is defined in (\ref{eqt:ABd}),  
and $A_d^{(p)}$ is defined in (\ref{eqt:Ad(w)}).
\end{corollary}

\begin{corollary}[Asymptotic normality of $\hat{f}_{p}(\theta)$] \label{corol:clt_spec_w}
    Suppose the conditions in Corollary \ref{corol:MSE_fq} hold.
    Assume $\sum_{i=0}^\infty i \theta_{\nu,i} < \infty$ for some $\nu>4$.
    If $\mu_1=\cdots=\mu_n$, then for $\theta \in [0, 0.5]$,
    \[
    	\left\{\frac{n}{\ell^{1+2p} A_d^{(p)} \varpi(\theta)}\right\}^{1/2}
	   	\left\{ \hat{f}_{p}(\theta) - \E\{\hat{f}_{p}(\theta)\} \right\} \inD \Normal(0, 4 f^2(\theta)).
    \]
\end{corollary}

Recall that 
we define $F_q = \int_{0}^{0.5} f_q^2(\theta) \dd \theta$ for $q\in\mathbb{N}_0$
and 
$F = F_0$.
If $\ell=o\{ n^{1/(1+p+q)}\}$, then the mean integrated squared error of $\hat{f}_{p,\tight}$ satisfies
\begin{eqnarray} \label{eq:MSE_fq}
	{\MISE}_0 \left \{ \hat{f}_{p,\tight}(\cdot) \right \}
	\sim \left(\frac{B_d}{\ell^{q}}\right)^2 F_{p+q}
     + \frac{4 A_d^{(p)} \ell^{1+2p} F}{n}. 
\end{eqnarray}

Next, we derive the asymptotic \textsc{mise}-optimal bandwidth and tight difference sequence for estimating $f_p$.

\begin{corollary}[Optimal bandwidth for estimating $f_p(\theta)$] \label{corol:l_fq}
Assume the conditions stated in Corollary \ref{corol:MSE_fq}. 
Further assume $F_{p+q} \neq 0$ and $\ell=o\{n^{1/(1+p+q)}\}$. 
Let $\ell = \xi n^{1/(1+2p+2q)}$ for some $\xi\in\mathbb{R}^+$.
Then, with $B_d$ defined in (\ref{eqt:ABd}) 
and $A_d^{(p)}$ defined in (\ref{eqt:Ad(w)}), 
the asymptotic \textsc{mse}-optimal $\xi$ and the asymptotic \textsc{mise}-optimal $\xi$ are 
\begin{eqnarray} \label{eq:op_lq_spec}
    \bar{\xi}_{f_p(\theta)}^* \equiv \bar{\xi}_{f_p(\theta)}^*(K, d_{0:m}) &=&  \left \{ \left (\frac{f_{p+q}(\theta)}{f(\theta)} \right )^2 \frac{ B_d^2 q }{2(1+2p) A_d^{(p)} \varpi(\theta)  } \right \}^{\frac{1}{1+2p+2q}}, \\
	\bar{\xi}_{f_p}^* \equiv \bar{\xi}_{f_p}^*(K, d_{0:m}) &=&  \left\{ \frac{F_{p+q}}{F}\frac{ B_d^2 q }{2(1+2p) A_d^{(p)}  } \right \}^{\frac{1}{1+2p+2q}}, 
\end{eqnarray} 
respectively.
\end{corollary}

By Corollary \ref{corol:l_fq}, 
if $\bar{\ell}_{f_p}^*= \bar{\xi}_{f_p}^* n^{1/(1+2p+2q)}$ is used,
then the best possible MSE satisfies 
\begin{align*}
n^{\frac{2q}{1+2p+2q}}{\MSE}_0 &\left \{\hat{f}_{p,\tight}(\theta, K,\bar{\ell}^*_{f_p}, d_0, \ldots, d_m)\right \} \\
&\to \left ( 1+\frac{2q}{1+2p}\right)\left \{  (B_d f_{p+q}(\theta) )^{1+2p}   \left (\frac{2(1+2p) A_d^{(p)} \varpi(\theta) f(\theta)}{q} \right )^{q} \right \} ^{\frac{2}{1+2p+2q}},
\end{align*}
and the best possible MISE satisfies 
\begin{align*}
n^{\frac{2q}{1+2p+2q}}{\MISE}_0 &\left \{\hat{f}_{p,\tight}(\cdot, K,\bar{\ell}^*_{f_p}, d_0, \ldots, d_m)\right \} \\
&\to \left ( 1+\frac{2q}{1+2p}\right)\left \{  (B_d^2 F_{p+q} )^{1+2p}   \left (\frac{2(1+2p) A_d^{(p)} F}{q} \right )^{2q} \right \} ^{\frac{1}{1+2p+2q}}.
\end{align*}
We remark that for a given kernel $K$, and given values of $p$ and $q$, 
the \textsc{mse}-optimal tight difference sequence and the \textsc{mise}-optimal tight difference sequence for $\hat{f}_{p}(\theta)$ are always the same,
which is moreover, the same as that for $\hat{v}_{p,\tight}$; 
see Example \ref{eg:optDiffSeq_v_w}.

Next, we study the behavior of $\hat{f}_{p,\tight}$ under the autocovariance structures (\ref{equ:acvf:poly})--(\ref{equ:acvf:b-dep}).
Note that all results below hold for $\hat{v}_{p,\tight}$ as well because it is 
a special case of $\hat{f}_{p,\tight}(\theta)$ when $\theta=0$.

If (\ref{equ:acvf:poly}) holds with $c<-1-q$, 
then 
the asymptotic properties and the optimal parameters are the same as 
assuming $u_q = \sum_{k \in \mathbb{Z}} |k|^q |\gamma_k| < \infty$; see Proposition \ref{prop:poly_decay}. 

If (\ref{equ:acvf:geom}) holds, then we use the dual-flat kernel in $\hat{f}_{p,\tight}$, 
i.e., 
\[
	\hat{f}_{p,\tight}(\theta, K_{\DF}) 
    =   \sum_{|k| \leq 2b} |k|^p K_{\DF}( {k}/\ell ) \hat{\gamma}_k \cos(2\pi\theta k),
\]
where 
$K_{\DF}$ is defined in (\ref{eqt:flattop_Ksym}), and 
$\{D_i\}$ are difference statistics with lag $h=\ell$.
The properties of $\hat{f}_{p,\tight}(\theta, K_{\DF})$ are stated in the proposition below. 
\begin{proposition}\label{prop:geom_decay_spec}
    Suppose Assumption \ref{assump:weakdep} holds.
    Assume that $\{\gamma_k\}_{k\in\mathbb{Z}}$ satisfies 
    (\ref{equ:acvf:geom}).
    If $p\in\mathbb{N}_0$,  $m\in\mathbb{N}$, $\lambda=h/\ell=1$, and $1/\ell+\ell/n\rightarrow 0$, 
    then, 
    for $\theta \in [0,0.5]$,
	\begin{align*}
		{\Bias}_0\left\{\hat{f}_{p,\tight}(K_{\DF}) \right\} &= O\left(\frac{\ell}{n} \right) + O\left( \phi^{c_0 \ell}\right), \\
		{\Var}_0\left\{\hat{f}_{p,\tight}(K_{\DF}) \right\} &= \frac{4 A_d^{(p)} \ell^{1+2p} f^2(\theta) \varpi(\theta)}{n} + o\left(\frac{\ell}{n}\right).
	\end{align*}
    The asymptotic \textsc{mse}-optimal and \textsc{mise}-optimal tight difference sequences are equivalent. That is 
        $d_{0:m}^{*(p) \textsc{geom}}(K_{\DF}) = \argmin_{d_{0:m}\in\mathcal{D}_m}A_d^{(p)}$.
\end{proposition}
By Proposition \ref{prop:geom_decay_spec}, 
if we set 
$\ell \sim \xi \log n$ with $\xi<-1/(c_0\log\phi)$, then 
$\hat{f}_{p,\tight}(\theta, K_{\DF}) = f_p(\theta) + O_p\{\sqrt{\log n/n}\}$,
which is the same as the non-difference-based estimators
with flat-top kernels in \citet{politis2003}.
The optimal tight difference sequence for $\hat{f}_{0,\tight}(\theta, K_{\DF})$
is the same as that for $\hat{v}_{\tight}(K_{\DF})$.

If (\ref{equ:acvf:b-dep}) holds, then we use the dual-flat kernel with $c_0=c_1=1/2$ in $\hat{f}_{p,\tight}$, 
i.e., 
\begin{align}\label{eqt:f_w_b-dep}
	\hat{f}_{p,\tight}(\theta, K_{\DF}^\circ) 
	= \sum_{|k| \leq 2b} |k|^p K_{\DF}^\circ(k/2b) \hat{\gamma}_k \cos(2 \pi \theta k)
	= \sum_{|k| \leq \ell} |k|^p \hat{\gamma}_k \cos(2 \pi \theta k),
\end{align}
where $\{D_i\}$ are difference statistics with lag $h=\ell=2b$.
Its properties are stated below.
\begin{proposition}\label{prop:b-dep_spec_w}
Suppose Assumption \ref{assump:weakdep} holds.
Assume that $\{\gamma_k\}_{k\in\mathbb{Z}}$ satisfies 
(\ref{equ:acvf:b-dep}).
Let $p\in\mathbb{N}_0$.
Consider the estimator $\hat{f}_{p,\tight}(\theta, K_{\DF}^\circ)$ defined in (\ref{eqt:f_w_b-dep}).
    \begin{enumerate}
        \item Let $B = 2mb+b$, ${W}$ and ${\Xi}$ be defined according to Proposition \ref{prop:b-dep}.
        		Also let  
				\begin{align*}
					{W}_f^{(p)} &= {W}_f \circ \left (0^p, 1^p , 2^p ,\ldots,B^p \right )^\T ,\\
					{W}_f 	 &= {W} \circ \left (1, \cos(2\pi\theta), \cos(4\pi\theta), \ldots, \cos(2B\pi\theta)\right )^\T, 
				\end{align*} 
				where $\circ$ represents entry-wise product. Then, for $\theta \in [0,0.5]$,  
        \begin{align*}
        {\Bias}_0 \{\hat{f}_{p,\tight}(\theta, K_{\DF}^\circ) \} 
        	&= O \left (\frac{b^{1+p}}{n} \right) , \\
        {\Var}_0 \{\hat{f}_{p,\tight}(\theta, K_{\DF}^\circ) \} 
        	&= 4 {{W}_f^{(p)}}^\T {\Xi} {W}_f^{(p)} + o\left(\frac{b^{1+2p}}{n} \right).
        \end{align*}

        \item The asymptotic \textsc{mise}-optimal tight difference sequence is  
            \begin{align}\label{eqt:opt_tight_DF_van_spec_w}
            	 d_{f_p,0:m}^{*\textsc{vani}} = \argmin_{d_{0:m}\in\mathcal{D}_m} \delta_1^2.
			\end{align}

        \item The asymptotic \textsc{mse}-optimal tight difference sequence for $\hat{f}_{p,\tight}(\cdot, K_{\DF}^\circ)$ is  
        \begin{align}\label{eqt:opt_d_vw}
			 d_{f_p(\theta),0:m}^{*\textsc{vani}}
            = \argmin_{d_{0:m}\in\mathcal{D}_m} \left\{ \left(\mathcal{H}_1^{(p)} + 2 \mathcal{H}_4^{(p)}\right) \delta_1^2  
            +  \left(\mathcal{H}_2^{(p)} + 2 \mathcal{H}_3^{(p)} \right)  \delta_1  \right\}.
        \end{align}
        where 
        \begin{align*}
            \mathcal{H}_1^{(p)} &= 
                 \sum_{i=b+1}^{B} i^{2p} \cos^2(2\pi \theta i) \Xi_{ii} \\
                &\sim \frac{1}{n} \sum_{i=b+1}^{2b} i^{2p} \cos^2(2\pi \theta i) \sum_{k=b-i}^{i-b} \gamma_k^2 +  
                    \frac{1}{n} \sum_{i=2b+1}^{B} i^{2p} \cos^2(2\pi \theta i) \sum_{k=-b}^{b}\gamma_k^2,\\
            \mathcal{H}_2^{(p)} &= \begin{cases}
                0, &\text{if } p\neq0; \\
                \displaystyle{\sum_{j=b+1}^{2b} j^p  \cos(2\pi \theta j)  \Xi_{0j}
                \sim \frac{2}{n} \sum_{j=b+1}^{2b} j^p \cos(2\pi \theta j)  \sum_{k=-b}^{b-j} \gamma_k\gamma_{k+j}}, &\text{if } p =0,
            \end{cases}
                  \\
            \mathcal{H}_3^{(p)} &= 
                 \sum_{i=1}^b \sum_{j=b+1}^{i+2b} (ij)^p \cos(2\pi \theta i) \cos(2\pi \theta j) \Xi_{ij} \\
                &\sim\frac{1}{n} \sum_{i=1}^b \sum_{j=b+1}^{i+2b} (ij)^p \cos(2\pi \theta i) \cos(2\pi \theta j) \sum_{k=i-b}^{b-i-j} \left\{\gamma_k \gamma_{k+(j-i)} + \gamma_{k-i} \gamma_{k+j} \right\}, \\
            \mathcal{H}_4^{(p)} &= 
                 \sum_{i=b+1}^{B-1} \sum_{j=i+1}^{\min\{i+2b,\ B\}} (ij)^p \cos(2\pi \theta i) \cos(2\pi \theta j) \Xi_{ij} \\
                &\sim \frac{1}{n}\sum_{i=b+1}^{B-1} \sum_{j=i+1}^{\min\{i+2b,\ B\}} (ij)^p \cos(2\pi \theta i) \cos(2\pi \theta j) \sum_{k=-b}^{\min\{i-b,\ b-(j-i) \}} \gamma_k\gamma_{k+j-i}.
        \end{align*} 
    \end{enumerate}
\end{proposition}

Note that $\mathcal{H}_j^{(0)} \equiv \mathcal{H}_j$ for $j=1,2,3,4$, and 
$d_{f_0(\theta), 0:m}^{*\textsc{vani}} \equiv d_{0:m}^{*\textsc{vani}}$. 
Hence, Proposition \ref{prop:b-dep_spec_w} generalizes Proposition \ref{prop:b-dep} of the main text. 
Interestingly, 
the asymptotic \textsc{mise}-optimal tight difference sequence and the asymptotic \textsc{mse}-optimal tight difference sequence for $\hat{f}_{p,\tight}(\cdot, K_{\DF}^\circ)$ are different. 
To get the asymptotic \textsc{mse}-optimal tight difference sequence $d_{f_p(\theta),0:m}^{*\textsc{vani}}$, 
we may first compute 
$\{\tilde{\gamma}_0,\ldots,\tilde{\gamma}_b\}$
and then plug them into (\ref{eqt:opt_d_vw}) to perform numerical optimization. 
We also remark that (\ref{eqt:opt_tight_DF_van_spec_w}) remains unchanged for all $p\in\mathbb{N}_0$; see
the numeric values of (\ref{eqt:opt_tight_DF_van_spec_w}) in Table \ref{table:dj_spec_bdep}.

\begin{table}[t]
\def~{\hphantom{0}}
\setlength{\tabcolsep}{4pt}
\centering
\captionsetup{font=small}
\caption{
The asymptotic \textsc{mise}-optimal tight difference sequence for $\hat{f}_{p,\tight}(\cdot, K_{\DF}^\circ)$
under vanishing autocovariance structure stated in (\ref{equ:acvf:b-dep}).
}
\begin{tabular}{ll} 
$m$ & Tight difference sequence ($\lambda=1$)  \\
$2$ & $0.707107,\, 0,\, -0.707107$ \\
$3$ & $0.341141,\, 0.319095,\, 0.200846,\, -0.861082$ \\
$4$ & $0.237365,\, 0.169634,\, 0.316248,\, 0.164364,\, -0.887611$ \\
\end{tabular} 
\label{table:dj_spec_bdep}
\end{table}

\section{Additional simulation experiments and real-data applications}

\subsection{Comparison between global demeaning and local demeaning} \label{sec:ex_nondiff_bad}

We present a simple experiment illustrating that the non-difference-based choice $m=0$ (global demeaning) can perform poorly when the mean is time-varying.
Let $X_i=\mu(i/n)+Z_i$ for $i=1,\ldots,n$,
where $\mu(t)=\sin(2\pi t)$ is smooth and non-constant, and the noise $\{Z_i\}$ follows an $\ARMA(1,1)$ model of \S\ref{sec:lrvEst_constMean} with $\phi=0.6$.
The true value of the long-run variance is $v=16$.
We compare $\hat v(K,\ell,d_{0:m},\lambda)$ in \eqref{eq:v_hat} with $\lambda=1$, $K=K_\s$ in Example \ref{ex:Ks} under $q=3$, and $m\in\{0,1,2\}$.
For $m=1,2$, we use the corresponding optimal difference sequences (Table \ref{table:dj_s_opt}).
For $m=0$, we set $d_0=1$, and the estimator reduces to the standard non-difference-based kernel estimator based on the globally demeaned series:
\[
\hat v_{\non}
= \sum_{|k|\le \ell}
K\left(\frac{k}{\ell}\right)
\hat\gamma_k^{(\non)},
\quad \text{where }\;
\hat\gamma_k^{(\non)}
=
\frac{1}{n}\sum_{i=|k|+1}^n
(X_i-\bar X_n)(X_{i-|k|}-\bar X_n).
\]

\begin{figure}[t]
\centering
\includegraphics[width=\textwidth]{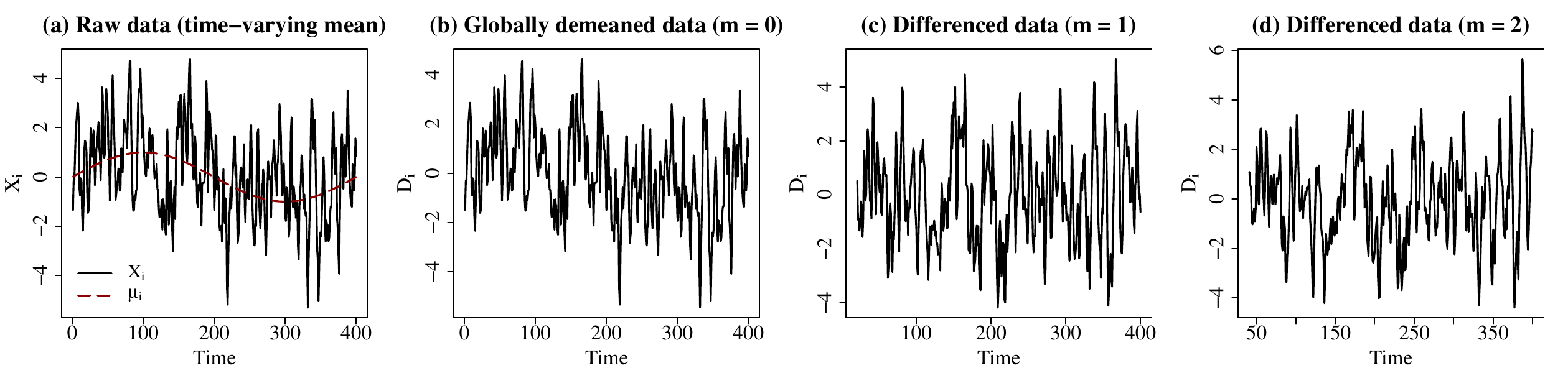}
\captionsetup{font=small}
\caption{A typical realization with $n=400$.
(a) Raw series $X_i$ with time-varying mean;
(b) Transformed series $D_i = X_i - \bar{X}_n$ under global demeaning ($m=0$);
(c)--(d) Transformed series $D_i = \sum_{j=0}^m d_j X_{i-jh}$ under local differencing ($m=1,2$).}
\label{fig:diffVSNonDiff}
\end{figure}

\begin{table}[t]
\centering
\captionsetup{font=small}
\caption{Monte Carlo results for 
$\hat v = \hat v(K, \ell, d_{0:m}, \lambda=1)$ for the experiment in \S\ref{sec:ex_nondiff_bad} under $2^{10}$ replications. 
Denote $\SMSE(\hat v) = \MSE(\hat v)/v^2$ as the standardized mean squared error.}
\begin{tabular}{lccccc}
$m$ & $\E(\hat v)$ & $\Bias^2(\hat v)$ & $\Var(\hat v)$ & $\MSE(\hat v)$ & $\SMSE(\hat v)$ \\
$0$ & 24.6310 & 74.4948 & 51.0394 & 125.4844 & 0.4902 \\
$1$ & 15.9181 & 0.0067  & 25.3785 & 25.3604  & 0.0991 \\
$2$ & 16.5034 & 0.2534  & 26.2399 & 26.4677  & 0.1034 \\
\end{tabular}
\label{tab:mc_diff_vs_nondiff}
\end{table}

Figure \ref{fig:diffVSNonDiff} displays a realization with $n=400$.
When $m=0$, only the global level is removed, leaving a pronounced time-varying trend component $\mu(i/n)-\bar\mu_n$ in $D_i$,
where $\bar{\mu}_n = \sum_{i=1}^n \mu(i/n)/n$. 
Thus, it induces non-negligible effects in the sample autocovariances $\hat\gamma_k^{(\non)}$. 
In contrast, differencing with $m=1,2$ performs local de-trending, effectively removing the non-constant mean component.
As shown in Table \ref{tab:mc_diff_vs_nondiff}, the (relative) mean squared error for $m=0$ is about five times larger than that for $m=1,2$.
This example shows that global demeaning ($m=0$) is inadequate for time series with time-varying means. 
Moreover, differencing can substantially improve the estimation of $v$.

\subsection{Accuracy and precision under constant mean}\label{sec:lrvEst_constMean_additional}
This section provides additional simulation results for the first experiment in Section \ref{sec:lrvEst_constMean}. 
The relative MSEs are visualized in Figure \ref{fig:efficiency_ARMA11} of the main text. 
For reference, we also report the MSEs, variances, and squared biases of the estimators in 
Figures \ref{fig:efficiency_ARMA11_rawMSE}, \ref{fig:efficiency_ARMA11_rawVar}, and \ref{fig:efficiency_ARMA11_rawBias}, respectively. 
We observe that the order of magnitude of the squared bias is substantially smaller than 
that of the variance. 
Hence, the MSE is mainly determined by the variance in this finite-sample experiment. 
Note that when the sample size $n$ is sufficiently large (as $n \rightarrow \infty$), 
we will observe the asymptotic property that 
$ {\Bias}_0^2(\hat{v}_{\tight}) : {\Var}_0(\hat{v}_{\tight}) : {\MSE}_0(\hat{v}_{\tight}) \rightarrow 1 : 2q : (1 + 2q)$
as $n \rightarrow \infty$.

\begin{figure}[!t]  
	\begin{center}
		\includegraphics[width=.95\linewidth]{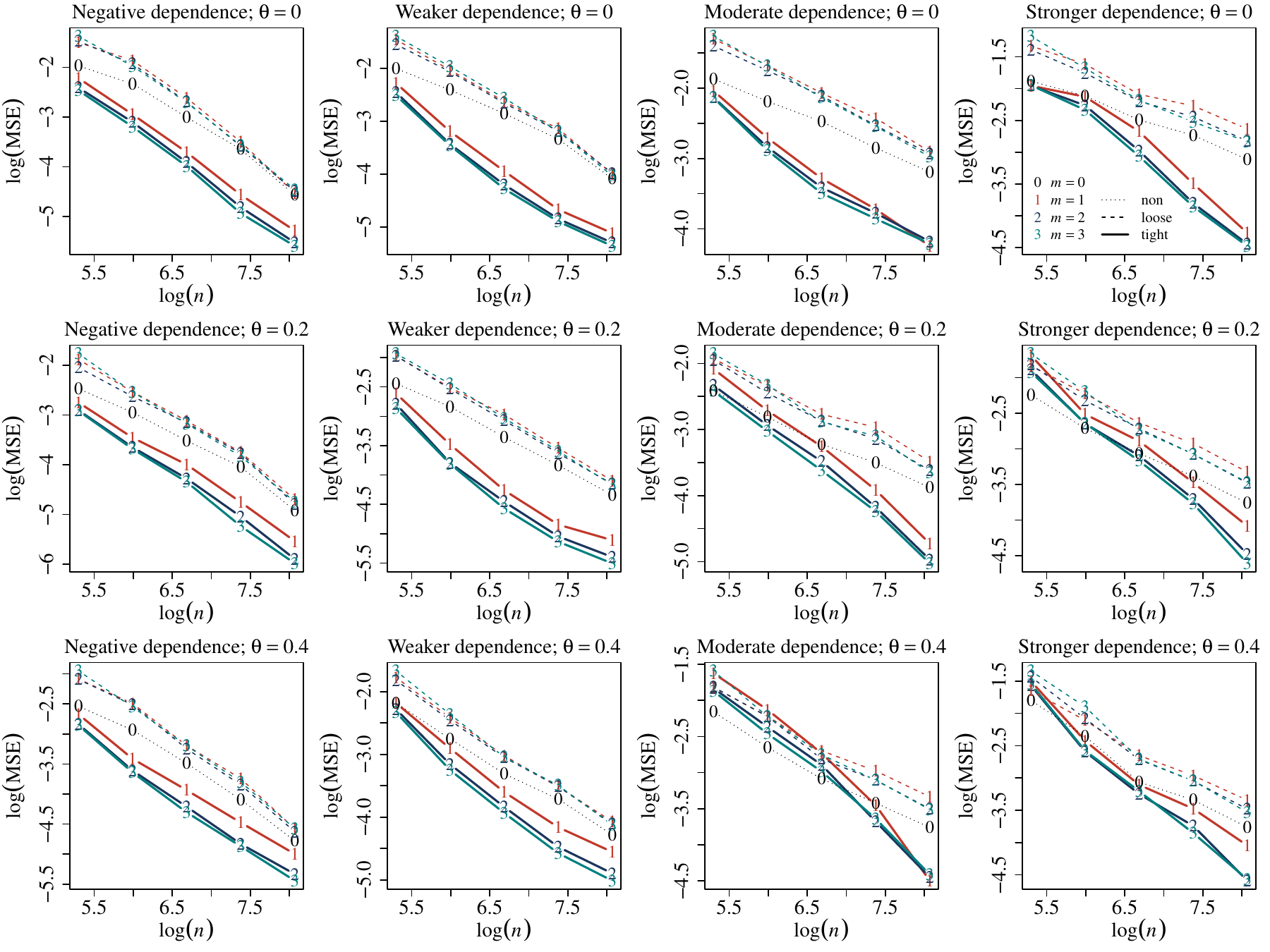} 
		\captionsetup{font=small}
		\caption{The MSEs of estimators of $f(\theta)$ are compared under the $\ARMA(1,1)$ model; 
		see Section \ref{sec:lrvEst_constMean}. 
		The dashed lines correspond to $\hat{f}_{\loose}(\theta)$, 
		the solid lines correspond to $\hat{f}_{\tight}(\theta)$, 
		and the dotted lines correspond to $\hat{f}_{\non}(\theta)$. 
		The numbers on the lines represent the order of differencing. 
		Lower curves indicate higher efficiency.
        }
		\label{fig:efficiency_ARMA11_rawMSE} 
	\end{center} 
\end{figure}

\begin{figure}[!t]  
	\begin{center}
	\captionsetup{font=small}
		\includegraphics[width=.99\linewidth]{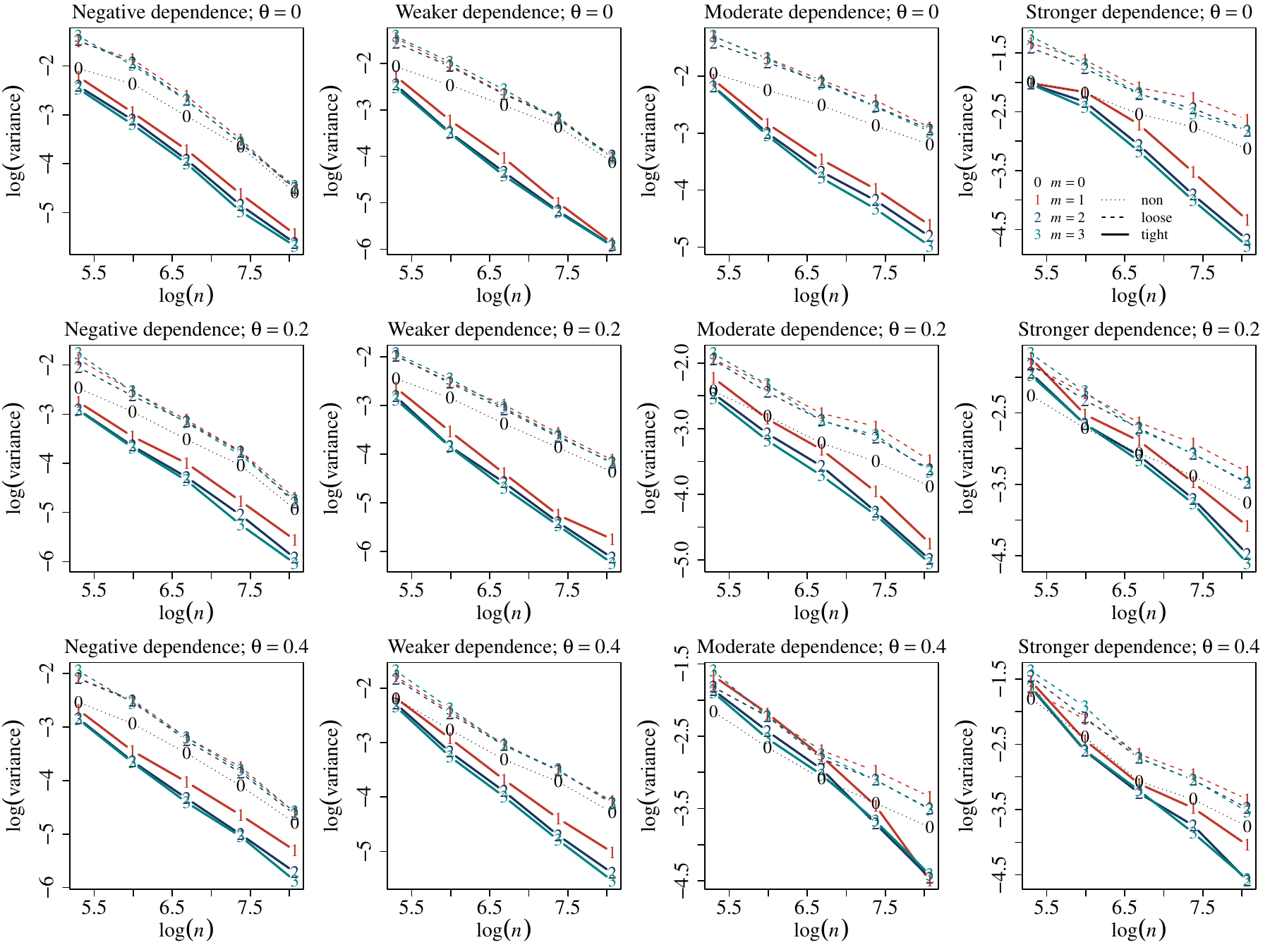} 
		\caption{The variances of estimators of $f(\theta)$ are compared under the $\ARMA(1,1)$ model; 
		see Section \ref{sec:lrvEst_constMean}. 
		The dashed lines correspond to $\hat{f}_{\loose}(\theta)$, 
		the solid lines correspond to $\hat{f}_{\tight}(\theta)$, 
		and the dotted lines correspond to $\hat{f}_{\non}(\theta)$. 
		The numbers on the lines represent the order of differencing. Lower curves indicate higher efficiency.
        }
		\label{fig:efficiency_ARMA11_rawVar} 
	\end{center} 
\end{figure}

\begin{figure}[!t]  
	\begin{center}
	\captionsetup{font=small}
		\includegraphics[width=.99\linewidth]{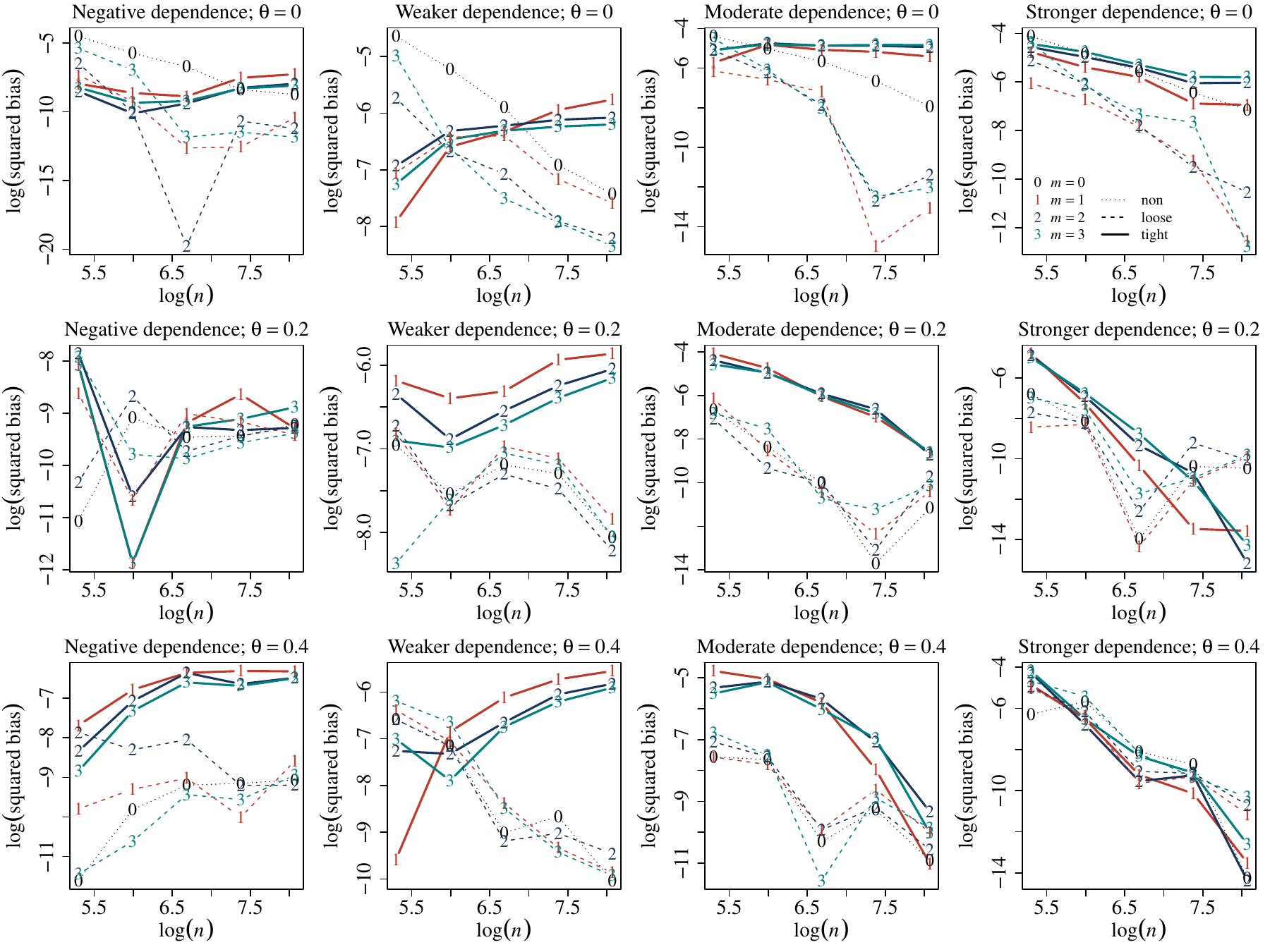} 
		\caption{The squared biases of estimators of $f(\theta)$ are compared under the $\ARMA(1,1)$ model; 
		see Section \ref{sec:lrvEst_constMean}. 
		The dashed lines correspond to $\hat{f}_{\loose}(\theta)$, 
		the solid lines correspond to $\hat{f}_{\tight}(\theta)$, 
		and the dotted lines correspond to $\hat{f}_{\non}(\theta)$. 
		The numbers on the lines represent the order of differencing. Lower curves indicate higher efficiency.
        }
		\label{fig:efficiency_ARMA11_rawBias} 
	\end{center} 
\end{figure}

\begin{figure}[!t]
    \centering
    \captionsetup{font=small}
    \includegraphics[width=0.99\linewidth] {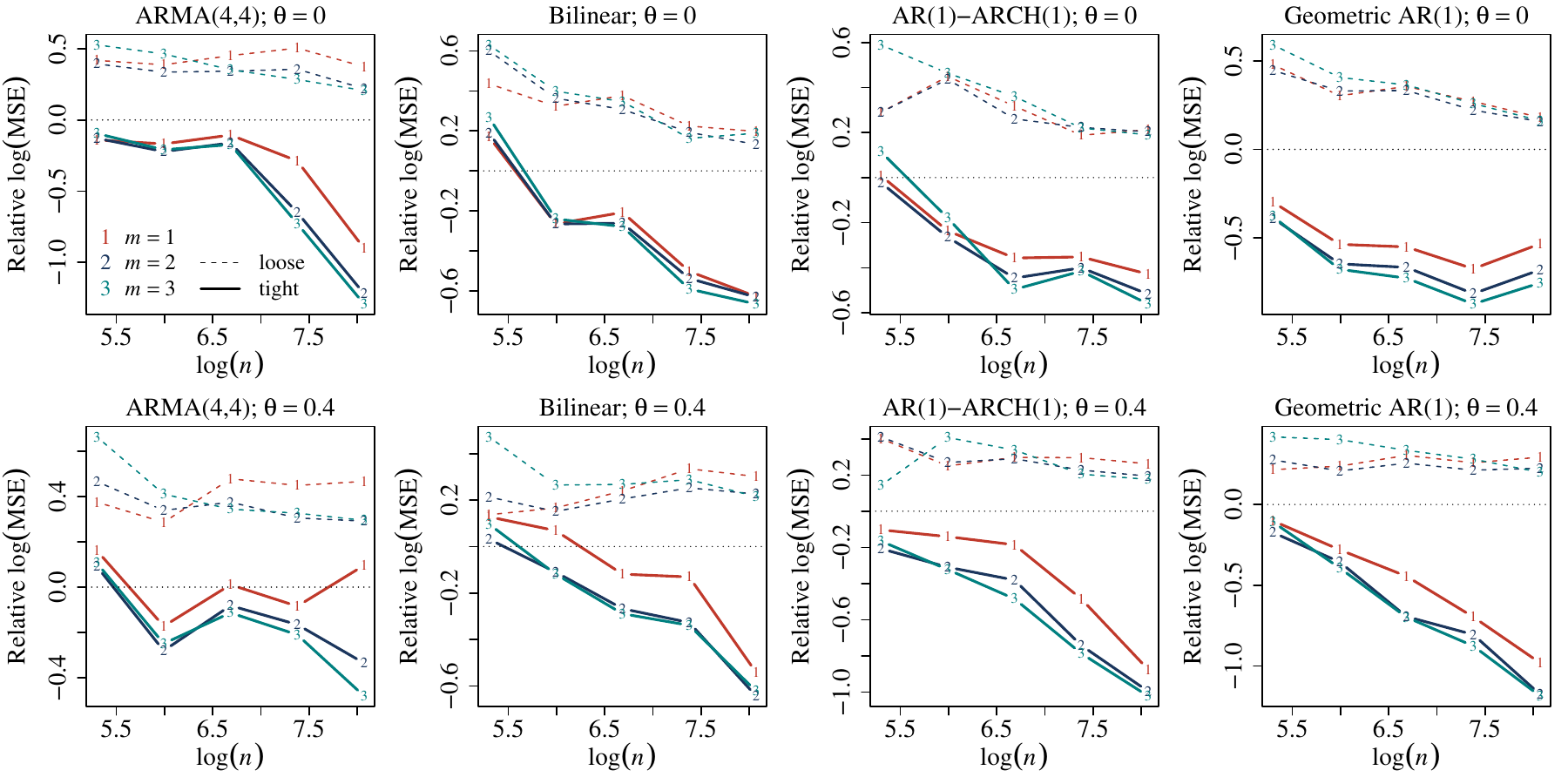}
        \vspace{-0.3cm}
    \caption{The mean squared errors of estimators of $f(\theta)$ relative to $\hat{f}_{\non}(\theta)$ are compared under four different time series models. 
    Caption details from Figure \ref{fig:efficiency_ARMA11} also apply here.}
    \label{fig:efficiency_other} 
\end{figure}

In addition, we study four more complicated time series models:
(i) $\ARMA(4,4)$ model: $Z_i = \sum_{j=1}^4 (a_j Z_{i-j} + b_j \varepsilon_{i-j}) + \varepsilon_i$, 
		where $(a_1, \ldots, a_4) = (0.4, 0.4^2, 0.4^3, -0.4^4)$ and $(b_1, \ldots, b_4) = (0.4, 0.3, 0.2, 0.1)$;
(ii) Bilinear $\ARMA(1,1)$ model: $Z_i = 0.4 Z_{i-1} + 0.8 \varepsilon_{i-1} + 0.4Z_{i-1}\varepsilon_{i-1} + \varepsilon_i$; 
(iii) $\AR(1)$ model with $\textsc{arch}(1)$ errors: $Z_i = 0.4 Z_{i-1} + (1+0.4Z_{i-1}^2)^{1/2} \varepsilon_i$; and 
(iv) Geometric $\AR(1)$ model: $Z_i = e^{Z_i'/2} - e^{(1-0.4^2)/8}$, where $Z_i' = 0.4 Z_{i-1}' + \varepsilon_i$.
In the above models, $\varepsilon_i\sim \Normal(0,1)$ independently. 
Figure \ref{fig:efficiency_other} shows the results when $\theta\in\{0,0.4\}$. 
Generally, $\hat{f}_{\tight}(\theta)$ remains the most efficient estimator and shows increasing efficiency when $m$ increases. 
The results are similar to those presented in Figure \ref{fig:efficiency_ARMA11}. 
It confirms that the performance of $\hat{f}_{\tight}(\theta)$ remains good under non-linearity and 
more involved dependence structures.

\subsection{Estimation of the spectrum}\label{sec:full_spectrum_MCexp}
In this subsection, 
we assess the performance over the full spectrum. 
Four $\ARMA(2,1)$ noise models are considered: 
\[
	Z_i=\phi_1 Z_{i-1}+\phi_2 Z_{i-2}+0.5\,\varepsilon_{i-1}+\varepsilon_i
\]
with 
$(\phi_1,\phi_2)\in\{(0.4,0.2),(0.4,-0.2),(-0.4,0.2),(-0.4,-0.2)\}$, labeled as Models I--IV, 
where $\varepsilon_i\sim \Normal(0,1)$ independently. 
Figure \ref{fig:spec} shows the true spectral densities. 
We set $n=400$ and consider: 
\[
	\mu(t)=0, \qquad 
	\mu(t)=v^{1/2} t, \qquad 
	\mu(t)= v^{1/2} \mathbb{1}(t>0.5).
\]
Estimators $\hat f_{\non}(\theta)$, $\hat f_{\loose}(\theta)$, and $\hat f_{\tight}(\theta)$ 
with $q=3$ are applied using their estimated $\textsc{mise}$-optimal bandwidths for full spectrum estimation.
All difference-based estimators use $m=3$.
Figure \ref{fig:spec} shows the 2.5\%, 50\%, and 97.5\% error quantiles over $2^{12}$ replications. 
With a constant mean, $\hat f_{\loose}(\theta)$ performs worst with the widest bands, 
while $\hat f_{\tight}(\theta)$ is comparable to the oracle $\hat f_{\non}(\theta)$. 
With non-constant means, $\hat f_{\non}(\theta)$ has large errors, especially near $\theta\approx0$,
whereas $\hat f_{\tight}(\theta)$ remains more precise than $\hat f_{\loose}(\theta)$. 
For example, under Model I, $\MISE(\hat f_{\non}):\MISE(\hat f_{\loose}):\MISE(\hat f_{\tight})\approx 32:5:1$.

\begin{figure}[!t]  
	\begin{center}
	\captionsetup{font=small}
		\includegraphics[width=.95\linewidth]{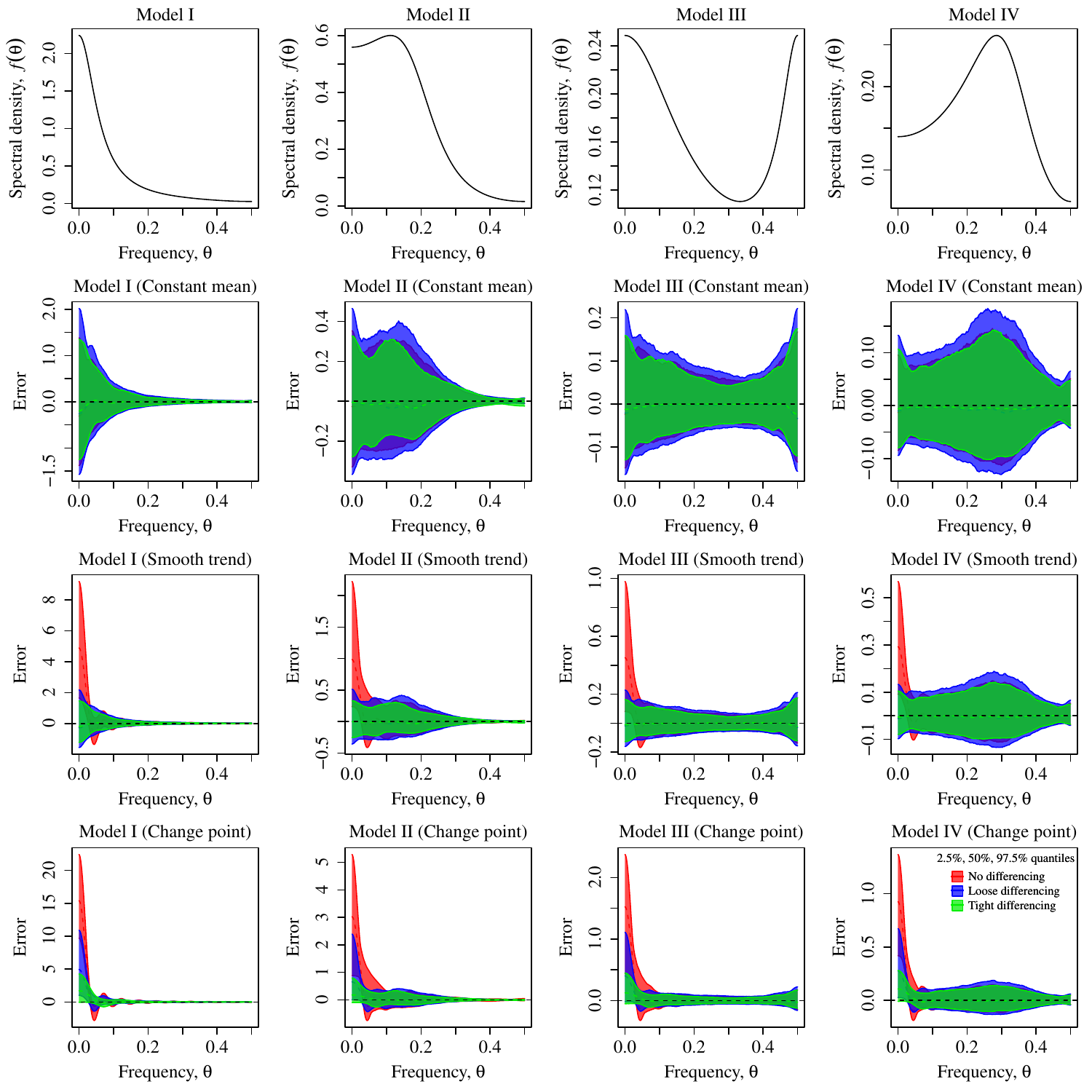}  \vspace{-0.2cm}
		\caption{
		The first row plots $f(\theta)$ under Models I--IV in \S\ref{sec:lrvEst_nonconstMean}. 
		The last three rows plot the 2.5\%, 50\%, and 97.5\% error quantiles.
		Narrower bands indicate higher precision of the estimator; bands closer to zero indicate greater accuracy.
        }
		\label{fig:spec} 
	\end{center} 
\end{figure}

\subsection{Computation time} \label{sec:computation}

We examine the computational cost of the proposed estimator $\hat v_{\tight}$. 
The tight-difference estimator can be written as a weighted sum of autocovariances of a transformed series. 
For implementation in \texttt{R}, these autocovariances can be computed using the built-in \texttt{acf} function
as in the non-difference-based estimator.  
We remark that 
we only need to compute $O(\ell)$ number of lags of sample autocovariances
because the kernel $K$ is truncated, i.e., $K(t)=0$ for $|t|\geq 1$. 
Hence, using the fast Fourier transform does not necessarily outperform the built-in \texttt{acf}.

We compare computation times across seven estimators: the tight-, loose-, and non-difference-based estimators, each implemented with and without trimming, together with a detrend-based estimator. 
The built-in \texttt{acf} function is used for all estimators. 
The detrend-based estimator first estimates the mean by local linear smoothing using \texttt{KernSmooth::locpoly}, with the trend bandwidth selected by \texttt{KernSmooth::dpill}. 
A Bartlett-kernel long-run variance estimator is then applied to the resulting residuals, with bandwidth selected by the same plug-in rule used for the other kernel-based estimators. 
All computations were carried out on a MacBook Air equipped with an Apple M2 chip. 
Computation times are averaged over $400$ replications.

\begin{figure}[t]
\centering
\includegraphics[width=0.5\textwidth]{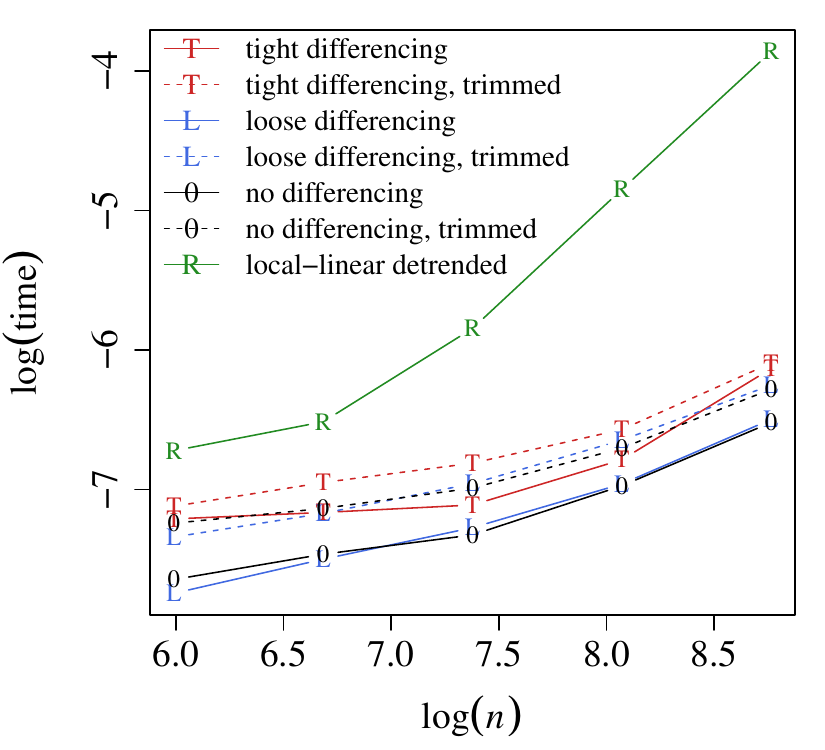}
\captionsetup{font=small}
\caption{Computation time (in seconds) is plotted against the sample size $n$, with both axes on the logarithmic scale. 
Red, blue, black, and green lines correspond to $\hat v_{\tight}$, $\hat v_{\loose}$, $\hat v_{\non}$, 
and the kernel-based estimator based on residuals from the local linear regression, respectively. 
Solid and dashed lines denote the untrimmed and trimmed versions, respectively.
}
\label{fig:computation_time}
\end{figure}

Figure \ref{fig:computation_time} displays the average computation time (in seconds) for sample sizes ranging from $n=2^2\times 100$ to $n=2^6\times 100$. 
The tight-difference-based estimator has computation time comparable to the loose-difference-based and non-difference-based estimators, for both trimmed and untrimmed versions. 
In contrast, the de-trend-based estimator is substantially slower, with computation time increasing more rapidly in $n$. 
Overall, the proposed estimator is computationally efficient, with only mild growth in computation time as the sample size increases.

\subsection{Mean test under strong persistence} \label{sec:t_test_pw}
We study the finite-sample behavior of $t$-tests for the overall mean  
in a linear regression model with strongly persistent noise. 
Specifically, consider
\[
	Y_i = \beta_0 + \beta(i/n)+Z_i, \qquad 
	Z_i = 0.9Z_{i-1} + \varepsilon_i, 
\qquad i=1,\ldots,n,
\]
where $\varepsilon_i$'s are independent innovations with mean zero and finite variance,  
$\beta_0\in\mathbb{R}$ is the overall mean, and 
$\beta(t)$ is a deterministic time-varying component satisfying $\int_0^1 \beta(t) \dd t = 0$.
We are interested in testing the null hypothesis $H_0:\beta_0=0$ against $H_1: \beta_0\neq 0$ 
based on the statistic
\[
	T = \frac{\hat \beta_0}{\sqrt{\hat v / n}},
	\qquad \text{where} \qquad
	\hat \beta_0 = \frac{1}{n} \sum_{i=1}^n Y_i,
\]
and $\hat v$ is a long-run variance estimator. 
We reject $H_0$ at 5\% size if $|T| > \Phi^{-1}(0.975)$, where $\Phi^{-1}$ is the quantile function of $\Normal(0,1)$. 

It is well known that inference based on long-run variance estimation can suffer from severe size distortions under strong dependence. 
To address this issue, \citet{AndrewsMonahan1992} proposed a prewhitening procedure combined with kernel estimation, which substantially improves size control; see also \citet{casini2024prewhitened} for recent developments, and \citet{liuchan2026}
for an alternative approach.
We consider three estimators: non-differencing, 
loose differencing, and tight differencing, together with their prewhitened counterparts. 

First, we study the constant-mean case $\beta(t) = 0$ for all $t$.
Prewhitening is implemented by fitting an auxiliary $\AR(1)$ model to the residuals $\hat{Z}_i = Y_i - \hat \beta_0$, 
using which we obtain the estimated innovations: 
\begin{align}\label{eqt:phihat_PW}
	\hat{\varepsilon}_i = \hat{Z}_i - \hat{\phi}\hat{Z}_{i-1}, 
	\qquad \text{where} \qquad
    \hat{\phi} = \frac{\sum_{i=2}^n \hat{Z}_i\hat{Z}_{i-1}}
    {\sum_{i=2}^n \hat{Z}_{i-1}^2}.
\end{align}
The pre-whitened long-run variance estimator is $\hat{v}_{\textsc{pw}} = \tilde{v}_{\textsc{pw}}/(1-\hat{\phi})^2$, 
where $\tilde{v}_{\textsc{pw}}$ is a long-run variance estimator based on the dataset $\{ \hat{\varepsilon}_i \}_{i=2}^n$.
Monte Carlo experiments are based on $2^{10}$ replications with nominal size $\alpha_0=5\%$.
Table \ref{tab:type1_ar1} reports empirical type-I error rate. 
Without prewhitening, all estimators exhibit size distortion in small samples, but the rejection frequencies decrease toward the nominal level as $n$ increases, indicating asymptotic size control. 
Prewhitening substantially improves size control, particularly in small samples.
Figure \ref{fig:power_pw} reports power against $\beta_0=\Delta/\sqrt{n}$. 
Prewhitening preserves power while improving size accuracy. 
Because the mean is not time-varying, the $t$ tests with all three long-run variance estimators perform very similarly. 
Overall, tight differencing works well with prewhitening, which provides 
a finite-sample improvement under strong persistence.

\begin{table}[t]
\centering
\captionsetup{font=small}
\caption{Empirical type-I error rate of the mean tests for strongly correlated data when $\beta(t) = 0$ for all $t$.
The nominal size is $\alpha_0=5\%$.
}
\label{tab:type1_ar1}
\begin{tabular}{lcccc}
Estimator $\hat{v}$ & $n=400$ & $n=800$ & $n=1600$ & $n=3200$ \\
Non-difference         & $0.1416$ & $0.1094$ & $0.0869$ & $0.0605$ \\
Loose-difference       & $0.1455$ & $0.1016$ & $0.0850$ & $0.0605$ \\
Tight-difference       & $0.1641$ & $0.1172$ & $0.0898$ & $0.0645$ \\
Prewhitened non-difference      & $0.0830$ & $0.0693$ & $0.0605$ & $0.0518$ \\
Prewhitened loose-difference    & $0.0811$ & $0.0635$ & $0.0596$ & $0.0527$ \\
Prewhitened tight-difference    & $0.0801$ & $0.0713$ & $0.0576$ & $0.0527$ 
\end{tabular}
\end{table}

\begin{figure}[t]
\centering 
\includegraphics[width=0.98\textwidth]{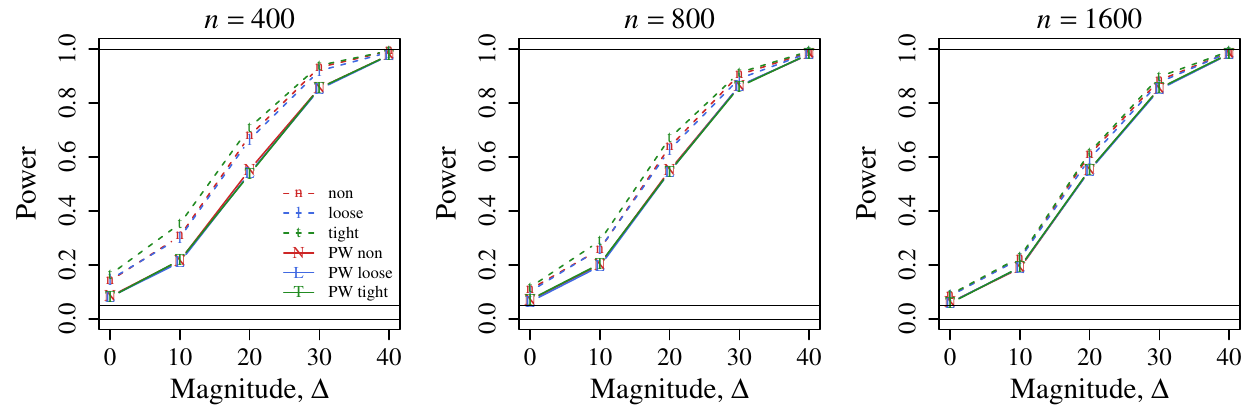}
\captionsetup{font=small}
\caption{Power curves of the mean tests for strongly correlated data when $\beta(t) = 0$. 
The power is plotted against $\Delta$, where $\beta_0=\Delta/\sqrt{n}$. 
The nominal size is $\alpha_0 = 5\%$. 
The dashed lines denote non-prewhitened estimators, while the solid lines denote prewhitened estimators. }
\label{fig:power_pw}
\end{figure}

Then, we consider a time-varying $\beta(t)
=5\left\{
0.4\mathbb 1(t>0.25)-\mathbb 1(t>0.5)
+0.8\mathbb 1(t>0.75)
\right\} +5(t^3-0.25)$, which follows the form of $\mu_{\textsc{f}}(t)$ in \S \ref{sec:lrvEst_nonconstMean} with $\Delta = 0.5$.
In this setting, the residuals used in the test remain
$\widehat Z_i=Y_i-\widehat\beta_0$.
Prewhitening is adapted to the non-constant mean by modifying only the estimation of
the auxiliary $\AR(1)$ coefficient. Specifically, we keep the construction in
\eqref{eqt:phihat_PW}, except that, when estimating $\widehat\phi$, the residual series
$\{\widehat Z_i\}$ is replaced by the locally differenced series
\[
D_i^{\pw}
=
\sum_{j=0}^{2} d_j^{\pw}Y_{i-jh_{\pw}},
\qquad 
h_{\pw}=\max\{\lceil 2n^{1/3}\rceil,5\},
\]
where $(d_0^{\pw},d_1^{\pw},d_2^{\pw})=(0.809,-0.5,-0.309)$ is the optimal
difference sequence of order $2$ in \citet{Hall1990}.
Keeping all other procedures the same as in the previous experiment, we report the empirical type-I error rate in Table \ref{tab:type1_ar1_mix_mean} and the corresponding power curves in Figure \ref{fig:power_pw_mix_mean}.
Under the time-varying mean, the non-differencing estimator has rejection frequencies far below the nominal level, reflecting its inability to remove mean variation. 
Loose differencing partially mitigates this issue but still leads to noticeable size distortion.
Tight differencing, particularly when combined with prewhitening, provides reliable inference 
with accurate size control and high power 
under both strong persistence and time-varying means.

\begin{figure}[t]
\centering 
\includegraphics[width=\textwidth]{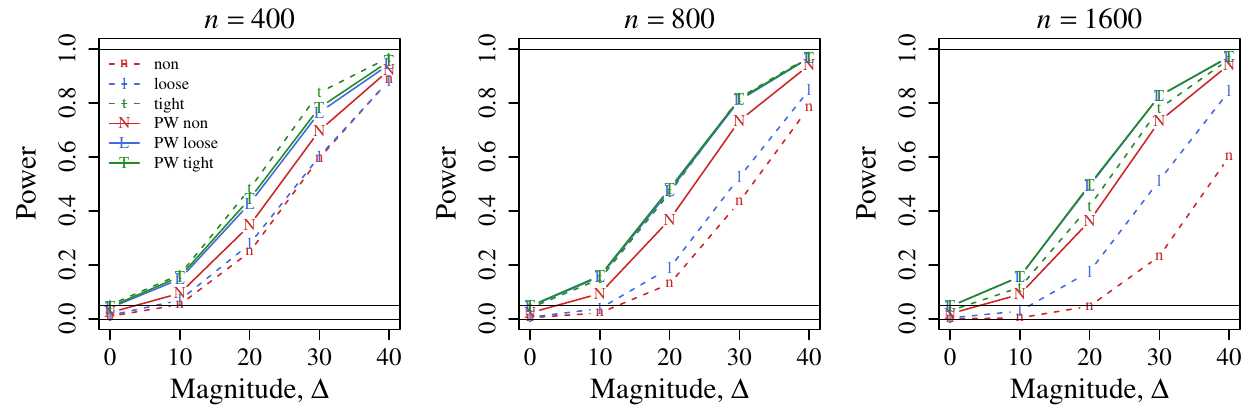}
\captionsetup{font=small}
\caption{Power curves of the mean tests for strongly correlated data with a centered time-varying mean.
The power is plotted against $\Delta$, where $\beta_0=\Delta/\sqrt{n}$. 
The nominal size is $\alpha_0 = 5\%$. 
The dashed lines denote non-prewhitened estimators, while the solid lines denote prewhitened estimators. }
\label{fig:power_pw_mix_mean}
\end{figure}

\begin{table}[t]
\centering
\captionsetup{font=small}
\caption{Empirical type I error (nominal level $\alpha_0=5\%$) of the $t$-tests for strongly correlated data with a centered time-varying mean.}
\label{tab:type1_ar1_mix_mean}
\begin{tabular}{lcccc}
Estimator & $n=400$ & $n=800$ & $n=1600$ & $n=3200$ \\
Non-difference                  & $0.0020$ & $0.0000$ & $0.0000$ & $0.0000$ \\
Loose-difference                & $0.0132$ & $0.0059$ & $0.0034$ & $0.0054$ \\
Tight-difference                & $0.0518$ & $0.0405$ & $0.0298$ & $0.0259$ \\
Prewhitened non-difference      & $0.0234$ & $0.0225$ & $0.0190$ & $0.0200$ \\
Prewhitened loose-difference    & $0.0410$ & $0.0459$ & $0.0459$ & $0.0483$ \\
Prewhitened tight-difference    & $0.0444$ & $0.0474$ & $0.0464$ & $0.0488$ 
\end{tabular}
\end{table}

\subsection{Real-data application} \label{sec:realdata_supp}
In this section, we analyze global mean sea level (GMSL) variations from 2012 to 2023 to test whether the linear trend of sea level rise is stationary over time. 
Prior research suggests an average GMSL rise rate of approximately $3.4$ millimeters per year since 1993, with evidence of acceleration in recent decades; similar findings can be found in \citet{NOAA2024}.
The dataset used for this analysis is the HDR Global Mean Sea Level Data, compiled at NASA’s Goddard Space Flight Center under the NASA Sea Level Change program,
integrating observations from multiple satellite altimetry missions (TOPEX/Poseidon, Jason-1, OSTM/Jason-2, Jason-3, and Sentinel-6 Michael Freilich Version 5.2. PO.DAAC, CA, USA) into a consistent reference framework. 
The dataset was accessed on 2024-12-23 at \href{https://doi.org/10.5067/GMSLM-TJ152}{https://doi.org/10.5067/GMSLM-TJ152}.

The primary variable of interest is the smoothed GMSL variation (in millimeters), which includes the Global Isostatic Adjustment (GIA) correction and has annual and semi-annual signals removed (provided in column 12 of the dataset). 
By narrowing the temporal scope to 2012–2023, this analysis focuses on more recent sea level dynamics but also faces the challenge of handling a smaller sample size, with $n=442$.

To assess the stationarity of the trend, we use a test analogous to the mean stationarity test described in Section \ref{sec:KPSS}. Specifically, we test the null hypothesis of trend linearity:
$H_{0}: X_i = \alpha + \beta i/n + Z_i$ for $i=1,\ldots ,n$,
where $X_i$ represents the observed GMSL variations, 
$\alpha$ is the intercept, $\beta$ captures the slope of trend, 
and $Z_i$ represents the stationary error process.
We consider the test statistic proposed by \cite{Hobijn2004}, i.e., 
$T = n^{-2}\sum_{t=1}^{n}(\sum_{i=1}^{t} R_i)^2/\hat{v}$, where 
$R_i$ is the residual of regressing $X_i$ on $i/n$ by the ordinary least squares procedure,
and the estimated long-run variance
$\hat{v}$ 
is obtained using either the proposed estimator $\hat{v}_{\tight}$ or the competing estimator $\hat{v}_{\loose}$, 
both of which incorporate optimal parameters under differencing order $m=3$ and kernel order $q=2$.
Based on the asymptotic distribution of the test statistic under $H_{0}$ (see \citet{Hobijn2004}), 
the $p$-value is computed based on a second level Brownian bridge simulated under $10^6$ Monte Carlo replications 
with sample size $n=442$.

\begin{figure}[!t] 
	\begin{center}
	\captionsetup{font=small}
		\includegraphics[width=0.95\linewidth]{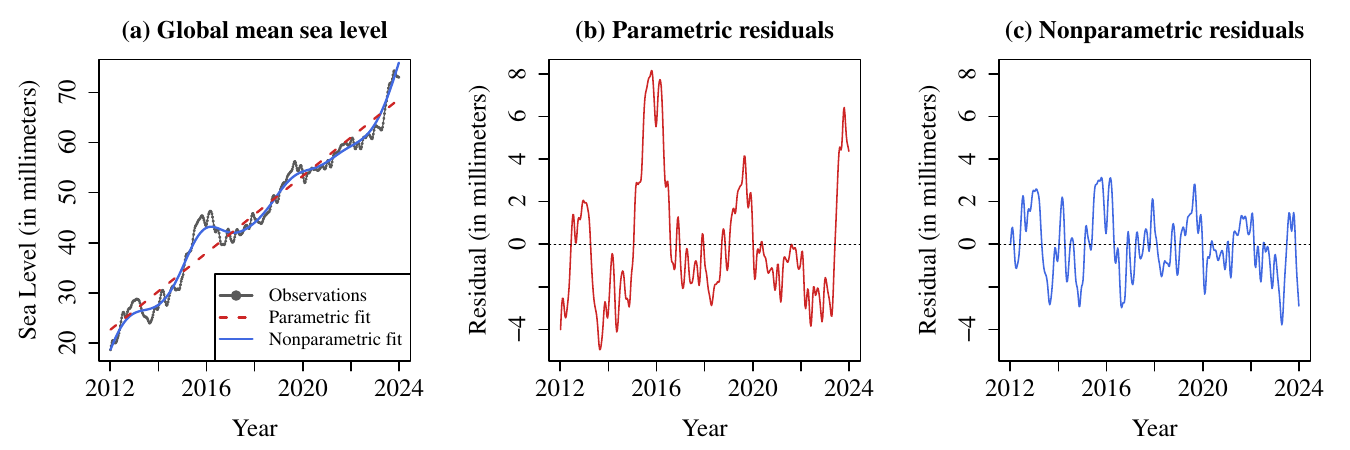} 
		\caption{
		Plot (a) shows the global mean sea level variations (in millimeters) from 2012 to 2023 with GIA correction 
		applied and annual and semi-annual signals removed. 
		The linear fit (i.e., parametric) and locally linear fit (i.e., nonparametric) of the sea levels are also plotted. 
		Plots (b) and (c) show the residuals based on the parametric model and nonparametric model, respectively. }
		\label{fig:seaLevel} 
	\end{center}   
\end{figure} 

Using our proposed estimator $\hat{v}_{\tight}$, 
the test yields a $p$-value of $0.004065$, 
strongly rejecting the null hypothesis and indicating significant non-linearity in the trend. 
In contrast, the competing estimator $\hat{v}_{\loose}$ 
produces a $p$-value of $0.454464$, 
failing to reject the null hypothesis at the conventional $5\%$ significance level.
We observed that the realized value of $\hat{v}_{\tight}$ is smaller than that of $\hat{v}_{\loose}$, 
which explains why two different testing results were obtained. 
In other words, since $\hat{v}_{\tight}$ is not significantly inflated due to the non-constant mean structure, 
the test statistic based on $\hat{v}_{\tight}$ is larger than that based on $\hat{v}_{\loose}$. 
This leads to higher power in rejecting the null hypothesis when there is non-linear mean structure. 

The linear fit (i.e., parametric) and locally linear fit (i.e., nonparametric) 
of the sea level trends are plotted in Figure \ref{fig:seaLevel}, 
where the parametric fit is based on ordinary least squares 
and the nonparametric fit is based on the procedure stated in \cite{wu_zhao_2007}, 
with bandwidth selected based on our proposed $ \hat{v}_{\tight}$.
The residuals based on the parametric fit and the nonparametric fit are also plotted for reference. 
Compared with the nonparametric residuals, it can be seen that 
there are still suspiciously non-constant trend structures around 2016 in the parametric residuals.
Such variation is a potential reason why $H_0$ should be rejected.

\section{Theoretical proofs of results in the main text}\label{sec:supp}
\subsection{Proof of Proposition \ref{prop:CE_Kdiff}}\label{pf:prop_CE_Kdiff}

Recall the definition of the differencing kernel.
\[
K_d(t)
= \sum_{s=\lceil -(1+t)\rceil}^{\lfloor 1-t\rfloor}
\delta_s K(t+s),
\qquad
\delta_{-s}=\delta_s, \quad \delta_0=1,
\]
where $K$ is a symmetric kernel supported on $[-1,1]$ with $K(0)=1$ and $K(-1)=K(1)=0$.

First, consider the case of $\lambda=1$. We consider different values of $t$.
\begin{itemize}
\item[(i)] $t=0$.
\[
K_d(0) =\sum_{s=\left \lceil -(1+0)/1 \right \rceil}^{\left \lfloor (1-0)/1 \right \rfloor} \delta_{|s|}K(0+s)
= \sum_{s=-1}^{1}\delta_s K(s)
= \delta_0 K(0)+\delta_1\{K(1)+K(-1)\}
= 1.
\]
\item[(ii)] $0< |t| <1$.
If $0<t<1$, then $s=-1,0$; if $-1<t<0$, then $s=0,1$.

In both cases, using the symmetry of $K$ and $\delta_s$, we obtain
\[
K_d(t)
=
K(t)+\delta_1 K(|t|-1),
\qquad
0<|t|<1.
\]

\item[(iii)] $|t|=1$. Consider $t=1$ as the case $t=-1$ follows by symmetry.
\[
K_d(1)
=\sum_{s=\left \lceil -(1+1)/1 \right \rceil}^{\left \lfloor (1-1)/1 \right \rfloor} \delta_{|s|}K(1+s)
= \sum_{s=-2}^{0}\delta_s K(1+s)
= \delta_{-2}K(-1) +  \delta_{-1}K(0) + \delta_{0}K(1)
= \delta_1.
\]
\end{itemize}
Combining (i)--(iii), for $\lambda=1$ we have
\[
K_d(t)=K(t)+\delta_1 K(|t|-1), \quad 0\leq|t|\leq1.
\]
Then, 
\begin{align*}
\CE(K_d) &=\sup\left \{q \in \mathbb{N} \cup \{ \infty\}:\lim_{t \downarrow 0} \frac{| K_d(t)-K_d(0)|}{|t|^q} < \infty \right \} \\ 
&= \sup\left \{q \in \mathbb{N} \cup \{ \infty\}:\lim_{t \downarrow 0} \frac{|\delta_1\{K(t-1)-K(1)\} + \{K(t) - K(0)\}|}{|t|^q} < \infty \right \}. 
\end{align*}
Note that for any $q\geq 0$,
\[
\frac{\left |\delta_1\{K(t-1)-K(1)\} + \{K(t) - K(0)\}\right |}{|t|^q} \leq    
    |\delta_1|\cdot \frac{\left |K(t-1)-K(1)\right |}{|t|^q} + 
        \frac{\left |K(t)-K(0)\right |}{|t|^q}.
\]
And recall that
\begin{align*}
\CE(K) &=\sup\left \{q \in \mathbb{N} \cup \{ \infty\}:\lim_{t \downarrow 0} \frac{| K(t)-K(0)|}{|t|^q} < \infty \right \}, \\
\BCE(K)&=\sup\left \{q^{\prime} \in \mathbb{N} \cup \{ \infty\}: \lim_{t \downarrow 0} \frac{| K(1)-K(1-t)| }{ |t|^{q^{\prime}}} < \infty \right \}.
\end{align*}
So, we obtain $$\CE(K_d) \geq \min\{ \CE(K) , \BCE(K) \} = \min(q,q^{\prime}).$$
Conversely, if $q>\CE(K)$, then ${|K(t)-K(0)|}/{|t|^q}\to\infty$,
implying ${|K_d(t)-K_d(0)|}/{|t|^q}\to\infty$.
If $q>\BCE(K)$ and $\delta_1\neq 0$, then ${|K(1-|t|)-K(1)|}/{|t|^q}\to\infty$, again forcing divergence.
Therefore, when $\delta_1\neq 0$,
\[
\CE(K_d)
=\min\{\CE(K),\BCE(K)\}
=\min(q,q').
\]

Then, we consider the case of $\lambda\in (1,2)$.
For any integer $s$ with $|s|\ge2$,
\[
|t+\lambda s| \ge \lambda|s|-|t| \ge 2\lambda-1 > 1,
\]
and thus $K(t+\lambda s)=0$. Hence only the indices $s\in\{-1,0,1\}$ may contribute.
Moreover, for $t\in[-1,1]$,
\[
-\frac{1+t}{\lambda}\le 0 \le \frac{1-t}{\lambda},
\]
which implies that $s=0$ always lies in the summation range. 
We next identify when $s=\pm1$ are included. The index $s=-1$ appears if and only if
\[
-1 \ge -\frac{1+t}{\lambda}
\quad\Longleftrightarrow\quad
t\ge \lambda-1,
\]
whereas the index $s=1$ appears if and only if
\[
1 \le \frac{1-t}{\lambda}
\quad\Longleftrightarrow\quad
t\le 1-\lambda.
\]
Consequently,
\[
K_d(t)=
\begin{cases}
K(t)+\delta_1 K(t+\lambda), & -1\le t\le 1-\lambda,\\
K(t), & 1-\lambda<t<\lambda-1,\\
K(t)+\delta_1 K(t-\lambda), & \lambda-1\le t\le 1,
\end{cases}
\qquad t\in[-1,1].
\]
Then, choose $\varepsilon\in(0,\lambda-1)$. Then for all $|t|<\varepsilon$, we have
\[
K_d(t)-K_d(0)=K(t)-K(0).
\]
Recalling
\[
\CE(K_d)
=\sup\Bigl\{q\in\mathbb N\cup\{\infty\}:\ \lim_{t\to 0}\frac{|K_d(t)-K_d(0)|}{|t|^q}<\infty\Bigr\},
\]
the above identity implies that, for every $q$,
\[
\lim_{t\to 0}\frac{|K_d(t)-K_d(0)|}{|t|^q}<\infty
\quad\Longleftrightarrow\quad
\lim_{t\to 0}\frac{|K(t)-K(0)|}{|t|^q}<\infty.
\]
Hence $\CE(K_d)=\CE(K)$.

Finally, we consider the case of $\lambda\in [2, \infty)$. Since $\delta_0=1$, we have
\begin{align*}
K_d(t) &= \sum_{s=\left \lceil -(1+t)/\lambda \right \rceil}^{\left \lfloor (1-t)/\lambda \right \rfloor} \delta_{|s|}K(t+ \lambda s) = \sum_{s=0}^{0}\delta_{|s|}K(t+ \lambda s) = \delta_0 K(t) = K(t). 
\end{align*}
Thus,  $\CE(K_d) =\CE(K)=q.$

\subsection{Proof of Proposition \ref{prop:flatness_ock}}\label{pf:prop_flatness_ock}
By Proposition \ref{prop:CE_Kdiff}, when $\lambda=1$, we have
$$\CE(K_d) = \min\{ \CE(K) , \BCE(K) \} = \min(q,q^{\prime}).$$ 
By the definition of an origin-characterized kernel $K \in \mathcal{K}_{q,q^{\prime}}$ with $q\leq q^{\prime}$, we have $\CE(K_d) =  \CE(K)=q$.

Then, we derive the near-origin flatness of $K_d$. Note that 
\begin{align*}
B_d &= \displaystyle \lim_{t \downarrow 0} \frac{K_d(t)-K_d(0)}{|t|^q} \\
		  &= \lim_{t \downarrow 0} \frac{\sum_{s=\left \lceil -(1+t) \right \rceil}^{\left \lfloor 1-t \right \rfloor} \delta_{|s|}K(t+s) -1}{|t|^q} \\
		  &= \lim_{t \downarrow 0} \frac{\delta_0 K(t) + \delta_1 K(t-1) -1}{|t|^q} \\
		  &= \delta_0 \lim_{t \downarrow 0} \frac{K(t)-K(0)}{|t|^q} - \delta_1  \lim_{t \downarrow 0}\frac{ K(1)-K(1-t)}{|t|^{q^{\prime}}}\cdot |t|^{q^{\prime}-q}\\
		  &= B - \delta_1 B^{\prime}\mathbb{1}(q=q^{\prime}),
\end{align*}
which concludes the proof.

\subsection{Proof of Theorem \ref{thm:MSE_propsal_lrv}}\label{pf:thm_MSE(K)}
Theorem \ref{thm:MSE_propsal_lrv} is a special case of Theorem \ref{thm:MSE_proposal_density}, which is proved in \S\ref{pf:thm_MSE(f)},
since $\hat{v}_{\tight} = 2 \pi \hat{f}_{\tight}(0)$ is just a special case of $\hat{f}_{\tight}(\theta)$.
Thus, the proof is omitted. 

\subsection{Proof of Theorem \ref{thm:clt_vhat}} \label{pf:thm_clt_vhat}
We drop the subscript ${}_{\tight}$ in this proof to lighten notation. 
First, by the re-scaled differencing kernel $\tilde{K}_d(\cdot)$ in (\ref{eq:K_tilde}), 
we can rewrite $\check{v} = \sum_{|k|\leq \ell+mh} K_d \left ( \frac{k}{\ell} \right ) \hat{\gamma}_k^X$ as
\begin{align} \label{eqt:1_pf_clt_v}
    \check{v} 
            &= \sum_{|k|\leq L} \tilde{K}_d \left ( \frac{k}{L} \right ) \hat{\gamma}_k^X \nonumber \\ 
            &= \sum_{|k|\leq L} \tilde{K}_d \left ( \frac{k}{L} \right ) \frac{1}{n} 
                    \sum_{i=|k|+1}^n (X_i - \bar{X}_n) (X_{i-|k|} - \bar{X}_n) \nonumber \\ 
            &= \sum_{i,j=1}^n \frac{1}{n} \tilde{K}_d \left ( \frac{|i-j|}{L} \right ) 
                    \mathbb{1}\left(|i-j|\leq L \right)(X_i - \bar{X}_n) (X_j - \bar{X}_n) \nonumber \\ 
            &= \sum_{i=1}^n a_n(i,i) (X_i - \bar{X}_n)^2 
                    + 2 \sum_{1\leq i<j\leq n}a_n(i,j) (X_i - \bar{X}_n) (X_j - \bar{X}_n) \nonumber \\ 
            &= \frac{1}{n} \sum_{i=1}^n (X_i - \bar{X}_n)^2 
                    + 2 \sum_{1\leq i<j\leq n}a_n(i,j) (X_i - \bar{X}_n) (X_j - \bar{X}_n),
\end{align}
where $a_n(i,j) \equiv \tilde{K}_d (|i-j|/L)/n \cdot \mathbb{1}\left(|i-j|\leq L \right)$.
Then, we define a martingale-approximated process $M_i$ as follows. Let
\begin{align*}
    \tilde{X}_i \equiv \E(X_i \mid \epsilon_{i-\psi}, \ldots, \epsilon_i ), \quad 
    m_i \equiv \sum_{k=0} ^\infty \E\left(\tilde{X}_{i+k} \mid \mathcal{F}_{i} \right), \quad \text{and }
    M_i\equiv m_i - \E(m_i \mid \mathcal{F}_{i-1}),
\end{align*}
where $\psi$ may depend on $n$.
And a martingale-approximated estimator can be
\[
    \tilde{v} = \frac{1}{n} \sum_{i=1}^n M_i^2 
                    + 2 \sum_{1\leq i<j\leq n}a_n(i,j) M_i M_j.
\]
Then, by defining $\varrho_n^2 = \sum_{t=2}^n \sum_{j=1}^{t-1} a_n^2 (j,t)$,
we can obtain that 
$\left\| \check{v} -\E(\check{v}) - \tilde{v} +\E(\tilde{v})\right\| = o(\varrho_n)$,
which is proved in B.2 Results in Section 3 of \cite{liuchan2025},
as a modified version of the proof in \citet{WuShao2007}.
For the completeness of the proof, we state the details here.
First, define a variant of $\hat{v}$ by assuming the mean $\mu$ is known,
i.e.,
\[
    \bar{v} = \frac{1}{n} \sum_{i=1}^n (X_i - \bar{X}_n)^2 
        + 2 \sum_{1\leq i<j\leq n}a_n(i,j) (X_i - \mu) (X_j - \mu).
\]
By Minkowski inequality, we have
\[
    \left \| \check{v} -\E(\check{v}) - \bar{v} +\E(\bar{v}) \right \|
        \leq
    \left \| \check{v} -\bar{v}\right \| +  |\E(\check{v}) -\E(\bar{v})|.
\]
First, we compute $\varrho_n^2$ by defining $A_{t,n} \equiv \sum_{j=1}^{t-1} a_n^2(j,t)$ 
and thus $\varrho_n^2 = \sum_{t=2}^n A_{t,n} $.
If $t> L$, we have
\begin{align*}
    A_{t,n} &= \sum_{j=1}^{t-L-1} a_n^2(j,t) + \sum_{j=t-L}^{t-1} a_n^2(j,t) 
        = 0 + \sum_{k=1}^L \frac{1}{n^2} \tilde{K}_d^2 \left ( \frac{k}{L} \right ) \mathbb{1}\left(k\leq L \right) = \frac{1}{n^2}  \sum_{k=1}^L \tilde{K}_d^2 \left ( \frac{k}{L} \right ).
\end{align*}
And if $t \leq L$, $A_{t,n} = \sum_{k=1}^{t-1} \tilde{K}_d^2 (k/L) /n^2$.
Then, the quantity
\begin{align*}
    \varrho_n^2 &= \sum_{t=2}^n A_{t,n}  
        = \sum_{t=2}^L A_{t,n} +\sum_{t=L+1}^n A_{t,n} \\
        &=  \sum_{t=2}^L \sum_{k=1}^{t-1} \frac{1}{n^2} \tilde{K}_d^2 \left ( \frac{k}{L} \right ) 
            + \sum_{t=L+1}^n \sum_{k=1}^L \frac{1}{n^2} \tilde{K}_d^2 \left ( \frac{k}{L} \right ) \\
        &= o\left(\frac{L^2}{n^2} \right) + \frac{1}{n^2} \sum_{t=L+1}^n \left\{\sum_{k=1}^L \tilde{K}_d^2 \left ( \frac{k}{L} \right ) \right\} \\
        &\sim \frac{1}{n^2}  \cdot n \cdot L \int_0^1 \tilde{K}_d^2 (t) \dd t  
        = \frac{\ell}{n} A_d,
\end{align*}
where $A_d = 2 \kappa_1 \sum_{s=1}^{m}\delta_{s-1}\delta_s + 
                \kappa_2 \sum_{|s|\leq m}\delta_s^2$ is defined in Theorem \ref{thm:MSE_propsal_lrv}.
Then, we use Lemma \ref{lemma:widehat_bar}
to prove that $\left \| \check{v} -\bar{v}\right \| = o(\varrho_n)$.

\begin{lemma}{(Modified version of Lemma F.1 in \citet{leung2026principles})}. \label{lemma:widehat_bar}
Let $\nu > 2$ and $X_1 \in \mathcal{L}^{\nu}$. If $\Delta_{\nu} < \infty$, then
\[
    \left \| \check{v} -\bar{v}\right \|_{\nu/2} = 
        O(H_{1,n}) + O(n^{-1/2} H_{2,n}^{1/2}) + O(n^{-1} H_{3,n}),
\]
where
\begin{align*}
    H_{1,n} &= \max_{1\leq i \leq n} |a_n(i,i)|,\\
    H_{2,n} &= \sum_{i=2}^n \left |\sum_{j=1}^{i-1} a_n(i,j) \right |^2 
            + \sum_{j=1}^{n-1} \left |\sum_{i=j+1}^n a_n(i,j) \right |^2, \\ 
    H_{3,n} &= \left |\sum_{i=2}^n \sum_{j=1}^{i-1} a_n(i,j) \right |.
\end{align*}
\end{lemma}

Note that, in our case, we have 
\begin{align*}
    H_{1,n} &= \frac{1}{n} =  o(\varrho_n/\sqrt{n}),  \\
    H_{2,n} &\leq 2n \left \{\max_{1\leq i \leq n} \sum_{j=1}^n |a_n(i,j)| 
    \right\}^2 
        \leq 2n \max_{1\leq i \leq n}  \sum_{k=-L}^L  \frac{1}{n^2} \tilde{K}_d^2 \left (\frac{k}{L} \right ) 
        \sim 2n \frac{2\ell A_d}{n^2} = O(n) o(\varrho_n^2), \\
    H_{3,n} &\leq n \max_{1\leq i \leq n} \sum_{i=1}^n |a_n(i,j)| 
        \leq n \max_{1\leq i \leq n} \sum_{k=-L}^L  \frac{1}{n} \tilde{K}_d \left (\frac{k}{L} \right ) 
        = O(\ell)
        = n \cdot o\left (\frac{\ell}{n}\right )^{1/2}
        = O(n) o(\varrho_n).
\end{align*}
Therefore, with $X_1 \in \mathcal{L}^4$ and $\Theta_4 < \infty$, 
by Lemma \ref{lemma:widehat_bar}, we get 
$\left \| \check{v} -\bar{v}\right \| = o(\varrho_n)$.
And thus, 
\[
    |\E(\check{v}) -\E(\bar{v})| \leq E(|\check{v} - \bar{v} |) 
        \leq \left \| \check{v} -\bar{v}\right \| = o(\varrho_n).
\]
So, we obtain $\left \| \check{v} -\E(\check{v}) - \bar{v} +\E(\bar{v}) \right \| = o(\varrho_n)$.

Without loss of generality, we assume $\mu=0$ for the remaining part of the proof. 
Define an approximated version for $\bar{v}$ to be
\[
    \bar{\bar{v}} = \frac{1}{n} \sum_{i=1}^n \tilde{X}_i^2 
        + 2 \sum_{1\leq i<j\leq n}a_n(i,j) \tilde{X}_i \tilde{X}_j.
\]
Then, we first prove that $\bar{\bar{v}}$ is an approximation of $\bar{v}$ 
and then prove that $\tilde{v}$ is an approximation of $\bar{\bar{v}}$.
Similarly by Minkowski inequality, we have
\[
    \left \| \bar{v} -\E(\bar{v}) - \tilde{v} +\E(\tilde{v}) \right \|
        \leq
    \left \| \bar{v} -\E(\bar{v}) - \bar{\bar{v}} +\E(\bar{\bar{v}}) \right \|
        +
    \left \| \bar{\bar{v}} - \E(\bar{\bar{v}}) - \tilde{v} +\E(\tilde{v}) \right \|.
\]
We now prove that $\left \| \bar{v} -\E(\bar{v}) - \bar{\bar{v}} +\E(\bar{\bar{v}}) \right \| = o(\varrho_n)$ by Lemma \ref{lemma:barbar_bar}.

\begin{lemma}{(Modified version of Lemma F.2 in \citet{leung2026principles})}.
\label{lemma:barbar_bar}
    Let $\nu > 2$ and $X_1 \in \mathcal{L}^{\nu}$. 
    Also let $\nu ' = \min(\nu/2, 2)$. 
    If $\Delta_{\nu} < \infty$, then
    \[
        \left \| \bar{v} -\E(\bar{v}) - \bar{\bar{v}} +\E(\bar{\bar{v}}) \right \|_{\nu/2} = O(n^{1/\nu '}H_{1,n} ) + O(n^{1/\nu '} d_{\psi, \nu}H_{4,n}^{1/2}),
    \]
    where $H_{1,n}$ is defined in Lemma \ref{lemma:widehat_bar}, 
    \[
        d_{\psi, \nu} = \sum_{i=0}^\infty \min \left ( \theta_{\nu, i} \sqrt{\sum_{i=\psi+1}^\infty \theta_{\nu, i}^2}\right) 
        \qquad \text{and} \qquad
        H_{4,n} = \max_{1\leq i \leq n} \sum_{j=1}^n a_n^2(i,j) .
    \]
\end{lemma}
Note that, in our case, we have  
\[
    H_{4,n} \leq \max_{1\leq i \leq n}  \sum_{k=-L}^L  \frac{1}{n^2} \tilde{K}_d^2 \left (\frac{k}{L} \right ) 
    \sim \frac{2\ell A_d}{n^2} = O\left ( \frac{\varrho_n^2}{n} \right).
\]
Thus, we have $H_{4,n}^{1/2} = o(n^{-1/2} \varrho_n)$. 
Therefore, letting $\nu=4$ and by Lemma \ref{lemma:barbar_bar}, we have 
$
    \left \| \bar{v} -\E(\bar{v}) - \bar{\bar{v}} +\E(\bar{\bar{v}}) \right \| = o(\varrho_n).
$

Next, we prove that $\left \| \bar{\bar{v}} - \E(\bar{\bar{v}}) - \tilde{v} +\E(\tilde{v}) \right \| = o(\varrho_n)$ 
by Lemma \ref{lemma:tilde_barbar}.
\begin{lemma} {(Modified version of Lemma F.3 in \citet{leung2026principles})}. 
\label{lemma:tilde_barbar}
    Let $\nu > 2$ and $X_1 \in \mathcal{L}^{\nu}$. 
    Also let $\nu ' = \min(\nu/2, 2)$. 
    If $\Delta_{\nu} < \infty$, then
    \[
        \left \| \bar{\bar{v}} - \E(\bar{\bar{v}}) - \tilde{v} +\E(\tilde{v}) \right \|_{\nu /2}
        = O(\psi^{1-1/\nu'} n^{1/\nu'} H_{1,n}) + O(\psi^{3-1/\nu'} n^{1/\nu'} H_{5,n}),
    \]
    where $H_{1,n}$ is defined in Lemma \ref{lemma:widehat_bar} and
    \begin{align*}
        H_{5,n} &= \max_{1\leq i,j \leq n}|a_n(i,j)| + \max_{1\leq i \leq n} \left [ \sum_{j=2}^n  \left \{ 2|a_n(i,j) - a_n(i,j-1) | ^{\nu'} \right \} \right]^{1/\nu'}.
    \end{align*}
\end{lemma}

Note that letting $\nu=4$, we have $\nu' = 2$ and
\begin{align*}
    \frac{1}{2} H_{5,n} 
    &\leq  
    \max_{1\leq i,j \leq n}|a_n(i,j)| + \max_{1\leq i \leq n} \left \{ \sum_{j=2}^n   |a_n(i,j) - a_n(i,j-1) | ^2  \right\}^{1/2} \\
    &\leq
    \frac{1}{n} + \frac{2}{n} \max_{1\leq i \leq n} \left\{ \sum_{k=1}^L \left |\tilde{K}_d \left (\frac{k}{L} \right ) - \tilde{K}_d \left (\frac{k-1}{L} \right ) \right |^2 \right \}^{1/2}.
\end{align*}

Assume there exists a constant $P \in (0,1)$ such that $\lim_{\ell \to \infty} \ell^P |\tilde{K}_d \left (k/L \right ) - \tilde{K}_d \{ (k-1)/L \} |=0$
since $\tilde{K}_d(t)$ is Lipschitz continuous in $[-1,1]$. Thus, we have
\[
    \frac{1}{2} H_{5,n} \leq \frac{1}{n} + o\left ( \frac{\ell^{1/2-P}}{n} \right ).
\]
If $\ell=O(n^{\vartheta})$ for some $\vartheta >0$, 
there exists a constant $c \in (1/2, 1/2+P\vartheta)$ 
such that $H_{5,n} = o(\varrho_n n^{-c})$.
Then, we select $\psi=O\{ \min (n^{(-1+2c)/5}, n^{2k}) \}$ 
such that $\psi \to \infty$ as $n \to \infty$ for the $\psi$-dependent approximation to hold.
Therefore, we obtain 
\[
\left \| \bar{\bar{v}} - \E(\bar{\bar{v}}) - \tilde{v} +\E(\tilde{v}) \right \| = o(\varrho_n).
\]

Finally, we can conclude that $\left\| \hat{v} -\E(\hat{v}) - \tilde{v} +\E(\tilde{v})\right\| = o(\varrho_n)$ as  $\left\| \hat{v} -\E(\hat{v}) - \check{v} +\E(\check{v})\right\| = o(\varrho_n)$ as well \citep{Chan2022}.
And thus,
\begin{align} \label{eq:martingale_approx}
\frac{\hat{v} -\E(\hat{v})} {\varrho_n} - \frac{\tilde{v} -\E(\tilde{v})} {\varrho_n} \inP 0.
\end{align}
Then, it suffices to consider the martingale-approximated version:
\[
    \tilde{v} -\E(\tilde{v}) 
        = \frac{1}{n}\sum_{i=1}^n (M_i^2 - \left\| M_0\right\|^2) 
        + 2 \sum_{1\leq i<j\leq n} a_n (i,j)M_i M_j.
\]
Let $T_n \equiv \sum_{1\leq i<j\leq n} a_n (i,j)M_i M_j$.
By Theorem 1 in \citet{WuShao2007}, to prove that $T_n/\varrho_n \inD \Normal(0, \left \|M_0 \right\|^4)$,
we need to check conditions (i)--(iv) stated in the theorem.
Condition (i) can be easily verified as
\[
\max_{2\leq t \leq n} A_{t,n} = \frac{1}{n^2}  \sum_{k=1}^L \tilde{K}_d^2 \left ( \frac{k}{L} \right )
    \sim \frac{1}{n^2} \cdot \ell A_d = o(\varrho_n^2).
\]
Condition (ii) also holds since
\begin{align*}
\sum_{j=1}^{n-J} a_n^2 (j, j+J) 
    & = \sum_{j=1}^{n-J} \frac{1}{n^2} \tilde{K}_d^2 \left ( \frac{J}{L} \right )  \mathbb{1}\left(J\leq L \right) \\
    & = \frac{n-J}{n^2}\tilde{K}_d^2 \left ( \frac{J}{L} \right )  \mathbb{1}\left(J\leq L \right)   
    = O \left ( \frac{1}{n} \right )  
    = o(\varrho_n^2).
\end{align*}
Condition (iii) can be checked by
\begin{align*}
    n \sum_{t^\prime=1}^{n-1} B_{t',n}^2 
    &= n \sum_{t^\prime=1}^{n-1} \left\{ \sum_{t=t'+1}^n a_n^2 (t',t)\right \}^2 \\
    &=  n \sum_{t^\prime=1}^{n-1} \left\{ \sum_{t=t'+1}^n \frac{1}{n^2} \tilde{K}_d^2 \left ( \frac{|t-t'|}{L} \right ) \mathbb{1}\left(|t-t'|\leq L \right)  \right \}^2 \\
    &= n \sum_{t^\prime=1}^{n-L} \left\{ \frac{1}{n^2}  \sum_{k=1}^L \tilde{K}_d^2 \left ( \frac{k}{L} \right ) \right \}^2 + n \sum_{t^\prime=n-L+1}^{n}  \left\{ \frac{1}{n^2}  \sum_{k=1}^{n-t'} \tilde{K}_d^2 \left ( \frac{k}{L} \right ) \right \}^2 \\
    &\sim n \sum_{t^\prime=1}^{n-1} \left ( \frac{\ell A_d}{n^2} \right)^2 
    \sim \left ( \frac{\ell A_d}{n} \right)^2 
    = O(\varrho_n^4).
\end{align*}
Finally, we verify condition (iv) by
\begin{align*}
    \sum_{s=1}^{n-1} \sum_{t=1}^{s-1} \left \{ \sum_{j=1+s}^n a_n(s,j) a_n(t,j) \right\}^2
     &= \sum_{s=1}^{n-1}  \sum_{t=1}^{s-1} \left \{\sum_{j=1+s}^{\max(n,s+L)} \frac{1}{n} \tilde{K}_d \left ( \frac{|j-s|}{L} \right ) a_n(t,j) \right\}^2 \\
     &= \sum_{s=1}^{n-1} \sum_{t=\max(1,1+s-L)}^{s-1} \left \{\sum_{j=1+s}^{\max(n,s+L)} \frac{1}{n} \tilde{K}_d \left ( \frac{|j-s|}{L} \right ) a_n(t,j) \right\}^2 \\
     &\leq \sum_{s=1}^{n-1} L \cdot \left ( \frac{\ell A_d}{n^2}\right)^2 
     = O\left ( \frac{\ell^3}{n^3}\right) = o(\varrho_n^4).
\end{align*}
And note that since $M_0 = m_0- \E(m_0 | \mathcal{F}_{-1}) 
= \sum_{k=0}^\infty \{\E(\tilde{X}_k | \mathcal{F}_{0}) -  \E(\tilde{X}_k | \mathcal{F}_{-1})\}
= \sum_{k=0}^\infty \{\E({X}_k | \mathcal{F}_{0}) -  \E({X}_k | \mathcal{F}_{-1})\}$,
we have $\left\|M_0 \right \|^2 = v$.
Therefore, we can conclude that 
\begin{align} \label{eq:convg_Tn}
T_n / \varrho_n \inD \Normal(0, v^2).
\end{align}
Then, denote $\Omega_n = \sum_{i=1}^n (M_i^2 - \left\| M_0\right\|^2)/n$. 
Following the idea in \citet{WuShao2007}, we plan to prove that $\Omega_n / \varrho_n \inP 0$.
Define the projection operator as 
$\mathcal{P}_j(\cdot) = \E (\cdot | \mathcal{F}_j) - \E(\cdot | \mathcal{F}_{j-1})$. Then,
\begin{align*}
    \left \| \sum_{i=1}^n \left (M_i^2 - \left\| M_0\right\|^2 \right )  \right \|
    &= \left \| \sum_{i=1}^n \left ( \sum_{k=0}^{\infty} \mathcal{P}_{i-k} M_i^2 + \left \| M_i \right \|^2 - \left \| M_0 \right \|^2 \right ) \right \| \\
    &\leq \sum_{k=0}^{\infty} \left \| \sum_{i=1}^n \mathcal{P}_{i-k} M_i^2\right \|  
    \leq  \sqrt{n}  \sum_{k=0}^{\infty} \left \| \mathcal{P}_0 M_k^2 \right\|.
\end{align*}
The last inequality holds by Burkholder's moment inequality for martingale differences \citep{Burkholder1988}.
Next, we define $M_{k, \{0\}} = G_{k,0} - \E (G_{k,0}|\mathcal{F}_{k-1}')$
where $$G_{k,0}= \sum_{i=kL+1}^{(k+1)L}\sum_{j=i-L}^{i-1}\tilde{K}_d \{(i-j)/L \}X_iX_j.$$ 
Thus,
\begin{align*}
    \mathcal{P}_0 M_k^2 &= \E(M_k^2 | \mathcal{F}_0) - \E(M_k^2 | \mathcal{F}_{-1}) 
    = \E \{M_k^2 - M_{k, \{ 0\}}^2 |\mathcal{F}_0\} \\
    &= \E \{ (M_k + M_{k, \{ 0\}}) (M_k - M_{k, \{ 0\}}) |\mathcal{F}_0 \}.
\end{align*}
Therefore, we have
\[
    \left \| \mathcal{P}_0 M_k^2 \right\| 
    =  \left \| (M_k + M_{k, \{ 0\}}) (M_k - M_{k, \{ 0\}}) \right\|
    \leq \left \| M_k + M_{k, \{ 0\}} \right \|_4 \left \| M_k - M_{k, \{ 0\}} \right \|_4,
\]
which holds by H$\ddot{o}$lder's inequality. Finally, we get
\begin{align*}
    \left \| \Omega_n \right \| &= \frac{1}{n} \sum_{i=1}^n (M_i^2 - \left\| M_0\right\|^2) 
    \leq \frac{1}{n}   \sqrt{n}  \sum_{k=0}^{\infty} \left \| \mathcal{P}_0 M_k^2 \right\| \\
    &\leq \frac{1}{\sqrt{n}} \sum_{k=0}^{\infty} \left \| M_k + M_{k, \{ 0\}} \right \|_4 \left \| M_k - M_{k, \{ 0\}} \right \|_4 \\
    &\leq \frac{1}{\sqrt{n}} \sum_{k=0}^{\infty}\sum_{t=0}^{\infty}  \left \| X_{t+k} - X_{t+k, \{ 0\}} \right \|_4 
    = o(\varrho_n).
\end{align*}
The last inequality holds since $\sum_{i=1}^{\infty} i \theta_{\nu, i} = \sum_{i=1}^{\infty} i \left \|X_i - X_{i, \{0\}} \right \|_\nu < \infty$ for some $\nu >4$.

Therefore, with $\Omega_n/ \varrho_n \inP 0$, we can further conclude from (\ref{eq:convg_Tn}) that 
\[
\frac{\tilde{v}-\E(\tilde{v})} {\varrho_n} \inD  \Normal(0, 4v^2),
\]
which together with (\ref{eq:martingale_approx}) implies that 
\[
    \frac{\hat{v} - \E(\hat{v})}{\varrho_n} \inD  \Normal(0, 4v^2)
    \quad \text{where}\quad \varrho_n^2 \sim \frac{\ell A_d}{n}.
\]

\subsection{Proof of Proposition \ref{prop:l_lrv}}\label{sec:proof_prop_optimal_d_l} 
The minimizer $\bar{\xi}^*(K, d_{0:m})$ can be derived by setting the derivative of the objective function (\ref{eq:MSE(K)}) to zero and 
checking the second order condition.   
Consequently, the optimal mean-squared error satisfies 
\[
	n^{2q/(1+2q)}{\MSE}_0\{\hat{v}_{\tight}(K, \ell, d_{0:m})\} \rightarrow (1+2q) v^2 \left\vert v_q/v  \right\vert^{2/(1+2q)}\left \{ |B_d|  ( 2 A_d /q )^{q}  \right \} ^{2/(1+2q)}. 
\]
Note that the difference sequence only affects $A_d$ and $B_d$. 
So, to find $d^*_{0:m}(K)$, it suffices to minimize $|B_d|  ( 2 A_d /q )^{q}$. 
By (\ref{eqt:ABd}), we obtained the desired result.

\subsection{Proof of Proposition \ref{prop:robust}}\label{pf:prop_robust} 
(1) It is the purely continuous trends case. 
The mean function is $\mu(i/n)=\mathcal{C}_0 + \mathcal{C} i/n$. Then, we get 
\[
D_i^{\mu} = \sum_{j=0}^m d_j \left(\mathcal{C}_0 + \mathcal{C} \frac{i-jh}{n} \right )
    = -\frac{h\mathcal{C}}{n} \sum_{j=0}^m j d_j.
\]
Thus,
\begin{align*}
    \hat{v}^{\mu\mu} &= \frac{1}{n} \sum_{|k| \leq \ell} K \left ( \frac{k}{\ell} \right)
                \sum_{i=mh+|k|+1}^n D_i^{\mu} D_{i-|k|}^{\mu} \\
        &= \frac{\mathcal{C}^2 h^2}{n^3} \sum_{|k| \leq \ell} K \left ( \frac{k}{\ell} \right)
            \sum_{i=mh+|k|+1}^n  \left(- \sum_{j=0}^m j d_j  \right) \left(- \sum_{j=0}^m j d_j \right) \\
        &= \frac{\mathcal{C}^2 h^2}{n^3}  \left(\sum_{j=0}^m j d_j \right)^2 \sum_{|k| \leq \ell} \{n-(mh+|k|+1)\} K \left ( \frac{k}{\ell} \right)   \\
        &= \frac{\mathcal{C}^2 h^2 \ell}{n^2}  \left(\sum_{j=0}^m j d_j \right)^2 \sum_{|k| \leq \ell} \left \{\frac{n-(mh+|k|+1)}{n} \right \} K \left ( \frac{k}{\ell} \right) \frac{1}{\ell}  \\
        &= \mathcal{C}^2 \lambda^2 \xi^3 n^{3\vartheta-2}  \left(\sum_{j=0}^m j d_j \right)^2 \sum_{|k| \leq \ell} \left (1- \frac{mh+|k|+1}{n} \right ) K \left ( \frac{k}{\ell} \right) \frac{1}{\ell}.
\end{align*}
Therefore, as $n \to \infty$,
\[
\frac{\hat{v}^{\mu\mu} }{n^{3\vartheta-2}} \sim 2 \mathcal{C}^2 \lambda^2 \xi^3 \left(\sum_{j=0}^m j d_j \right)^2 \int_{0}^{1} K(t) \dd t 
= 2 \mathcal{C}^2 \xi^3 \kappa_0 \lambda^2 \left(\sum_{j=0}^m j d_j \right)^2.
\]

(2) It is the case of piecewise constant trends. Recall that $\mu(t) = \sum_{s=0}^{\mathcal{J}}\iota_s\mathbb{1}(T_s/n \leq t < T_{s+1}/n)$. By the assumption that $\mathcal{G} \gtrsim \ell+mh$ and the constraint $\sum_{j=0}^m d_j = 0$, we can find that all $D_i^{\mu}$ with $i=mh+|k|+1,\ldots, n$, except for those with $s=1,\ldots,\mathcal{J}$, 
\begin{align*}
    D_{T_s}^{\mu} &= D_{T_s +1}^{\mu} = \cdots = D_{T_s +h-1}^{\mu} = d_0 \iota_s, \\
    D_{T_s+h}^{\mu} &= D_{T_s +h+1}^{\mu} =\cdots =D_{T_s +2h-1}^{\mu} =\iota_s\sum_{j=0}^1 d_j, \\
    D_{T_s+2h}^{\mu} &= D_{T_s +2h+1}^{\mu} =\cdots =D_{T_s +3h-1}^{\mu} =\iota_s\sum_{j=0}^2 d_j, \\
    & \vdots \\
    D_{T_s+(m-1)h}^{\mu} &= D_{T_s +(m-1)h+1}^{\mu} =\cdots =D_{T_s +mh-1}^{\mu} =\iota_s \sum_{j=0}^{m-1} d_j
\end{align*}
are zero. So, we get 
\begin{align*}
    \hat{v}^{\mu\mu} &= \frac{1}{n} \sum_{|k| \leq \ell} K \left ( \frac{k}{\ell} \right)
                \sum_{i=mh+|k|+1}^n D_i^{\mu} D_{i-|k|}^{\mu} \\
        &= \frac{1}{n} \sum_{|k| \leq \ell} K \left ( \frac{k}{\ell} \right)  (h-|k|) \left\{ d_0^2 + \left(\sum_{j=0}^1 d_j \right)^2 +\cdots+ \left(\sum_{j=0}^{m-1} d_j \right)^2\right\} \sum_{s=1}^{\mathcal{J}} \iota_s^2 \\
        &\quad + \frac{1}{n} \sum_{|k| \leq \ell} K \left ( \frac{k}{\ell} \right) |k| \left\{  \left(\sum_{j=0}^1 d_j \right) d_0  + \left(\sum_{j=0}^2 d_j \right)\left(\sum_{j=0}^1 d_j \right) +\cdots+ \left(\sum_{j=0}^{m-1} d_j \right) \left(\sum_{j=0}^{m-2} d_j \right)\right\} \sum_{s=1}^{\mathcal{J}} \iota_s^2 \\
        &= \frac{1}{n} h \left\{ d_0^2 + \left(\sum_{j=0}^1 d_j \right)^2 +\cdots+ \left(\sum_{j=0}^{m-1} d_j \right)^2\right\} \sum_{s=1}^{\mathcal{J}} \iota_s^2  \sum_{|k| \leq \ell} K \left ( \frac{k}{\ell} \right) \\
        &\quad + \frac{1}{n} \left\{ d_1 d_0 + d_2\left(\sum_{j=0}^1 d_j \right) +\cdots+ d_{m-1}\left(\sum_{j=0}^{m-2} d_j \right) -\left(\sum_{j=0}^{m-1} d_j \right)^2 \right\}  \sum_{s=1}^{\mathcal{J}} \iota_s^2  \sum_{|k| \leq \ell} K \left ( \frac{k}{\ell} \right) |k|  \\
        &= \frac{h \ell \sum_{s=1}^{\mathcal{J}} \iota_s^2}{n} \left\{ d_0^2 + \left(\sum_{j=0}^1 d_j \right)^2 +\cdots+ \left(\sum_{j=0}^{m-1} d_j \right)^2\right\}\sum_{|k| \leq \ell} K \left ( \frac{k}{\ell} \right) \frac{1}{\ell}
        - \frac{\sum_{s=1}^{\mathcal{J}} \iota_s^2}{2n} \sum_{|k| \leq \ell} K \left ( \frac{k}{\ell} \right) |k|  \\
        &= \lambda \xi^2 n^{2\vartheta-1} \left (\sum_{s=1}^{\mathcal{J}} \iota_s^2 \right) \left\{ \sum_{p=0}^{m-1} \left(\sum_{j=0}^{p} d_j \right)^2\right\}  \sum_{|k| \leq \ell} K \left ( \frac{k}{\ell} \right) \frac{1}{\ell}
        - \frac{\sum_{s=1}^{\mathcal{J}} \iota_s^2}{2n} \sum_{|k| \leq \ell} K \left ( \frac{k}{\ell} \right) |k|. 
\end{align*}
Therefore, as $n \to \infty$,
\[
\frac{\hat{v}^{\mu\mu} }{n^{2\vartheta-1}} \to 2 \lambda \xi^2 \left (\sum_{s=1}^{\mathcal{J}} \iota_s^2 \right) \left\{ \sum_{p=0}^{m-1} \left(\sum_{j=0}^{p} d_j \right)^2\right\} \int_{0}^{1} K(t) \dd t 
= 2 \xi^2 \kappa_0 \lambda \left (\sum_{s=1}^{\mathcal{J}} \iota_s^2 \right) \left\{ \sum_{p=0}^{m-1} \left(\sum_{j=0}^{p} d_j \right)^2\right\},
\]
which completes the proof.

\subsection{Proof of Proposition \ref{prop:poly_decay}} \label{pf:poly_decay}
Note that if $c<-1-q$, then the polynomial decay $|\gamma_k| \lesssim a(k+b)^c$ implies that
\[
    u_q = \sum_{k \in \mathbb{Z}}|k|^q |\gamma_k| 
    \lesssim 
    \sum_{k \in \mathbb{Z}}|k|^q C \cdot |k|^c
    = 2C \sum_{k=1}^{\infty} k^{q+c} < \infty,
\]
where $C\in\mathbb{R}$ is a constant.
If we have $u_q < \infty$, then all results stated in Sections \ref{sec:consistency} and \ref{sec:Op_tightD} hold.

\subsection{Proof of Proposition \ref{prop:geom_decay}} \label{pf:geom_decay}
For simple notation, we denote $\hat{v} \equiv \hat{v}_{\tight}(K_{\DF})$ in this proof.
Also, let $K_{\DF, d}$ be the effective kernel of $K_{\DF}$.
Define $\check{v} = \sum_{|k|\leq (m+1)\ell} K_{\DF, d} (k/\ell) \gamma_k^X$.
If the means are constant, we have
\begin{align*}
    {\Bias}_0(\hat{v}) &= \E(\hat{v}) - v 
    = \E(\check{v}) - v + \E\left(\hat{v} - \check{v}\right) \\
    &= \sum_{|k|\leq (m+1)\ell} K_{\DF, d}\left(\frac{k}{\ell} \right) \left \{\gamma_k + O\left(\frac{1}{n}\right) \right\} - \sum_{k \in \mathbb{Z}} \gamma_k + O\left(\frac{\ell}{n}\right),
\end{align*}
which holds since $\E(\hat{\gamma}_k) = \gamma_k +O(1/n)$ for $|k|\leq (m+1)\ell$
and $ \E(\hat{v} - \check{v}) \leq \left\|\hat{v} - \check{v} \right\| = O\{(m+1)\ell/n\}$ by Proposition 2.2 of \citet{Chan2022}.
Then, we have
\begin{align*}
    {\Bias}_0(\hat{v})
    &= \underset{\equiv A_1}{\underbrace{\sum_{|k|\leq (m+1)\ell} \left\{ K_{\DF, d}\left(\frac{k}{\ell} \right) -1\right\} \gamma_k }}
        + \underset{\equiv A_2}{\underbrace{\sum_{|k|\leq (m+1)\ell} K_{\DF, d}\left(\frac{k}{\ell} \right) O\left(\frac{1}{n}\right)}}\\
    &\qquad
        -\underset{\equiv A_3}{\underbrace{\sum_{|k| > (m+1)\ell}\gamma_k}} + O\left(\frac{\ell}{n}\right).
\end{align*}
Note that 
\begin{align*}
    A_2 &= \sum_{|k|\leq (m+1)\ell} K_{\DF, d}\left(\frac{k}{\ell} \right) O\left(\frac{1}{n}\right)
    \leq  \sum_{|k|\leq (m+1)\ell} 1\cdot  O\left(\frac{1}{n}\right) 
    = O\left(\frac{\ell}{n}\right), \\
    |A_3| &\leq 2 \sum_{k > (m+1)\ell}|\gamma_k| 
    \lesssim 2 \sum_{k > (m+1)\ell} \phi^{k}
    =\frac{2\phi^{(m+1)\ell+1}}{1-\phi} 
    =O\left\{ \phi^{(m+1)\ell}\right\}.
\end{align*}
To study $A_1$, we first remark that 
\begin{align*}
    K_{\DF, d}(t) &= \delta_1 K(|t|-1) + K(t) = \delta_1 \times 0 + 1 = 1, \qquad \text{if } |t| \leq c_0; \\
    K_{\DF, d}(t) &= \delta_1 K(|t|-1) + K(t) = \delta_1 \times 1 + 0 = \delta_1, \qquad \text{if } 1-c_0\leq |t| \leq 1+c_0.
\end{align*}
Then, we get
\begin{align*}
    |A_1| &\leq \sum_{|k|\leq (m+1)\ell} \left | K_{\DF, d}\left(\frac{k}{\ell} \right) -1\right| |\gamma_k| \\
    &\lesssim \sum_{|k|\leq (m+1)\ell} \left| K_{\DF, d}\left(\frac{k}{\ell} \right) -1\right|  \phi^{|k|} \\
    &= \sum_{|k|\leq c_0\ell} \left| K_{\DF, d}\left(\frac{k}{\ell} \right) -1\right|  \phi^{|k|}
    + \sum_{ c_0\ell < |k|\leq (m+1)\ell} \left| K_{\DF, d}\left(\frac{k}{\ell} \right) -1\right|  \phi^{|k|}\\
    &\leq 0+ 4\sum_{ c_0\ell < k\leq (m+1)\ell} \phi^{k} \\
    &= O(\phi^{c_0\ell}).
\end{align*}
Therefore, we obtain
\[
    {\Bias}_0(\hat{v}) =  O(\phi^{c_0\ell}) + O\left(\frac{\ell}{n}\right).
\]
Then, the derivation of the variance of $\hat{v}$ is similar to Section \ref{pf:thm_MSE(K)} and we can obtain 
\[
{\Var}_0(\hat{v} ) = \frac{4 A_d \ell v^2}{n} + o\left(\frac{\ell}{n}\right),
\]
where $A_d$ is given in (\ref{eqt:ABd}).
Let $\ell \sim \xi \log n$. 
If the constant $\xi$ satisfies $\xi<-1/(c_0\log\phi)$, 
\[
	\lim_{n\to \infty} \frac{\ell/n}{\phi^{c_0\ell}} 
		= \lim_{\ell\to \infty} \frac{\ell}{e^{\ell(1/\xi+c_0\log\phi)}}
		= \lim_{\ell\to \infty} \frac{1}{(1/\xi+c_0\log\phi)e^{\ell(1/\xi+c_0\log\phi)}}
		= 0 .
\]
Thus, 
the mean squared error of $\hat{v}$ is given by
\[
{\MSE}_0 (\hat{v}) \sim \frac{4 A_d \ell v^2}{n} 
= O\left(\frac{\log n}{n} \right).
\]
Therefore, we get $\hat{v} = v + O_p\{\sqrt{\log n/n} \}$.
Then, we can find the optimal tight difference sequence for $\hat{v}$ as follows:
\[
\left( d_0^*, \ldots, d_m^* \right)(K_{\DF}) = \argmin_{d_{0:m}\in\mathcal{D}_m}\ A_d.
\]
Some commonly used optimal tight difference sequences are shown in Table \ref{table:opT_dj_FT}.

\begin{table}[t]
\def~{\hphantom{0}}
\setlength{\tabcolsep}{4pt}
\centering
\captionsetup{font=small}
\caption{
The optimal values of $(d_0, \ldots, d_m)$ under 
tight differencing and loose differencing for estimation of $v$ and $f(\theta)$
upon assuming geometric decay of autocovariance in (\ref{equ:acvf:geom}). 
The following centrosymmetric dual-flat kernels are used: 
(i) $K_{\DF,\text{d}}(t) = \mathbb{1}(|t|<1/2)+ 1/2\cdot \mathbb{1}(|t|=1/2)$;
(ii) $K_{\DF,\text{c}}(t) = \mathbb{1}(|t|<1/3)+ (2-3|t|) \cdot \mathbb{1}(1/3\leq|t|\leq 2/3)$;
(iii) $K_{\DF,\text{s}}(t) = \mathbb{1}(|t|<1/3)+ (54|t|^3-81|t|^2+36|t|-4) \cdot \mathbb{1}(1/3\leq|t|\leq 2/3)$, 
which correspond to Examples \ref{ex:Kd}--\ref{ex:Ks}, respectively. 
}
\small
\begin{tabular}{llll} 
Kernel & $m$& Tight difference sequence ($\lambda=1$) & Loose difference sequence ($\lambda=2$) \\
$K_{\DF,\text{d}}$ & $2$ & $0.3090,\,0.5,\,-0.8090$ & $0.8090,\,-0.5,\,-0.3090$\\
 & $3$ & $0.1942,\,0.2809,\,0.3831,-0.8582$ &  $0.1942,\,0.2809,\,0.3832,\,-0.8582$\\ 
 & $4$ & $0.1409,\,0.1901,\,0.2464,\,0.3099,\,-0.8873$  & $0.2708,\,-0.0142,\,0.6909,\,-0.4858,\,-0.4617$\\
$K_{\DF,\text{c}}$ & $2$ & $0.2700,\,0.5323,\,-0.8523$ & $0.8090,\,-0.5,\,-0.3090$\\
 & $3$ & $0.1808,\,0.2572,\,0.4155,\,-0.8535$ & $0.1942,\,0.2809,\,0.3832,\,-0.8582$\\
  & $4$ & $0.1319,\,0.1777,\,0.2269,\,0.3464,\,-0.8829$ & $0.2708,\,-0.0142,\,0.6909,\,-0.4858,\,-0.4617$\\
$K_{\DF,\text{s}}$ & $2$ & $0.2803,\,0.5240,\,-0.8043$ & $0.8090,\,-0.5,\,-0.3090$\\
 & $3$ & $0.1843,\,0.2640,\,0.4067,\,-0.8550$ & $0.1942,\,0.2809,\,0.3832,\,-0.8582$\\
  & $4$ & $0.1343,\,0.1809,\,0.2325,\,0.3366,\,-0.8843$ & $0.2708,\,-0.0142,\,0.6909,\,-0.4858,\,-0.4617$\\
\end{tabular} 
\label{table:opT_dj_FT}
\end{table}

\subsection{Proof of Proposition \ref{prop:b-dep kernel}} \label{pf:prop_b-dep kernel}
For simple notation, we denote $\hat{v} \equiv \hat{v}_{\tight}(K_{\DF})$ in this proof.
Given the dual-flat kernel $K_{\DF} (\cdot)$ defined in (\ref{eqt:flattop_Ksym}) and constant $c_0 = \min(c_1, 1-c_2)$, 
recall that the effective kernel of $K_{\DF} (\cdot)$ has the feature that
    $K_{\DF, d}(t) = 1$ if $|t| \leq c_0$;
    and 
    $K_{\DF, d}(t) = \delta_1$ if $1-c_0 \leq |t| \leq 1+ c_0$.
Then, the tight-difference-based estimator satisfies 
\[
\hat{v} = \sum_{|k| \leq \ell} K_{\DF}\left (\frac{k}{\ell}\right ) \hat{\gamma}_k
\quad \text{with} \quad
\hat{\gamma}_k = \frac{1}{n} \sum_{i=m\ell+1}^n D_i D_{i-|k|},\ 
D_i = \sum_{j=0}^m d_j X_{i-j\ell}.
\]
Then similar to Section \ref{pf:geom_decay}, 
it suffices to consider the estimator with the effective kernel $K_{\DF,d}$, i.e., 
$\check{v} =\sum_{|k|\leq (1+m)\ell} K_{\DF, d}(k/\ell)\hat{\gamma}_k^X$. 
This estimator is asymptotically unbiased for $v$ if 
\[
K_{\DF, d} \left (\frac{k}{\ell}\right) = 1 
\quad \text{for all } |k| \leq b,
\]
which is equivalent to 
$|k/\ell| \leq c_0 = \min(c_1, 1-c_2)$ for all $|k| \leq b$. 
It implies that 
\[
	\ell \geq \frac{b}{\min(c_1, 1-c_2)} \geq 2b, 
\]
where the equal sign in the last inequality holds if $\min(c_1, 1-c_2) = c_1 =c_2=1/2$.
Therefore,
the bandwidth and the differencing lag should be set to be at least $b/\min(c_1, 1-c_2)$ to ensure consistency. 
Clearly, the minimum possible bandwidth and the minimum possible differencing lag are achieved if 
$c_1=c_2=1/2$, which leads to the kernel $K_{\DF}^\circ(t) = \mathbb{1}(|t|\leq 1/2)$. 

\subsection{Proof of Proposition \ref{prop:b-dep}} \label{pf:prop_b-dep}
For simple notation, we denote $\hat{v} \equiv \hat{v}_{\tight}(K_{\DF}^\circ)$ 
with the bandwidth $\ell=2b$ in this proof.
Also denote 
the estimator with the effective kernel $K_{\DF,d}^{\circ}$ as
$\check{v} = \sum_{|k|\leq (1+m)\ell} K_{\DF, d}^\circ (k/\ell)\hat{\gamma}_k^X$ in this proof. 
If the means are constant, we have
\begin{align*}
{\Bias}_0(\hat{v}) &= \E(\hat{v}) - v \\
    &= \E(\check{v}) -v +\E(\hat{v} - \check{v}) \\
    &= \sum_{|k| \leq (1+m)\ell} K_{\DF, d}^\circ \left(\frac{k}{\ell} \right) \left \{ \gamma_k + O\left( \frac{1}{n}\right) \right \} - \sum_{|k|\leq b} \gamma_k + O\left( \frac{b}{n} \right ) \\
    &= \sum_{|k| \leq b} K_{\DF, d}^\circ \left(\frac{k}{\ell} \right) \gamma_k - \sum_{|k|\leq b} \gamma_k
        + \sum_{\ell/2 < |k| \leq (1+m)\ell} K_{\DF, d}^\circ \left(\frac{k}{\ell} \right) \gamma_k + O\left( \frac{b}{n} \right ) \\
    &= O\left( \frac{b}{n} \right ) = O\left( \frac{1}{n}\right ),
\end{align*}
since $b/n \to 0$, $K_{\DF, d}^\circ (t) = 1$ if $t\leq 1/2$,
and $\gamma_k = 0$ if $|k| > b $.
Note that the third equation holds since $\E(\hat{\gamma}_k) = \gamma_k +O(1/n)$ for $|k|\leq (1+m)\ell$ 
and $|\E(\hat{v} - \check{v})| \leq \left\|\hat{v} - \check{v} \right\| = O\{(1+m)\ell/n\}$ by Proposition 2.2 of \citet{Chan2022}.
Then, the variance can be found as follows.
\begin{align*}
    {\Var}_0 (\hat{v}) &\sim {\Var}_0 (\check{v}) = {\Var}_0 \left(\sum_{|k| \leq (1+m)\ell} K_{\DF, d}^\circ  \left( \frac{k}{\ell} \right) \hat{\gamma}_k\right) \\
    &= 4 {\Var}_0 \left ( {W}^\top \begin{bmatrix}
        \hat{\gamma}_0 \\
        \hat{\gamma}_1 \\
        \vdots \\
        \hat{\gamma}_{(2m +1)b}
    \end{bmatrix} \right)  \\
    &\sim 4 {W}^\top \Xi {W},
\end{align*}
where the second last equation holds since $K_{\DF. d}^\circ (t) = 0$ if $|t|> m+1/2$ and 
$\Xi$ is the variance-covariance matrix of $\{\hat{\gamma}_0, \hat{\gamma}_1,\ldots,\hat{\gamma}_{(2m+1)b}\}$ with diagonal and off-diagonal elements defined as 
        \[
        \begin{array}{ll}
        \displaystyle{\Xi_{ii} \sim \frac{1}{n} \sum_{k=i-b}^{b-i} (\gamma_k^2 + \gamma_{k-i} \gamma_{k+i})}, & \text{for } 0\leq i \leq (2m+1)b; \\
        \displaystyle{\Xi_{ij} \sim \frac{1}{n} \sum_{k=i-b}^{b-j} \left\{ \gamma_k \gamma_{k+(j-i)} + \gamma_{k-i} \gamma_{k+j} \right\} }, & \text{for } 0 \leq i < j \leq (2m+1)b,
        \end{array}
        \]
which is a direct conclusion by the approximated expressions of covariances of estimated covariances in \citet{bartlett1964}.
Then, we get
\begin{align*}
    {\Var}_0 (\check{v}) &\sim 4 \left\{ \sum_{i=0}^{(2m+1)b} W_i^2 \Xi_{ii} + 2 \sum_{i=0}^{(2m+1)b-1}\sum_{j=i+1}^{(2m+1)b} W_i W_j \Xi_{ij} \right\} \\
    &= \Xi_{00} + 4 \sum_{i=1}^b \Xi_{ii} + 4 \delta_1^2 \sum_{i=b+1}^{(2m+1)b} \Xi_{ii} + 4 \sum_{j=1}^b \Xi_{0j} + 4 \delta_1 \sum_{j=b+1}^{(2m+1)b} \Xi_{0j}  \\
    &\quad + 8\sum_{i=1}^{b-1} \sum_{j=i+1}^{b} \Xi_{ij} + 8\delta_1 \sum_{i=1}^b \sum_{i=b+1}^{(2m+1)b} \Xi_{ij} + 8\delta^2 \sum_{i=b+1}^{(2m+1)b-1} \sum_{j=i+1}^{(2m+1)b} \Xi_{ij} \\
    &= \Xi_{00} + 4 \sum_{i=1}^b \Xi_{ii} + 4 \sum_{j=1}^b \Xi_{0j} + 8\sum_{i=1}^{b-1} \sum_{j=i+1}^{b} \Xi_{ij} \\
    &\quad + 4 \{ (\mathcal{H}_1 + 2 \mathcal{H}_4) \delta_1^2  +
    (\mathcal{H}_2 + 2 \mathcal{H}_3) \delta_1 \}.
\end{align*}
Therefore, to minimize the mean squared error of $\hat{v}$ is equivalent to minimizing the variance of it since the bias term is negligible.
That is,
we need to minimize the term $\mathcal{H} = (\mathcal{H}_1 + 2 \mathcal{H}_4) \delta_1^2  +
    (\mathcal{H}_2 + 2 \mathcal{H}_3) \delta_1$. 
Therefore, the asymptotic \textsc{mse}-optimal tight difference sequence for $\hat{v}$ is
$\argmin_{d_{0:m}\in\mathcal{D}_m}\left\{ (\mathcal{H}_1 + 2 \mathcal{H}_4) \delta_1^2  + (\mathcal{H}_2 + 2 \mathcal{H}_3) \delta_1 \right\}$.

\begin{remark} \label{remark:var-cov_Xi}
    The variance-covariance matrix $\Xi$ can be simplified by the assumption that $\gamma_k = 0$ for all $|k| >b$. That is,
    the diagonal elements
    \begin{align*}
        \Xi_{ii} &\sim \left \{ \begin{array}{ll}
            \displaystyle{\frac{1}{n}\sum_{k=i-b}^{b-i} \gamma_k^2 + \gamma_{k-i}\gamma_{k+i}}, & \text{if } i=0,..,b; \\
            \displaystyle{\frac{1}{n}\sum_{k=i-b}^{b-i} \gamma_k^2}, & \text{if } i=b+1,..,2b; \\
            \displaystyle{\frac{1}{n}\sum_{k=-b}^b \gamma_k^2}, & \text{if } i = 2b+1,\ldots,(2m+1)b.
        \end{array}\right. 
    \end{align*}
    And the off-diagonal elements 
    \begin{align*}
        \Xi_{ij} &\sim \left \{ \begin{array}{ll}
            0, &  \text{if } j-i \geq 2b+1;\\
            \displaystyle{\sum_{k=-b}^{ b-j+i} \frac{\gamma_k \gamma_{k+j-i}}{n}},
                &  \text{if } i = b+1,\ldots,2mb+b-1 
                    \text{ and }j=i+1,\ldots,\min\{i+2b,2mb+b\}.
        \end{array}\right. 
    \end{align*}
\end{remark}

\subsection{Proof of Theorem \ref{thm:MSE_proposal_density}} \label{pf:thm_MSE(f)}

Define the spectral density estimator with the effective kernel $K_d$ as 
\begin{align*}
	\check{f}(\theta)=\frac{1}{2\pi}\sum_{|k| \leq \ell+mh}K_d\left(\frac{k}{\ell}\right) \hat{\gamma}_k^X \cos(2 \pi \theta k), \quad \theta \in [0,0.5],
\end{align*}
where $\hat{\gamma}_k^X=\sum_{i=|k|+1}^{n}(X_i-\bar{X}_n)(X_{i-|k|}-\bar{X}_n)/n$, and $K_d$ is the differencing kernel for an origin-characterized kernel $K \in \mathcal{K}_{q,q^{\prime}}$ with $q\leq q^{\prime}$. And if the signal $\{\mu_i\}_{i \in \mathbb{Z}}$ is a constant, then, as $n \to \infty$, $\| \hat{f}(\theta)-\check{f}(\theta) \| = O\{(\ell+mh)/n \}$, which can be easily obtained by Proposition 2.2 in \citet{Chan2022}.
And thus, we have 
$|\E \{ \hat{f}(\theta)-\check{f}(\theta) \}| \leq \| \hat{f}(\theta)-\check{f}(\theta) \| = O\{(\ell+mh)/n \} = O(\ell/n)$.
If the means are constant, we have
\begin{align*}
&\quad\ {\Bias}_0 \left \{\hat{f}(\theta)\right \} \\
    &= \E \left \{\hat{f}(\theta)\right \}  - f(\theta) \\
    &= \E \left \{\check{f}(\theta)\right \} - f(\theta)
        + \E \left \{\hat{f}(\theta) - \check{f}(\theta)\right \} \\
    &= \frac{1}{2\pi}\sum_{|k| \leq \ell+mh}K_d\left(\frac{k}{\ell}\right) 
        \left \{\gamma_k +O\left(\frac{1}{n} \right)\right \} \cos(2 \pi \theta k) 
      - \frac{1}{2\pi}\sum_{|k|\in \mathbb{Z}} \gamma_k \cos(2 \pi \theta k)  
      + O\left( \frac{\ell}{n} \right) \\
    &= \frac{1}{2\pi}\sum_{|k| \leq \ell} \left\{K_d\left(\frac{k}{\ell}\right) -1 \right\} \gamma_k \cos(2 \pi \theta k) 
        + \frac{1}{2\pi}\sum_{\ell < |k| \leq \ell+mh}K_d\left(\frac{k}{\ell}\right) \gamma_k \cos(2 \pi \theta k) \\
    &\qquad - \frac{1}{2\pi} \sum_{|k|\geq \ell} \gamma_k \cos(2 \pi \theta k)
        + O\left( \frac{\ell}{n} \right). 
\end{align*}
Note that under the assumption that $u_q = \sum_{k \in \mathbb{Z}}|k|^q |\gamma_k| <\infty$, we have 
\begin{align*}
    \left\vert
    \frac{1}{2\pi}\sum_{\ell < |k| \leq \ell+mh}K_d\left(\frac{k}{\ell}\right) \gamma_k \cos(2 \pi \theta k) - \frac{1}{2\pi} \sum_{|k|\geq \ell} \gamma_k \cos(2 \pi \theta k) \right\vert
    \leq O(1) \sum_{|k|>\ell} \frac{|k|^q}{\ell^q} |\gamma_k| 
        = o\left(\frac{1}{\ell^q} \right).
\end{align*}
The remaining proof of the bias is a modified version of the proof of Theorem 5.1 of \citet{Chan2022}.
For the completeness of the proof, we state the details here.
Based on our definition of the origin boundary flatness of 
$K_d(\cdot)$, $B_d = \lim_{t \downarrow 0} \{K_d(t)-K_d(0)\}/|t|^q=\lim_{t \downarrow 0} (\delta_1-1) \{1-K(t)\}/|t|^q$,
there are two important features of $K_d(\cdot)$.
They are (i) $K_d(k/\ell) - 1 = B_d|k/\ell|^q + |k/\ell|^q e_k$ 
where $e_k \to 0$ as $\ell \to \infty$ for each $k=1,\ldots,\sqrt{\ell}$ 
and $\sup_{1\leq |k|\leq \sqrt{\ell}} |e_k|=o(1)$;
and (ii) there exists a sufficiently large constant $\bar{B}_d >0$ 
such that $|K_d(t)-1|\leq \bar{B}_d|t|^q$ for each $t\in [-1,1]$.
Then, we have
\begin{align*}
\frac{1}{2\pi}&\sum_{|k| \leq \ell} \left\{K_d\left(\frac{k}{\ell}\right) -1 \right\} \gamma_k \cos(2 \pi \theta k)  \\
    &= \frac{1}{2\pi}\sum_{|k| \leq \sqrt{\ell}} \left\{K_d\left(\frac{k}{\ell}\right) -1 \right\} \gamma_k \cos(2 \pi \theta k) 
        + \frac{1}{2\pi}\sum_{\sqrt{\ell} <|k| \leq \ell} \left\{K_d\left(\frac{k}{\ell}\right) -1 \right\} \gamma_k \cos(2 \pi \theta k) \\
    &= \frac{1}{\pi} \sum_{1\leq k \leq \sqrt{\ell}} \left\{B_d \left(\frac{k}{\ell}\right)^q + \left(\frac{k}{\ell}\right)^q e_k \right\} \gamma_k \cos(2 \pi \theta k) 
        + \frac{1}{\pi}\sum_{\sqrt{\ell} <k \leq \ell} \left\{K_d\left(\frac{k}{\ell}\right) -1 \right\} \gamma_k \cos(2 \pi \theta k) \\
    &= \frac{B_d f_q(\theta)}{\ell^q} + o\left(\frac{1}{\ell^q} \right)
        + \frac{1}{\pi \ell^q} \sum_{1\leq k \leq \sqrt{\ell}} k^q e_k \gamma_k \cos(2 \pi \theta k) 
        + \frac{1}{\pi}\sum_{\sqrt{\ell} <k \leq \ell} \left\{K_d\left(\frac{k}{\ell}\right) -1 \right\} \gamma_k \cos(2 \pi \theta k).
\end{align*}
Note that 
\begin{align*}
    \left | \frac{1}{\pi \ell^q} \sum_{1\leq k \leq \sqrt{\ell}} k^q e_k \gamma_k \cos(2 \pi \theta k)  \right|
    &\leq \frac{1}{\pi \ell^q} \left(\sup_{1\leq k \leq \sqrt{\ell}} |e_k|\right)
        \sum_{1\leq k \leq \sqrt{\ell}} k^q |\gamma_k \cos(2 \pi \theta k)| \\
    &= o\left(\frac{1}{\ell^q} \right) ,
\end{align*}
and 
\begin{align*}
    \left | \frac{1}{\pi}\sum_{\sqrt{\ell} <k \leq \ell} \left\{K_d\left(\frac{k}{\ell}\right) -1 \right\} \gamma_k \cos(2 \pi \theta k) \right|
    &\leq \frac{1}{\pi}\sum_{\sqrt{\ell} <k \leq \ell} \left |K_d\left(\frac{k}{\ell}\right) -1 \right | \left |\gamma_k \cos(2 \pi \theta k)\right | \\
    &\leq \frac{1}{\pi}\sum_{\sqrt{\ell} <k \leq \ell}\bar{B}_d \left |\frac{k}{\ell} \right |^q \left |\gamma_k \cos(2 \pi \theta k)\right | \\
    &= o\left(\frac{1}{\ell^q} \right).
\end{align*}
Finally, we have the following results for the bias: 
\begin{align*}
	{\Bias}_0\left\{\hat{f}(\theta)\right\} 
        &= \frac{1}{2\pi}\sum_{|k| \leq \ell} \left\{K_d\left ( \frac{k}{\ell} \right ) -1 \right\}\gamma_k \cos(2 \pi \theta k) + O \left( \frac{\ell}{n} \right) + o \left(\frac{1}{\ell^q} \right) \\
		&=\frac{B_d f_{q}(\theta)}{\ell^q} + O \left( \frac{\ell}{n} \right) + o \left(\frac{1}{\ell^q} \right).
\end{align*}
Then, we derive the results for the variance of $\hat{f}(\theta)$. Without loss of generality, suppose $\mu_i=0$ for $i \in \mathbb{Z}$. Denote $\bar{X}_n=\sum_{i=1}^{n}X_i/n$. We have
\begin{align*}
	\hat{\gamma}_k^X = \frac{1}{n} \sum_{i=|k|+1}^{n}X_i X_{i-|k|} -       
        \frac{n+|k|}{n}\bar{X}_n^2 + 
		\frac{1}{n}\bar{X}_n \left ( \sum_{i=1}^{|k|}X_i+\sum_{i=n-|k|+1}^{n}X_i  \right ).
\end{align*}
Then, we re-scale the differencing kernel by letting 
\begin{align} \label{eq:K_tilde}
\tilde{K}_d(t) &= K_d\left\{(1+m)t\right\}\quad \text{for }|t| \leq 1 
\qquad \text{and} \qquad
L=(1+m)\ell.
\end{align}
So, $\check{f}(\theta)$ can be decomposed as
\begin{align} \label{eqt:diff_f}
\check{f}(\theta) &=\sum_{|k| \leq \ell+mh} K_d\left(\frac{k}{\ell}\right)
			\hat{\gamma}_k^X \cos(2 \pi \theta k) \\
	&= \hat{\gamma}_0^X+\sum_{1\leq |k| \leq L}\tilde{K}_d\left(\frac{k}{L}\right)	
			\hat{\gamma}_k^X \cos(2 \pi \theta k) \\ \nonumber
	&= \hat{\gamma}_0^X+\sum_{1\leq |k| \leq L}\frac{1}{n}\tilde{K}_d \left(\frac{k}{L}\right) \cos(2 \pi \theta k) \sum_{i=|k|+1}^{n}X_i X_{i-|k|} \\ \nonumber
    &\quad - \sum_{1\leq |k| \leq L}\frac{n+|k|}{n} \tilde{K}_d \left(\frac{k}{L}\right) \cos(2 \pi \theta k) \bar{X}_n^2 \\ \nonumber
	&\quad +\sum_{1\leq |k|\leq L}\frac{1}{n}\tilde{K}_d\left(\frac{k}{L}\right) \cos(2 \pi \theta k)\bar{X}_n \left ( \sum_{i=1}^{|k|}X_i+\sum_{i=n-|k|+1}^{n}X_i  \right )\\ \nonumber
	&= \hat{\gamma}_0^X+2\sum_{k=1}^{L}\tilde{K}_d\left(\frac{k}{L}\right)\tilde{\gamma}_k^X \cos(2 \pi \theta k) - T_1 \bar{X}_n^2 + T_2 \bar{X}_n, 
\end{align}
where $\tilde{\gamma}_k = \sum_{i=1}^{n-k}X_iX_{i+k}/n$ and 
\begin{align*}
T_1 &=\sum_{1\leq |k| \leq L}\frac{n+|k|}{n} \tilde{K}_d\left(\frac{k}{L}\right) \cos(2 \pi \theta k),\\ 
T_2 &=\sum_{1\leq |k|\leq L}\frac{1}{n}\tilde{K}_d\left(\frac{k}{L}\right) \cos(2 \pi \theta k) \left ( \sum_{i=1}^{|k|}X_i+\sum_{i=n-|k|+1}^{n}X_i  \right ).
\end{align*}
We have the following observations:
\begin{enumerate}
    \item As $L \to \infty$, the variance of the first term $\hat{\gamma}^X_0$ is negligible compared with that of $\check{f}(\theta)$. 

    \item The third term $\left\|T_1\bar{X}_n^2 -\E\left(T_1\bar{X}_n^2\right) \right\| = T_1 \left\|\bar{X}_n^2 - \E\left(\bar{X}_n^2\right) \right\| = O(L/n)$. 
    The first equality is straightforward as $T_1$ is independent of $\bar{X}_n^2$. 
    Then, to prove the second equality, we have first by the central limit theorem, 
    \[
        \sqrt{n} \{\bar{X}_n - \E(\bar{X}_n)\} \inD \Normal (0, v).
    \]
    Let $g(u)=u^2$. By the second-order delta method and the zero-mean assumption, 
    \[
        n (\bar{X}_n^2 - 0^2) \inD v \frac{g''(u)}{2} \chi^2_1, 
        \quad \text{and thus} \quad
        n \bar{X}_n^2 \inD v \chi^2_1.
    \]
    Then, based on Assumption \ref{assump:weakdep}, we have $\E\left(|\bar{X}_n^2|^{1+\epsilon} \right )<\infty$ for some $\epsilon>1$. 
    So, the sequence $\left \{n \bar{X}_n^2 \right \}$ is uniformly integrable.
    Therefore, the distribution convergence of $n \bar{X}_n^2$ implies the $\mathcal{L}^1$ convergence of $n \bar{X}_n^2$, i.e, 
    \[
        \Var\left (n \bar{X}_n^2 \right ) = \E \left (n \bar{X}_n^2 \right) \to \E \left (v \chi^2_1\right ).
    \]
    Thus, 
    $\Var\left (n \bar{X}_n^2\right) = O(1).$
    Finally, we obtain
    \[
        \left\|\bar{X}_n^2 - \E\left(\bar{X}_n^2\right) \right\| = \sqrt{\Var\left(\bar{X}_n^2\right)} = \sqrt{\frac{1}{n^2}\Var\left(n \bar{X}_n^2\right)} = O\left (\frac{1}{n} \right ).
    \]
    Since $T_1 = O(L)$, we conclude that $T_1 \left\|\bar{X}_n^2 - \E\left(\bar{X}_n^2\right) \right\| = O(L/n)$.

    \item The fourth term $\left\|T_2\bar{X}_n-\E\left(T_2\bar{X}_n\right) \right\| \leq \left\|T_2 \right\|_4 \left\|\bar{X}_n \right\|_4 = O(L/n^{3/2})$. The proof is as follows.
    First, note that given $\E(\bar{X}_n)=0$ and by the Cauchy--Schwarz inequality, we have
    \begin{align*}
        \left\|T_2\bar{X}_n-\E\left(T_2\bar{X}_n\right) \right\|^2
        &= \Var\left(T_2 \bar{X}_n\right) = \E\left(T_2 \bar{X}_n\right)^2 = \E |T_2^2 \bar{X}_n^2|\\
        &\leq  \left \{\E\left (T_2^4\right ) \E \left (\bar{X}_n^4 \right) \right \}^{1/2}.
    \end{align*}
    And thus, we have $\left\|T_2\bar{X}_n-\E\left(T_2\bar{X}_n\right) \right\| \leq 
    \left\|T_2 \right\|_4 \left\|\bar{X}_n \right\|_4$.   
    In a similar manner, we can obtain that the sequence $\left \{n^2 \bar{X}_n^4 \right \}$ is uniformly integrable since 
    $\E\left(|\bar{X}_n^4|^{1+\epsilon} \right )<\infty$ for some $\epsilon>0$. 
    Therefore, the distribution convergence of $n \bar{X}_n^2$ implies the $\mathcal{L}^2$ convergence of $n \bar{X}_n^2$, i.e, 
    \[
        \E \left (n^2 \bar{X}_n^4 \right ) \to E \left (v \chi^2_1\right )^2,
        \quad \text{and thus} \quad
        \E \left (n^2 \bar{X}_n^4 \right ) = O(1).
    \]
    So, we get  
    \[
        \left\|\bar{X}_n \right\|_4 = \left \{ \E \left(\bar{X}_n^4 \right) \right \}^{1/4} = O\left ( \frac{1}{n^{1/2}} \right).
    \]
    Then, recall that 
    \[
    T_2 =\sum_{1\leq |k|\leq L}\frac{1}{n}\tilde{K}_d\left(\frac{k}{L}\right) \cos(2 \pi \theta k) \left ( \sum_{i=1}^{|k|}X_i+\sum_{i=n-|k|+1}^{n}X_i  \right ).
    \]
    So, we can easily obtain that there exists some constant $C \in \mathbb{R}$ such that 
    \[
    T_2^4 = C \cdot \frac{1}{n^4} \left \{\sum_{1\leq |k|\leq L} \left ( \sum_{i=1}^{|k|}X_i+\sum_{i=n-|k|+1}^{n}X_i  \right ) \right \}^4.
    \]
    Denote $\tilde{X}_L = \sum_{1\leq |k|\leq L} \left ( \sum_{i=1}^{|k|}X_i+\sum_{i=n-|k|+1}^{n}X_i  \right )$,
    which can be rewritten as 
    \[
        \tilde{X}_L = \sum_{i=1}^L w_i X_i + \sum_{i=n-L+1}^n w_i X_i.
    \]
    And $w_1,\ldots,w_L,w_{n-L+1},\ldots,w_m$ are some constant weights with the order $O(L)$. By Lemma 1 in \citet{Wu2010}, we have 
    \[
        \left\|\sum_{1\leq |k|\leq L} \tilde{X}_L  \right\|_4 
            = O\left ( \sum_{1\leq |k|\leq L} L^2 \right )^{1/2} 
            = O\left ( L^{3/2} \right ).
    \]
    And thus, $\left\| T_2 \right\|_4 = O(L^{3/2} / n)$.
    Finally, we can conclude that 
    \[
    \left\|T_2\bar{X}_n-\E\left(T_2\bar{X}_n\right) \right\| \leq 
    \left\|T_2 \right\|_4 \left\|\bar{X}_n \right\|_4 = O\left(L^{3/2}/n^{3/2} \right).
    \]
\end{enumerate}
Thus, all terms in (\ref{eqt:diff_f}) except the second term 
$2\sum_{k=1}^{L}\tilde{K}_d\left(\frac{k}{L}\right)\tilde{\gamma}_k^X \cos(2 \pi \theta k)$ have negligible variances
asymptotically. 
Let $P_n=2 \sum_{k=1}^{L} \tilde{K}_d(k/L)\tilde{\gamma}_k \cos(2 \pi \theta k)$.
Then, by Minkowski inequality, we get   
\[
\left\| \check{f}(\theta)- f(\theta) \right\| = \left\|P_n-\E(P_n) \right\| + O(L/n).
\]
Suppose $n-L=cL+r_n$ for some $c, r_n \in \mathbb{N}_0$ such that $r_n \leq L$. And partition $\mathbb{U}=\{i,j \in \{1,\ldots, n\} : |i-j| \leq L\}$ into two parts, $\bigcup_{k=1}^{c}\mathbb{V}_k$ and $\mathbb{R}=\mathbb{U}-\bigcup_{k=1}^{c}\mathbb{V}_k$, where $\mathbb{V}_k=\{ i,j \in \{kL+1,\ldots, (k+1)L\}: |i-j| \leq L \}$. So, we have  
\begin{align*}
P_n &= \frac{2L}{n}\sum_{k=1}^{c}\underbrace{\left [\frac{1}{L}\sum_{i=kL+1}^{(k+1)L}\sum_{j=i-L}^{i-1}\tilde{K}_d\left(\frac{i-j}{L}\right)X_iX_j \cos\left\{2 \pi \theta (i-j) \right\} \right]}_{G_k} \\
	&\quad + \frac{2L}{n}\underbrace{\left[\frac{1}{L} \underset{i,j \in \mathcal{R}}{\sum \sum}\tilde{K}_d\left(\frac{i-j}{L}\right)X_iX_j \cos\left\{2 \pi \theta (i-j) \right\} \right]}_{G'}.
\end{align*}
Note that $\left\|G'-\E(G') \right\|=O(1/c)\left\|(G_k-\E(G_k)) \right\|$. Thus,  
\begin{align*}
\Var\left(\frac{2L}{n} \sum_{k=1}^{c}G_k \right) 
	&=\frac{4L^2}{n^2}\left\{\sum_{k=1}^{c}\sum_{k'=1}^{c}\Cov \left(G_k, G_{k'}\right )  \right\} \\
  &= \frac{4L}{n}\Var(G_1)+\frac{8L^2}{n^2}\sum_{1 \leq k <k' \leq c} \Cov \left (G_k, G_{k'} \right ).
\end{align*} 
Then, by 
Lemma 3 in \citet{ChanYau2017} 
and Theorem 2 in \citet{Wu2010}, we have   
\[
	\Var(G_1) \to \varpi(\theta) f^2(\theta) \int_{0}^{1}\tilde{K}_d^2(t) \dd t \quad \text{and}\quad \Cov \left (G_1, G_2 \right ) \to 0,
\]
where $\varpi(\theta)=1$ if $\theta=0$ or $0.5$, and $\varpi(\theta)=1/2$ if $\theta \in (0,0.5)$. Therefore, we get   
\begin{align*} 
{\Var}_0\left\{\hat{f}(\theta)\right\} &= \frac{4L}{n} \varpi(\theta) f^2(\theta) \int_{0}^{1}\tilde{K}_d^2(t) \dd t + o \left( \frac{\ell}{n}\right)  \\ 
	&= \frac{4L}{n} \varpi(\theta) f^2(\theta) \int_{0}^{1}K_d^2 \left \{ (1+m)t \right \} \dd t + o \left( \frac{\ell}{n}\right) \\
   &= \frac{4 A_d\ell \varpi(\theta) f^2(\theta)}{n} + o \left( \frac{\ell}{n}\right),
\end{align*} 
where 
\[
A_d = 2 \kappa_1 \sum_{s=1}^{m}\delta_{s-1}\delta_s + 
                \kappa_2 \sum_{|s|\leq m}\delta_s^2,
\]
and $\kappa_1 = \int_{0}^{1}K (1-t)K (t)\dd t$, $\kappa_2=\int_{0}^{1}K^2(t)\dd t$ are kernel-related terms.
Note that it is easy to verify that $\int_{0}^{1}K_d^2 \left \{ (1+m)t \right \} \dd t = A_d/(1+m)$.
If a centrosymmetric $K \in \mathcal{K}_q^{\circ}$ is used, we have 
\begin{align*}
B 
= \lim_{t \downarrow 0}\frac{ K (t)-1}{|t|^q}
\qquad\text{and}\qquad
B^{\prime}
= \lim_{t \downarrow 0}\frac{ -K (1-t)}{|t|^{q}}.
\end{align*}
Thus, 
\begin{align*}
B_d  & = \lim_{t \downarrow 0} \frac{ K (t)-1 +\delta_1K (1-t)}{|t|^q} \\
	&= \lim_{t \downarrow 0} \frac{K (t) + \delta_1 \left[1-K \{1-(1-t)\}\right] -1}{|t|^q} \\
	&= \lim_{t \downarrow 0} \frac{(\delta_1-1) \left\{1-K (t)\right\}}{|t|^q}  \\
	&= (1-\delta_1)B .
\end{align*}
If $K_0(t) = (1-|t|^q) \mathbb{1}(|t| \leq 1)$, we have $B_d =\delta_1-1$.
Also note that in this case, we have  
\[
\kappa_1 = \int_{0}^{1}K (1-t)K (t)\dd t 
    = \int_{0}^{1}\left \{1-K (t)\right \}K (t)\dd t 
	= \kappa_0-\kappa_2.
\]
And thus,
\[
A_d = 2 (\kappa_0 -\kappa_2) \sum_{s=1}^{m}\delta_{s-1}\delta_s + 
                \kappa_2 \sum_{|s|\leq m}\delta_s^2,  
\]
where $\kappa_0=\int_{0}^{1}K (t)\dd t$, and $\kappa_2=\int_{0}^{1}K ^2(t)\dd t$.

\subsection{Proof of Theorem \ref{thm:clt_spec}} \label{pf:clt_spec}
Theorem \ref{thm:clt_spec} is a direct generalization of Theorem \ref{thm:clt_vhat}.
The proof of the asymptotic normality for tight-difference-based kernel estimator of spectral density  
is similar to that of the long-run variance.
Thus, the proof is omitted.

\subsection{Proof of Proposition \ref{prop:opt_l_f}}
The results can be derived by setting the derivative of the objective function \eqref{eq:MISE_spec} to zero and 
checking the second order condition.   

\subsection{Proof of Theorem \ref{thm:CP}} \label{pf:CP}
First, we prove that the number of observations to be removed should not be too large. 
To be more specific, $a$ should fall in the range of $[0, q/(1+2q)]$. 
For the smallest $\alpha_n/2$ and the largest $\alpha_n/2$ values in $D_i$, 
set $D_i=0$. 
Let $\Lambda_c$ be the complement set of $\Lambda$. 
And $\tilde{\gamma}_k$ be the version of $\hat{\gamma}_k$ that is based on the trimmed $\{D_i : i \in \Lambda\}$.
Then, by Minkowski inequalities, for any $|k| \leq \ell$,
\begin{align*}
\left\| \hat{\gamma}_k-\tilde{\gamma}_k \right\|_2 
	= \frac{1}{n}\left\| \sum_{i \in \Lambda^c} D_iD_{i-k} \right\|_2 
	&\leq \frac{1}{n} \sum_{i \in \Lambda^c} \left\| D_i - D_{i-k}  \right\|_2 \\
	&\leq \frac{1}{n} \sum_{i \in \Lambda^c} \left ( \left\| D_i \right\|_2 - 
        \left\| D_{i-k} \right\|_2 \right ) 
	= O\left (n^{a-1}\right ),
\end{align*}
since there are at most $\alpha_n=O(n^a)$ different values between $\left\| D_i \right\|_2$ and $\left\| D_{i-k} \right\|_2$.
Therefore, we have
\[
\left\| \sum_{|k|\leq \ell} K\left( \frac{k}{\ell} \right) \left ( \hat{\gamma}_k-\tilde{\gamma}_k\right ) \right\|_2  
	\ \leq\ \sum_{|k|\leq \ell} 1 \cdot \left\| \hat{\gamma}_k-\tilde{\gamma}_k \right\|_2 
	\ \leq\ O\left (\ell n^{a-1}\right ).
\]
Then, by the order of the optimal bandwidth, we have
\begin{align} \label{ineqt:tilde}
\left\| \sum_{|k|\leq \ell} K\left( \frac{k}{\ell} \right) \left ( \hat{\gamma}_k-\tilde{\gamma}_k\right ) \right\|_2^2  
	\ \leq\ O\left (\ell n^{a-1}\right )^2 \ =\ O\left (n^{\frac{1}{1+2q}}n^{a-1}\right )^2\ =\ O\left (n^{2a-\frac{4q}{1+2q}}\right ).
\end{align}
Let $v$ be the true value of the long-run variance, and by the convergence rate of $\hat{v}$, we have
\begin{align} \label{ineqt:MSE}
\MSE(\hat{v})\ =\ \left\| \sum_{|k|\leq \ell} K\left( \frac{k}{\ell} \right) 
\hat{\gamma}_k  -v  
\right\|_2^2
\ =\ O\left ( n^{\frac{-2q}{1+2q}}\right ).
\end{align}
Therefore, in order to make sure the optimal convergence, $a$ should satisfy,
\begin{align} \label{ineqt:a}
2a-\frac{4q}{1+2q} < -\frac{2q}{1+2q}, \quad \text{and thus, } \
a < \frac{q}{1+2q}.
\end{align}
Then, we prove that the ratio of the mean squared error of the updated estimator $\hat{v}^{\Diamond}$ and the original estimator $\hat{v}$ will converge to $1$ as $n \to \infty$. By the results in (\ref{ineqt:tilde})--(\ref{ineqt:a}), we have
\[
\left\| \sum_{|k|\leq \ell} K\left( \frac{k}{\ell} \right) \left ( \hat{\gamma}_k-\tilde{\gamma}_k\right ) \right\|_2^2
	\ \leq\  O\left ( n^{\frac{-2q}{1+2q}}\right )\quad \text{and}\quad
\left\| \sum_{|k|\leq \ell} K\left( \frac{k}{\ell} \right) \hat{\gamma}_k-v \right\|_2^2
	\ \leq\ O\left ( n^{\frac{-2q}{1+2q}}\right ).
\]
Thus, by Minkowski inequalities again, we have
\begin{eqnarray}
\MSE(\hat{v}^{\Diamond}) &= &\left\| \sum_{|k|\leq \ell} K\left( \frac{k}{\ell} \right)  \tilde{\gamma}_k-v  \right\|_2^2 \nonumber \\
&\leq& \left\| \sum_{|k|\leq \ell} K\left( \frac{k}{\ell} \right) \left (\tilde{\gamma}_k-\hat{\gamma}_k\right ) \right\|_2^2
+ \left\| \sum_{|k|\leq \ell} K\left( \frac{k}{\ell} \right) \hat{\gamma}_k-v \right\|_2^2 \nonumber \\
&\leq&  O\left ( n^{\frac{-2q}{1+2q}}\right ) +  O\left ( n^{\frac{-2q}{1+2q}}\right ) \nonumber \\
&=& O\left ( n^{\frac{-2q}{1+2q}}\right ), \nonumber 
\end{eqnarray}
which completes the proof.

\section{Theoretical proofs of results in the supplement}\label{sec:supp0}

\subsection {Proof of Theorem \ref{thm:bv_lambda_12}} \label{pf:thm_bv_lambda_12}
\begin{proof}
The proof follows the same steps and decomposition as in
\S\ref{pf:thm_MSE(f)} for the tight-differencing case, with the key observation that
when $\lambda\in(1,2)$ the induced effective kernel $K_d$ coincides with $K$ in a neighborhood of the origin and 
admits the form (\ref{eqt:Kd_subtight}).

We first derive the bias for $\hat v_\subtight$.
Based on the above formula for $K_d$ when $\lambda \in (1,2)$, we remark that $K_d(t)=K(t)$ for all $|t|<\lambda-1$.
Under $\mu_1=\cdots=\mu_n$, the same expansion as in \S\ref{pf:thm_MSE(f)} yields
\begin{align*} 
{\Bias}_0(\hat v_{\subtight})
&=
\sum_{|k|\le \ell}\Bigl\{K_d(k/\ell)-1\Bigr\}\gamma_k
+O\!\left(\frac{\ell}{n}\right)
+o\!\left(\frac{1}{\ell^q}\right),
\end{align*}
where the remainder terms are of the same order as those in the tight-differencing case.
Split the main sum at $|k|\le \sqrt{\ell}$:
\begin{align} \label{eqt:split_main_term_bias_subtight}
\sum_{|k|\le \ell}\{K_d(k/\ell)-1\}\gamma_k
&=
\sum_{|k|\le \sqrt{\ell}}\{K(k/\ell)-1\}\gamma_k
+\sum_{\sqrt{\ell}<|k|\le \ell}\{K_d(k/\ell)-1\}\gamma_k,
\end{align}
because $|k|/\ell\le \ell^{-1/2}<\lambda-1$ for all $|k|\le \sqrt{\ell}$ and large $\ell$.
Again, similar to the proof in \S\ref{pf:thm_MSE(f)}, we will use two important features of the kernels:
(i) since $\CE(K) = q$, we can write the near-origin expansion 
\[K(k/\ell) - 1 = B|k/\ell|^q + |k/\ell|^q e_k,\] 
where $e_k \to 0$ as $\ell \to \infty$ for each fixed $|k|\le\sqrt{\ell}$ 
and $\sup_{1\leq |k|\leq \sqrt{\ell}} |e_k|=o(1)$. 
(ii) There exists a sufficiently large constant $\bar{\bar{B}}_d >0$ such that 
\[
|K_d(t)-1 |\le \bar{\bar B}_d |t|^q
\qquad \text{for all } t\in[-1,1].
\]
We first consider the first term in \eqref{eqt:split_main_term_bias_subtight}.
By definition of $v_q=\sum_{k\in\mathbb Z}|k|^q\gamma_k$, the assumption $u_q<\infty$, and the feature (i), we have
\begin{align*}
\sum_{|k|\le \sqrt{\ell}}\{K(k/\ell)-1\}\gamma_k
&= \frac{B}{\ell^q}\sum_{|k|\le \sqrt{\ell}}|k|^q\gamma_k
+\frac{1}{\ell^q}\sum_{|k|\le \sqrt{\ell}}|k|^q e_k\gamma_k \\
&= \frac{B v_q}{\ell^q} + o\left(\frac{1}{\ell^q}\right)
+\frac{1}{\ell^q}\sum_{|k|\le \sqrt{\ell}}|k|^q e_k\gamma_k.
\end{align*}
And note that
\[
\left|
\frac{1}{\ell^q}
\sum_{|k|\le \sqrt{\ell}} |k|^q e_k \gamma_k
\right|
\le
\frac{1}{\ell^q}
\left(\sup_{1\le |k|\le \sqrt{\ell}}|e_k|\right)
\sum_{|k|\le \sqrt{\ell}} |k|^q |\gamma_k|
= o\left(\frac{1}{\ell^q}\right).
\]
Thus, the first term in \eqref{eqt:split_main_term_bias_subtight} is
\[
\sum_{|k|\le \sqrt{\ell}}\{K(k/\ell)-1\}\gamma_k 
= \frac{B v_q}{\ell^q} + o\left(\frac{1}{\ell^q}\right).
\]
Then, for the second term in \eqref{eqt:split_main_term_bias_subtight}, using feature (ii),
\[
\left| \sum_{\sqrt{\ell}<|k|\le \ell} \{K_d(k/\ell)-1\}\gamma_k\right| 
\le \sum_{\sqrt{\ell}<|k|\le \ell} \left|K_d(k/\ell)-1\right| | \gamma_k|
\le \bar{\bar B}_d
\sum_{\sqrt{\ell}<|k|\le \ell} \left|\frac{k}{\ell}\right|^q |\gamma_k|
= o\left(\frac{1}{\ell^q}\right),
\]
Finally, we have the following results for the bias:
\[
{\Bias}_0(\hat v_{\subtight})
= \frac{Bv_q}{\ell^q}
+O\left(\frac{\ell}{n}\right)
+o\left(\frac{1}{\ell^q}\right).
\]

We then derive the variance for $\hat v_\subtight$.
As in \S\ref{pf:thm_MSE(f)}, introduce the rescaled kernel
\[
\tilde K_d(t)=K_d\{(1+m\lambda)t\},\qquad |t|\le 1,
\qquad\text{and}\qquad
L=(1+m\lambda)\ell=\ell+mh.
\]
Then the same variance reduction to an $L^2$-norm of $\tilde K_d$ yields
\begin{align*}
\Var_0(\hat v_{\subtight})
&=\frac{4L}{n}v^2\int_0^1 \tilde K_d^2(t)\,dt
+o\left(\frac{\ell}{n}\right) 
=\frac{4\ell v^2}{n}\tilde A_d
+o\left(\frac{\ell}{n}\right),
\end{align*}
where
\[
\tilde A_d
=(1+m\lambda)\int_0^1 K_d^2\{(1+m\lambda)t\}\,\dd t
=\int_0^{1+m\lambda}K_d^2(t)\,\dd t.
\]
Expanding $K_d(t)=\sum_{s=-m}^m \delta_{|s|}K(t+\lambda s)$ gives
\begin{align*}
\tilde A_d
&=
\sum_{r=-m}^m\sum_{s=-m}^m \delta_{|r|}\delta_{|s|}
\int_0^{1+m\lambda} K(t+\lambda r)K(t+\lambda s)\dd t \\
&=
\sum_{r=-m}^m\sum_{s=-m}^m \delta_{|r|}\delta_{|s|}
\int_{-\lambda r}^{1+m\lambda-\lambda r} K(u)K\{u+\lambda(s-r)\}\dd u,
\qquad (u=t+\lambda r).
\end{align*}
Write $k=s-r$. Since $K(u)\neq 0$ only when $u\in(-1,1)$, the integrand is nonzero only if
\[
(-1,1)\cap(-1-\lambda k,\,1-\lambda k)\neq\varnothing
\quad\Longleftrightarrow\quad |\lambda k|<2.
\]
When $\lambda\in(1,2)$, the only integers satisfying $|\lambda k|<2$ are $k\in\{-1,0,1\}$.
Thus,
\begin{align*}
\tilde A_d
&= \sum_{s=-m}^m \delta_{|s|}^2
\int_{\max(-\lambda s,-1)}^{\min(1+m\lambda-\lambda s,\,1)}
K(u)^2\,\dd u \\
&\quad
+ \sum_{s=-m}^{m-1} \delta_{|s|}\delta_{|s+1|}
\int_{\max(-\lambda s,-1)}^{\min(1+m\lambda-\lambda s,\,1)}
K(u)K(u+\lambda)\,\dd u \\
&\quad
+ \sum_{s=-m+1}^{m} \delta_{|s|}\delta_{|s-1|}
\int_{\max(-\lambda s,-1)}^{\min(1+m\lambda-\lambda s,\,1)}
K(u)K(u-\lambda)\,\dd u \\
&= \sum_{s=-m}^m \delta_{|s|}^2 \kappa_{2,s}
+ \sum_{s=-m}^{m-1} \delta_{|s|}\delta_{|s+1|}\kappa^{(+)}_{1,s}
+ \sum_{s=-m+1}^{m} \delta_{|s|}\delta_{|s-1|}\kappa^{(-)}_{1,s},
\end{align*}
where $L_s=\max(-\lambda s,-1)$ and $U_s=\min\{1+(m-s)\lambda,1\}$, and
\[
\kappa_{2,s}=\int_{L_s}^{U_s}K(u)^2\dd u,\quad
\kappa_{1,s}^{(+)}=\int_{L_s}^{U_s}K(u)K(u+\lambda)\dd u,\quad
\kappa_{1,s}^{(-)}=\int_{L_s}^{U_s}K(u)K(u-\lambda)\dd u.
\]
Then collecting the $k=0$ terms yields $\sum_{|s|\le m}\kappa_{2,s}\delta_s^2$, while collecting $k=\pm1$
and symmetrizing over $s$ gives
\[
\tilde A_d
=
\sum_{|s|\le m}\kappa_{2,s}\delta_s^2
+\sum_{s=1}^m\Bigl\{\kappa_{1,s-1}^{(+)}+\kappa_{1,-s}^{(+)}+\kappa_{1,s}^{(-)}+\kappa_{1,-(s-1)}^{(-)}\Bigr\}\delta_{s-1}\delta_s.
\]
This proves
\[
{\Var}_0(\hat v_{\subtight})
= \frac{4\tilde A_d\,\ell v^2}{n}
+o\left(\frac{\ell}{n}\right).
\]
\end{proof}

\subsection{Proof of Corollary \ref{corol:MSE_vq}} \label{pf:corol_MSE_vq}
Corollary \ref{corol:MSE_vq} on the bias and variance of $\hat{v}_{p,\tight}$ is a special case of Corollary \ref{corol:MSE_fq} 
since $\hat{v}_{p,\tight} = 2 \pi \hat{f}_{p,\tight}(0)$ is just a special case of $\hat{f}_{p}(\theta)$.
Thus, the proof is omitted. 

\subsection{Proof of Corollary \ref{corol:clt_vq_hat}} \label{pf:corol_clt_vq_hat}
We drop the subscript ${}_{\tight}$ in this proof to lighten notation. 
First note that 
\[
\frac{1}{\ell^p} \hat{v}_{p}(K) = \hat{v}\left(K^{(p)}\right), 
\quad \text{where} \quad
K^{(p)} (t) = |t|^p K(t).
\]
Define $K^{(p)}_d(t)$ as the differencing kernel of $K^{(p)}(t)$ and
$\tilde{K}^{(p)}_d(t)=K^{(p)}_d\{(1+m)t\}=\{(1+m)|t|\}^p K_d\{(1+m)t\}$ as the re-scaled differencing kernel. 
And we define $\check{v}_{p}(K)/{\ell^p} =  \check{v}(K^{(p)})$
Then, by (\ref{eqt:1_pf_clt_v}), we have
\begin{align*}
        \frac{1}{\ell^p}\check{v}_{p} 
    &=  \sum_{|k|\leq \ell+mh} K^{(p)}_d \left ( \frac{k}{\ell}\right) \hat{\gamma}_k^X \\
    &= \sum_{|k|\leq L} \tilde{K}^{(p)}_d \left ( \frac{k}{L}\right) \hat{\gamma}_k^X \\
    &=
        2 \sum_{1\leq i<j\leq n}a_n^{(p)}(i,j) (X_i - \bar{X}_n) (X_j - \bar{X}_n),
\end{align*}
where $a_n^{(p)}(i,j) \equiv \tilde{K}_d^{(p)} (|i-j|/L)/n \cdot \mathbb{1}\left(|i-j|\leq L \right)$. 

Then we have $A_{t,n}^{(p)} \equiv \sum_{j=1}^{t-1} \left \{a_n^{(p)}(j,t) \right\}^2$ and 
\begin{align*}
    \left( {\varrho_n^{(p)} } \right)^2 
    &\equiv \sum_{t=2}^n A_{t,n}^{(p)} 
        = \sum_{t=2}^L A_{t,n}^{(p)} +\sum_{t=L+1}^n A_{t,n}^{(p)} \\
        &= \sum_{t=2}^L \sum_{k=1}^{t-1} \frac{1}{n^2} \left\{\tilde{K}_d^{(p)} \left ( \frac{k}{L} \right ) \right\}^2
            + \sum_{t=L+1}^n \sum_{k=1}^L \frac{1}{n^2} \left\{\tilde{K}_d^{(p)} \left ( \frac{k}{L} \right ) \right\}^2 \\
        &= \sum_{t=2}^L \sum_{k=1}^{t-1} \frac{k^p}{n^2 L^p} \tilde{K}_d^2 \left ( \frac{k}{L} \right ) + \frac{1}{n^2} \sum_{t=L+1}^n \left[\sum_{k=1}^L \left\{ \tilde{K}_d^{(p)} \left ( \frac{k}{L} \right )\right\}^2 \right] \\
        &= o\left (\frac{L^2}{n^2}\right) + \frac{1}{n^2}  \cdot n \cdot L \int_0^1 \left\{\tilde{K}_d^{(p)} (t)\right\}^2 \dd t \\
        &= \frac{1}{n} \cdot \frac{L}{1+m}  \int_0^1 \left\{K_d^{(p)} (t)\right\}^2 \dd t 
        =\frac{\ell}{n} A_d^{(p)}.
\end{align*}
Then, recall the martingale-approximated process $M_i$ we defined in Section \ref{pf:thm_clt_vhat},
\begin{align*}
    \tilde{X}_i \equiv \E(X_i \mid \epsilon_{i-\psi}, \ldots, \epsilon_i ), \quad 
    m_i \equiv \sum_{k=0} ^\infty \E\left(\tilde{X}_{i+k} \mid† \mathcal{F}_{i} \right), \quad \text{and }
    M_i\equiv m_i - \E(m_i \mid \mathcal{F}_{i-1}),
\end{align*}
where $\psi$ may depend on $n$.
Then the martingale-approximated estimator is 
\[
    \frac{1}{\ell^p}\tilde{v}_p = 
            2 \sum_{1\leq i<j\leq n}a_n^{(p)}(i,j) M_i M_j.
\]
Then, in a similar manner of the proof in Section \ref{pf:thm_clt_vhat},
we can prove that it is asymptotically equivalent to study the asymptotic properties of the martingale-approximated version since
\begin{align*}
    \left\| \check{v}_{p} -\E(\check{v}_{p}) - \tilde{v}_p +\E(\tilde{v}_p)\right\| &= o(\ell^p \varrho_n^{(p)}) .
\end{align*}
Next, 
similar to the derivation in Section \ref{pf:thm_clt_vhat},  
the conditions (i)--(iv) in Theorem 1 of \citet{WuShao2007} are satisfied after performing straightforward algebra.
Hence, we obtain 
\[
\frac{\sum_{1\leq i<j\leq n}a_n^{(p)}(i,j) M_i M_j} {\varrho_n^{(p)}} \inD
\Normal(0, v^2 ).
\]
Finally, we get
\[
\frac{\hat{v}_{p,\tight} -\E(\hat{v}_{p,\tight})} {\ell^p \varrho_n^{(p)}} \inD 
\Normal(0, 4v^2 ),
\]
thus, the result follows. 

\subsection{Proof of Proposition \ref{prop:l_vq}}
The results can be derived by setting the derivative of the limit of $n^{2q/(1+2p+2q)}\MSE_0(\hat{v}_{p,\tight})$ to zero and 
checking the second order condition.   

\subsection{Proof of Corollary \ref{corol:MSE_fq}} \label{pf:mse_spec_fq}
Throughout this proof, 
we denote the estimator as $\hat{f}_{p}(\theta)$ or $\hat{f}_{p}(\theta, K)$  to lighten notation.  
The proof of the bias and variance of $\hat{f}_{p}(\theta)$ under constant mean is similar to the proof in Section \ref{pf:thm_MSE(f)}.
Denote 
\begin{align*}
	\check{f}_{p}(\theta)=\frac{1}{2\pi}\sum_{|k| \leq \ell+mh}|k|^p K_d\left(\frac{k}{\ell}\right) \hat{\gamma}_k^X \cos(2 \pi \theta k), \quad \theta \in [0,0.5].
\end{align*}
And if the signal $\{\mu_i\}_{i \in \mathbb{Z}}$ is a constant, 
then, as $n \to \infty$, 
$\| \hat{f}_{p}(\theta)-\check{f}_{p}(\theta) \| = O\{\ell^{1+p}/n \}$, 
which can be easily obtained by following Proposition 2.2 in \citet{Chan2022}.
And thus, we have 
$|\E \{ \hat{f}_{p}(\theta)-\check{f}_{p}(\theta) \} |\leq \| \hat{f}_{p}(\theta)-\check{f}_{p}(\theta) \| = O(\ell^{1+p}/n )$.
In a similar manner of the proof of Theorem \ref{thm:MSE_proposal_density},
we have 
\begin{align*}
    {\Bias}_0 \left \{ \hat{f}_{p}(\theta) \right\} 
    &= \frac{1}{2\pi}\sum_{|k| \leq \ell} |k|^p \left\{K_d\left ( \frac{k}{\ell} \right ) -1 \right\}\gamma_k \cos(2 \pi \theta k) + O \left( \frac{\ell^{1+p}}{n} \right) + o \left(\frac{1}{\ell^q} \right) \\
		&=\frac{B_d f_{p+q}(\theta)}{\ell^q} + O \left( \frac{\ell^{1+p}}{n} \right) + o \left(\frac{1}{\ell^q} \right).
\end{align*}
Next, for the derivation of the variance of $\hat{f}_{p}(\theta)$, we first remark that
\[
    \frac{1}{\ell^p}\hat{f}_{p}(\theta, K) = \hat{f}\left(\theta, K^{(p)}\right), 
    \quad \text{where} \quad
    K^{(p)}(t) = |t|^p K(t), 
\]
where $\hat{f}\left(\cdot, \cdot\right)= \hat{f}_0\left(\cdot, \cdot\right)$.
Again, in a similar manner of the proof of Theorem \ref{thm:MSE_proposal_density},
we have
\begin{align*} 
{\Var}_0\left\{\frac{1}{\ell^p}\hat{f}_{p}(\theta)\right\} 
	&= \frac{4\ell}{n} \varpi(\theta) f^2(\theta) \int_{0}^{1}\left \{K^{(p)}_d (t)\right\}^2 \dd t + o \left( \frac{\ell}{n}\right) \\
   &= \frac{4\ell}{n} \varpi(\theta) f^2(\theta) \int_{0}^{1} |t|^{2p} K_d^2 (t) \dd t + o \left( \frac{\ell}{n}\right) \\
   &= \frac{4 A_d^{(p)} \ell \varpi(\theta) f^2(\theta)}{n} + o \left( \frac{\ell}{n}\right),
\end{align*}
and finally, we get
\[
{\Var}_0\left\{\hat{f}_{p}(\theta)\right\} 
= \frac{4 A_d^{(p)} \ell^{1+2p}  f^2(\theta) \varpi(\theta)}{n} 
    + o \left( \frac{\ell^{1+2p}}{n}\right).
\]

\subsection{Proof of Corollary \ref{corol:clt_spec_w}} \label{pf:clt_spec_w}
Corollary \ref{corol:clt_spec_w} is a direct generalization of Theorem  \ref{thm:clt_spec}.
The proof is similar and thus omitted.

\subsection{Proof of Corollary \ref{corol:l_fq}}
The results can be derived by setting the derivative of the limit of $n^{2q/(1+2p+2q)}\MISE_0\{\hat{f}_{p,\tight}(\cdot)\}$ to zero and 
checking the second order condition.   

\subsection{Proof of Proposition \ref{prop:geom_decay_spec}} \label{pf:geom_decay_spec}
Proposition \ref{prop:geom_decay_spec} is a direct generalization of Proposition \ref{prop:geom_decay}.
The proof is similar and thus omitted.

\subsection{Proof of Proposition \ref{prop:b-dep_spec_w}} \label{pf:b-dep_spec_w}
To simplify notation, we denote $\hat{f}_{p}(\theta) \equiv \hat{f}_{p,\tight}(\theta, K_{\DF}^\circ)$ 
with the bandwidth $\ell=2b$ in this proof.
And consider the equivalent form of $\hat{f}_{p}(\theta)$ based on the effective kernel $K_{\DF, d}^\circ(\cdot)$ under tight differencing. Denote
\[
\check{f}_p(\theta) = (2\pi)^{-1}\sum_{|k|\leq (1+m)\ell} |k|^p K_{\DF, d}^\circ \left (\frac{k}{\ell} \right )\hat{\gamma}_k^X \cos(2\pi\theta k).
\]
We assume $p \in \mathbb{N}_0$ and $\theta \in [0,1/2]$.
If the means are constant,
we have
\begin{align*}
	{\Bias}_0\left \{\hat{f}_p(\theta) \right \} 
    &= \E\left \{\check{f}_p(\theta) \right \} -f_p(\theta) +\E\left \{\hat{f}_p(\theta) -\check{f}_p(\theta) \right \}  \\
    &= \sum_{|k| \leq (1+m)\ell} K_{\DF, d}^\circ \left(\frac{k}{\ell} \right) |k|^p 
    	\left \{ \gamma_k + O\left( \frac{1}{n}\right) \right \} \cos(2\pi\theta k) \\
	&\qquad	- \sum_{|k|\leq b} |k|^p \gamma_k \cos(2\pi\theta k)
		+ O\left( \frac{\ell^{1+p}}{n} \right ) \\
    &= \sum_{|k| \leq b} K_{\DF, d}^\circ \left(\frac{k}{\ell} \right) |k|^p \gamma_k \cos(2\pi\theta k) 
    - \sum_{|k|\leq b} |k|^p \gamma_k \cos(2\pi\theta k) \\
     &\qquad + \sum_{\ell/2 < |k| \leq (1+m)\ell} K_{\DF, d}^\circ \left(\frac{k}{\ell} \right) |k|^p \gamma_k  \cos(2\pi\theta k)
     + O\left( \frac{\ell^{1+p}}{n} \right ) \\
    &= O\left( \frac{\ell^{1+p}}{n} \right ) 
        =O\left( \frac{b^{1+p}}{n} \right ).
\end{align*}

Note that the second equation holds since 
$$\left | \E\left\{\hat{f}_p(\theta)- \check{f}_p(\theta) \right\} \right |\leq \left\|\hat{f}_p(\theta) - \check{f}_p(\theta)\right\| = O\{\ell^{1+p}/n\},$$
which can be simply proved based on the results in
Proposition 2.2 of \citet{Chan2022}.
And since $\gamma_k = 0$ for $|k| > b$ under the vanishing autocovariance structure, and $K^\circ_{\mathrm{DF}}(k/\ell) = 1$ for $|k| \le b$, the above derivation holds.

Then, the variance of $\hat{f}_p(\theta)$ can be derived in a similar manner of the derivation of ${\Var}_0\{\hat{v}_{\tight}(K_{\DF}^\circ)\}$ in Section \ref{pf:prop_b-dep}.
Recall and define the following notations:
\begin{align*}
    {W} &= (1/2, {J}_{b}^\T, \delta_1{J}_{2mb}^\T)^\T, \\
    {W}_{f} &= {W} \circ (1, \cos(2\pi\theta), \cos(4\pi\theta), ..., \cos(2B \pi\theta))^\top, \\
    {W}_{f}^{(p)} &= {W}_{f} \circ (0^p, 1^p, 2^p, ..., B^p)^\top.
\end{align*}
where $A \circ B$ is the elementwise product. Let $\Xi$ be the same variance-covariance matrix defined in Proposition \ref{prop:b-dep}.
We get
\begin{align*}
    {\Var}_0 \left \{\hat{f}_p(\theta) \right \}  
    &\sim {\Var}_0 \left \{\check{f}_p(\theta) \right \}  \\
    & = {\Var}_0 \left(\sum_{|k| \leq (1+m)\ell} K_{\DF, d}^\circ  \left( \frac{k}{\ell} \right) |k|^p \hat{\gamma}_k \cos(2\pi\theta k)\right) \\
    &= 4 {\Var}_0 \left ( {{W}^{(p)}_f}^\top \begin{bmatrix}
        \hat{\gamma}_0 \\
        \hat{\gamma}_1 \\
        \vdots \\
        \hat{\gamma}_{(2m +1)b}
    \end{bmatrix} \right) 
    \sim 4 {{W}_f^{(p)}}^\top \Xi {W}_f^{(p)} .
\end{align*}
Then, similarly letting $W^{(p)}_{f,i}$ be the $(i+1)$th element in ${W}_{f}^{(p)}$, we have

\begin{align*}
{\Var}_0 \left \{\check{f}_p(\theta) \right \} 
    &\sim 4 \left[ \sum_{i=0}^{(2m+1)b} \left \{W^{(p)}_{f,i}\right\}^2 \Xi_{ii} + 2 \sum_{i=0}^{(2m+1)b-1}\sum_{j=i+1}^{(2m+1)b} W^{(p)}_{f,i} W^{(p)}_{f,j} \Xi_{ij} \right] \\
    &= 0^{2p} \Xi_{00} + 
    4 \sum_{i=1}^b i^{2p} \Xi_{ii} \cos^2(2\pi\theta i) + 4 \delta_1^2 \sum_{i=b+1}^{(2m+1)b} i^{2p} \Xi_{ii} \cos^2(2\pi\theta i) \\
    &\quad +  4 \sum_{j=1}^b 0^p j^p \Xi_{0j} \cos(2\pi\theta j)
        + 4 \delta_1 \sum_{j=b+1}^{(2m+1)b} 0^p j^p \Xi_{0j} \cos(2\pi\theta j)  \\
    &\quad + 8\sum_{i=1}^{b-1} \sum_{j=i+1}^{b} (ij)^p \Xi_{ij} \cos(2\pi\theta i)\cos(2\pi\theta j)\\
    &\quad + 8\delta_1 \sum_{i=1}^b \sum_{i=b+1}^{(2m+1)b} (ij)^p \Xi_{ij} \cos(2\pi\theta i)\cos(2\pi\theta j) \\
    &\quad + 8\delta^2 \sum_{i=b+1}^{(2m+1)b-1} \sum_{j=i+1}^{(2m+1)b} (ij)^p \Xi_{ij} \cos(2\pi\theta i)\cos(2\pi\theta j) \\
    &= 0^{2p} \Xi_{00} + 
    4 \sum_{i=1}^b i^{2p} \Xi_{ii} \cos^2(2\pi\theta i)
    +  4 \sum_{j=1}^b 0^p j^p \Xi_{0j} \cos(2\pi\theta j) \\
    &\quad + 8\sum_{i=1}^{b-1} \sum_{j=i+1}^{b}  (ij)^p \Xi_{ij} \cos(2\pi\theta i)\cos(2\pi\theta j) \\
    &\quad + 4 \left[ \left\{\mathcal{H}_{1}^{(p)} + 2 \mathcal{H}_{4}^{(p)}\right\} \delta_1^2  
     + \left\{\mathcal{H}_{2}^{(p)} +2 \mathcal{H}_{3}^{(p)}\right \}\delta_1 \right ],
\end{align*}
where 
\begin{align*}
    \mathcal{H}_1^{(p)} &= 
                 \sum_{i=b+1}^{B} i^{2p} \cos^2(2\pi \theta i) \Xi_{ii} \\
                &\sim \frac{1}{n} \sum_{i=b+1}^{2b} i^{2p} \cos^2(2\pi \theta i) \sum_{k=b-i}^{i-b} \gamma_k^2 +  
                    \frac{1}{n} \sum_{i=2b+1}^{B} i^{2p} \cos^2(2\pi \theta i) \sum_{k=-b}^{b}\gamma_k^2,\\
    \mathcal{H}_2^{(p)} &= \begin{cases}
                0, &\text{if } p\neq0; \\
                \displaystyle{\sum_{j=b+1}^{2b} j^p  \cos(2\pi \theta j)  \Xi_{0j}
                \sim \frac{2}{n} \sum_{j=b+1}^{2b} j^p \cos(2\pi \theta j)  \sum_{k=-b}^{b-j} \gamma_k\gamma_{k+j}}, &\text{if } p =0
            \end{cases},
\end{align*}

\begin{align*}
    \mathcal{H}_3^{(p)} &= 
                 \sum_{i=1}^b \sum_{j=b+1}^{i+2b} (ij)^p \cos(2\pi \theta i) \cos(2\pi \theta j) \Xi_{ij} \\
                &\sim\frac{1}{n} \sum_{i=1}^b \sum_{j=b+1}^{i+2b} (ij)^p \cos(2\pi \theta i) \cos(2\pi \theta j) \sum_{k=i-b}^{b-i-j} \left\{\gamma_k \gamma_{k+(j-i)} + \gamma_{k-i} \gamma_{k+j} \right\}, \\
    \mathcal{H}_4^{(p)} &= 
                 \sum_{i=b+1}^{B-1} \sum_{j=i+1}^{\min\{i+2b,\ B\}} (ij)^p \cos(2\pi \theta i) \cos(2\pi \theta j) \Xi_{ij} \\
                &\sim \frac{1}{n}\sum_{i=b+1}^{B-1} \sum_{j=i+1}^{\min\{i+2b,\ B\}} (ij)^p \cos(2\pi \theta i) \cos(2\pi \theta j) \sum_{k=-b}^{\min\{i-b,\ b-(j-i) \}} \gamma_k\gamma_{k+j-i}.
\end{align*}

Therefore, to minimize the asymptotic mean squared error of $\hat{f}_p(\theta)$ is again equivalent to
minimizing its asymptotic variance.  
Thus, the optimal $\delta_1$ is obtained by minimizing 
\[
\left\{\mathcal{H}_{1}^{(p)} + 2 \mathcal{H}_{4}^{(p)}\right\} \delta_1^2  + \left\{\mathcal{H}_{2}^{(p)} +2 \mathcal{H}_{3}^{(p)}\right \}\delta_1.
\]
Therefore, the asymptotic \textsc{mse}-optimal tight difference sequence for $\hat{f}_{p,\tight}(\cdot, K_{\DF}^\circ)$ is  
        \begin{align*}
			d_{f_p(\theta),0:m}^{*\textsc{vani}}
            = \argmin_{d_{0:m}\in\mathcal{D}_m} \left\{ \left(\mathcal{H}_1^{(p)} + 2 \mathcal{H}_4^{(p)}\right) \delta_1^2  
            +  \left(\mathcal{H}_2^{(p)} + 2 \mathcal{H}_3^{(p)} \right)  \delta_1  \right\}.
        \end{align*}
And we remark that all the above results hold for $\theta = 0$, i.e., $\hat{v}_{p,\tight}(K_{\DF}^\circ)$ under vanishing autocovariance structure. 
And the asymptotic \textsc{mse}-optimal tight difference sequence for $\hat{v}_{p,\tight}(K_{\DF}^\circ)$ can be found by solving the optimization problem above by setting $\theta=0$. 

And we can further find the asymptotic \textsc{mise}-optimal tight difference sequence for $\hat{f}_{p,\tight}(\cdot, K_{\DF}^\circ)$ as follows. 
First, the MISE of $\hat{f}_{p,\tight}(\theta, K_{\DF}^\circ)$ can be derived as
\begin{align*}
    \MISE\left \{ \hat{f}_{p,\tight}(\theta, K_{\DF}^\circ) \right \} 
    &= \int_0^{0.5} \left[{\Bias}_0^2\left \{ \hat{f}_{p,\tight}(\theta, K_{\DF}^\circ) \right \} 
        + {\Var}_0\left \{ \hat{f}_{p,\tight}(\theta, K_{\DF}^\circ) \right \} \right]\dd \theta \\
    &\sim 0^{2p} \Xi_{00} + \int_0^{0.5} \left\{ 
         4 \sum_{i=1}^b i^{2p} \Xi_{ii} \cos^2(2\pi\theta i) 
        + 4 \delta_1^2 \sum_{i=b+1}^{(2m+1)b} i^{2p} \Xi_{ii} \cos^2(2\pi\theta i)
        \right.\\
    &\qquad +  4 \sum_{j=1}^b 0^p j^p \Xi_{0j} \cos(2\pi\theta j)
        + 4 \delta_1 \sum_{j=b+1}^{(2m+1)b} 0^p j^p \Xi_{0j} \cos(2\pi\theta j)  \\
     &\qquad + 8\sum_{i=1}^{b-1} \sum_{j=i+1}^{b} (ij)^p \Xi_{ij} \cos(2\pi\theta i) \cos(2\pi\theta j) \\
     &\qquad + 8\delta_1 \sum_{i=1}^b \sum_{i=b+1}^{(2m+1)b}  (ij)^p \Xi_{ij} \cos(2\pi\theta i) \cos(2\pi\theta j) \\
     &\qquad \left.+ 8\delta^2 \sum_{i=b+1}^{(2m+1)b-1} \sum_{j=i+1}^{(2m+1)b}  (ij)^p \Xi_{ij} \cos(2\pi\theta i) \cos(2\pi\theta j) \right\} \dd\theta.
\end{align*}
Then, notice that
\begin{align*} 
    \mathcal{H}_{1,f}^{(p)}
        &= \sum_{i=b+1}^{(2m+1)b} i^{2p}\Xi_{ii} 
            \int_0^{0.5}\cos^2(2\pi\theta i) \dd \theta 
        = \frac{1}{4} \sum_{i=b+1}^{(2m+1)b} i^{2p} \Xi_{ii} \\
        &\sim \frac{1}{4n} \sum_{i=b+1}^{2b}  i^{2p}         
            \sum_{k=b-i}^{i-b} \gamma_k^2 +  
            \frac{1}{4n} \sum_{i=2b+1}^{(2m+1)b} i^{2p}
            \sum_{k=-b}^{b}\gamma_k^2, \\
    \mathcal{H}_{2,f}^{(p)} &=  \sum_{j=b+1}^{(2m+1)b} 0^p j^p \Xi_{0j} \int_0^{0.5}\cos(2\pi\theta j) \dd \theta =0, \\
    \mathcal{H}_{3,f}^{(p)} &= 
           \sum_{i=1}^b \sum_{j=b+1}^{i+2b} (ij)^p \Xi_{ij}
           \int_0^{0.5}\cos(2\pi\theta i) \cos(2\pi\theta j) \dd \theta =0,  \\
    \mathcal{H}_{4,f}^{(p)} &= 
          \sum_{i=b+1}^{(2m+1)b-1} \sum_{j=i+1}^{\min\{i+2b,\ (2m+1)b\}} (ij)^p \Xi_{ij} 
          \int_0^{0.5}\cos(2\pi\theta i) \cos(2\pi\theta j) \dd \theta =0.
\end{align*}
And we have
\begin{align*} 
\MISE\left \{ \hat{f}_{p,\tight}(\theta, K_{\DF}^\circ) \right \} 
    &\sim 0^{2p} \Xi_{00} + \int_0^{0.5} \left\{ 4 \sum_{i=1}^b i^{2p}\Xi_{ii} \cos^2(2\pi\theta i)
        \right\} \dd\theta \\
    &\quad +\int_0^{0.5} \left\{ 4 \sum_{j=1}^b 0^p j^p \Xi_{0j} \cos(2\pi\theta j) \right\} \dd\theta \\
    &\quad + \int_0^{0.5} \left\{ 8 \sum_{i=1}^{b-1} \sum_{j=i+1}^{b} (ij)^p \Xi_{ij} \cos(2\pi\theta i) \cos(2\pi\theta j)  \right\} \dd\theta \\
    &\quad + 4 \mathcal{H}_{1,f}^{(p)} \delta_1^2.
\end{align*}
Therefore, the $\textsc{mise}$-optimal difference sequence of $\hat{f}_{p,\tight}(\cdot, K_{\DF}^\circ)$ is simply
\begin{align*}
d_{f_p,0:m}^{*\textsc{vani}} = 
    \argmin_{d_{0:m}\in\mathcal{D}_m} \left \{\mathcal{H}_{1,f}^{(p)} \delta_1^2 \right\}
    =\argmin_{d_{0:m}\in\mathcal{D}_m} \delta_1^2.
\end{align*}

  \bibliographystyle{rss}
  {\small
  \setstretch{1.0}
  \bibliography{myRef}
  }

\vspace{-2cm}
\end{document}